\newif\ifCONF

\CONFfalse

\ifCONF
\documentclass[sigconf]{acmart}
\else
\documentclass[acmlarge,nonacm,balance=false]{acmart}
\fi

\ifCONF

\copyrightyear{2024}
\acmYear{2024}
\setcopyright{acmlicensed}
\acmConference[CODASPY '24] {Proceedings of the Fourteenth ACM Conference on Data and Application Security and Privacy}{June 19--21, 2024}{Porto, Portugal.}
\acmBooktitle{Proceedings of the Fourteenth ACM Conference on Data and Application Security and Privacy (CODASPY '24), June 19--21, 2024, Porto, Portugal}
\acmISBN{979-8-4007-0421-5/24/06}
\acmDOI{10.1145/3626232.3653280}

\else 
\renewcommand\footnotetextcopyrightpermission[1]{} 
\fi

\usepackage{enumitem}  
\usepackage{calc}

\usepackage{graphicx,color,amsmath,amsthm,bm,multicol}
\usepackage{subcaption}

\graphicspath{{figs/}{figs/single_computation/}{figs/single_mixed/}{figs/two_computations/}{figs/three_comp/}{figs/diagrams/}{figs/3d_plots/}}
\newcommand{\awae}{\text{awae}}
\newcommand{\twae}{\text{twae}}

\newcommand{\nn}{\nonumber}

\newcommand{\X}{\mathcal{X}}
\newcommand{\Y}{\mathcal{Y}}

\newcommand{\ID}{\textsf{ID}}
\newcommand{\id}{\textsf{id}}

\newcommand{\GG}{\textsf{G}}

\newcommand{\grp}{\mathbb{G}}

\newcommand{\of}{\subseteq}
\renewcommand{\P}{\mathcal{P}}

\newcommand{\unif}{\mathcal{U}}

\newcommand{\N}{\mathcal{N}}
\newcommand{\R}{\mathbb{R}}
\newcommand{\tran}{\text{T}}
\newcommand{\pois}{\text{Pois}}
\renewcommand{\S}{\mathcal{S}}
\DeclareMathSymbol{\mdot}{\mathord}{symbols}{"01}

\newcommand{\norm}[1]{\left\lvert#1\right\rvert}

\newcommand{\lp}[1]{\left(#1\right)}
\newcommand{\lpp}[1]{(#1)}
\newcommand{\lb}[1]{\left[#1\right]}
\newcommand{\lbr}[1]{\left\{#1\right\}}

\newcommand{\ignore}[1]{}
\renewcommand{\Pr}{\textup{Pr}}

\newcommand{\summ}[2]{\sum_{#1}^{#2}}

\newcommand{\Ex}[1]{\text{E}\left[#1\right]}
\newcommand{\Cov}[1]{\text{Cov}\left[#1\right]}
\newcommand{\Var}[1]{\text{Var}\left[#1\right]}
\newcommand{\floor}[1]{\left\lfloor #1 \right\rfloor}

\newcommand{\myquad}[1][1]{\hspace*{#1em}\ignorespaces}

\newcommand{\breakcmd}{\\ &}
\newcommand{\eqanchor}{&}

\newtheorem{definition}{Definition}
\newtheorem{claim}{Claim}
\newtheorem{conjecture}{Conjecture}

\newcounter{casenum}
\newenvironment{caseof}{\setcounter{casenum}{1}}{\vskip.5\baselineskip}
\newcommand{\case}[2]{\vskip.5\baselineskip\par\noindent {{\itshape Case }\arabic{casenum}:} #1#2\addtocounter{casenum}{1}}

\begin{document}

\title{Understanding Information Disclosure from Secure Computation Output: A Study of Average Salary Computation}

\author{Alessandro Baccarini}
\orcid{0000-0002-2028-5440}
\affiliation{%
\institution{University at Buffalo}
\city{Buffalo}
\state{New York}
\country{USA}}
\email{anbaccar@buffalo.edu}

\author{Marina Blanton}
\orcid{0009-0008-9934-2700}
\affiliation{%
\institution{University at Buffalo}
\city{Buffalo}
\state{New York}
\country{USA}}
\email{mblanton@buffalo.edu}

\author{Shaofeng Zou}
\orcid{0000-0002-2821-6941}
\affiliation{%
\institution{University at Buffalo}
\city{Buffalo}
\state{New York}
\country{USA}}
\email{szou3@buffalo.edu}

\begin{abstract}
  Secure multi-party computation has seen substantial performance improvements in recent years and is being increasingly used in commercial products. While a significant amount of work was dedicated to improving its efficiency under standard security models, the threat models do not account for information leakage from the output of secure function evaluation. Quantifying information disclosure about private inputs from observing the function outcome is the subject of this work. Motivated by the City of Boston gender pay gap studies, in this work we focus on the computation of the average of salaries and quantify information disclosure about private inputs of one or more participants (the target) to an adversary via information-theoretic techniques. We study a number of distributions including log-normal, which is typically used for modeling salaries. We consequently evaluate information disclosure after repeated evaluation of the average function on overlapping inputs, as was done in the Boston gender pay study that ran multiple times, and provide recommendations for using the sum and average functions in secure computation applications. Our goal is to develop mechanisms that lower information disclosure about participants' inputs to a desired level and provide guidelines for setting up real-world secure evaluation of this function.
\end{abstract}

\ifCONF


\else \fi

\maketitle

\section{Introduction}

Secure multi-party computation and other forms of computing on cryptographically protected data (such as homomorphic encryption) open up possibilities for great utilization and analysis of private data distributed across different domains, which otherwise might not be feasible due to the sensitive nature of the data. For example, analysis of health-related records and medical images distributed across different medical facilities and extracting cues from them lead to medical advances 
without the need to see the records themselves. Today, data analysis practices by researchers are hindered by laws regulating access to health data in different countries and analyzing medical data at scale presents challenges. Similarly, analyzing sensitive data such as salaries to understand disparities by gender, race, or other types of marginalization can supply decision makers with important information and empower them to address the disparities. This was the case with the Boston area gender pay gap surveys~\cite{bwwc2016,bwwc2017,lap16,lap16b,lap18} that initiated in 2015 and ran through 2017 with more participants and data analysis by additional categories including race. More broadly, wider adoption of privacy-preserving technologies, and secure computation in particular, can lead to higher security standards and practices for a broad range of different aspects of our society.

Secure computation techniques have seen significant advances in recent decades in terms of their speed, as well as availability of implementations and tools to facilitate their use for a variety of applications. Tech giants such as Google and Apple started using secure computation techniques in their products~\cite{google-ads,walker2019,ion2020,bhowmick2021apple} and the number of start-up companies offering related products is growing (see, e.g., \cite{partisia,inpher,ligero,nthparty}). However, a number of fundamental questions still need to be addressed by the research community in order to make secure computing practices common place. 

One of the fundamental questions is how much information about a participant's private input(s) might be available as a result of evaluating a desired function on private inputs. Standard security definitions adopted in the cryptographic community require that no information about private inputs is disclosed during function evaluation. That is, given a function $f$ that we evaluate on private inputs $x_1, x_2, \ldots$ coming from different sources, security is achieved if a participant does not learn more information than the function output and any information that can be deduced from the output and its private input. However, there are no constraints on types of functions that can be evaluated in this framework, and thus the information a participant can deduce from the output and its private input about another participant's private input is potentially large. This problem is typically handled by assuming that the function being evaluated is agreed upon by and acceptable to the data owners as not to reveal too much information about private inputs. However, our ability to evaluate functions in this aspect and determine what functions might be acceptable is currently limited. This question is the subject of this work.

Intuitively, what we want is to guarantee that the function being evaluated on private data is non-invertible, i.e., observing the output does not reveal its private input. Cryptography uses the notion of one-way functions -- and assumes this property for hash functions -- to model non-invertibility. However, what is needed in this case is to ensure that the possible space for a target private input is still large after the adversary observes the result of function evaluation. This notion of non-invertibility was first used in the context of secure multi-party computation in solutions for business applications such as supply chain management and component procurement~\cite{deshpande2005secure_tr,deshpande2005secure,deshpande2011} and was formulated as the inability to narrow down the (private) input of another party to a single value or a small set of possible values. 
Consequently, a series of publications by Ah-Fat and Huth~\cite{ah2017secure,ah2019optimal,ah2020protecting,ah2020two} put forward formal definitions that use entropy to measure the amount of uncertainty about one or more participants' private inputs after using them in secure multi-party computation. The definitions are framed from a) an attacker's perspective who aims to maximize information disclosure of a target's private input and b) from a target's perspective who determines the maximum information disclosure about their inputs when deciding whether to contribute their inputs to secure evaluation of a particular function. The above formulations are general and applicable to any function, while application-specific formulations of what constitutes sufficient input protection and function non-invertibility also emerged. One example is building machine learning models resilient to membership inference attacks~\cite{shokri2017membership,song2021systematic} that guarantee that it is infeasible to determine whether someone's data was used for training the model.

\ifCONF \else \medskip \noindent \fi
\textbf{Our contributions.}
In this work, we use the entropy-based definitions from~\cite{ah2017secure} as our starting point and analyze a specific function of significant practical relevance. In particular, we focus on the case of average salary computation as used in the Boston gender pay gap study~\cite{lap16}. When the total number of inputs is known (which is typically the case), the average computation is equivalent to computing the sum. We intuitively understand that the larger the number of inputs used in the computation of the average is, the better protection each individual contributing its input obtains. In the extreme case of two participants\footnote{We use the term ``participants'' to denote parties contributing inputs to the computation. The computation itself can be performed by a different set of parties, but our result is independent of the mechanism used to realize secure function evaluation.} no protection can be achieved. This was understood by the designers of the Boston gender pay gap study who recommended running the computation with at least 5 contributors~\cite{lap16b}. However, the information disclosure was not quantified, which we remedy in this work.  

We start by analyzing the function itself and formally show that the amount of information an attacker learns is independent of his/her own inputs. This is consistent with our intuition that, given a sum, one can always remove their contribution to the sum and analyze the resulting value. Thus, the protection depends on the number of spectators, i.e., input parties distinct from those controlled by the adversary and the party or parties being targeted. 

We analyze the target's input entropy remaining after participating in the computation (and consequently the entropy loss as a result of participation) for a number of discrete and continuous distributions including uniform, Poisson, normal (Gaussian), and log-normal. Log-normal is typically used for modeling salary data, but is the least trivial to analyze. An unexpected finding of our analysis is that for a given distribution, the absolute entropy loss is normally independent of the distribution parameters and the absolute entropy loss remains very close for different distributions as we vary the number of participants/spectators. Quantifying the information loss allows us to devise a mechanism to lower information disclosure to any desired level (e.g., 1\% of original entropy, 0.05 bits of entropy, etc.).

We extend our analysis of information loss to the case when the computation is run more than once (as was the case for the Boston gender pay gap study) and examine the case with two evaluations. This corresponds to (i) the target participating in two computations with the same input where the set of participants differs between the executions and (ii) the target participating in one computation, where the other is run without the target's input. We observe that information loss increases as a result of multiple computations, regardless of whether the target participates once or twice. Furthermore, the protection is maximized when one half of the original contributors are replaced, i.e., 50\% of the initial participants remain and the other 50\% are replaced with new participants. Our multi-execution analysis is based on the normal distribution, but we expect the outcome to be similar for other distributions as well.

\ifCONF
We provide additional proofs and generalize our analysis to three and more executions in the full version of the text~\cite{baccarini2022understanding}.
\else 
We proceed with generalizing our analysis to three and more, $M$, executions.
\fi
An interesting finding is that the best configuration that minimizes information loss is determined by pairwise overlaps of participants between the executions and not by other parameters and sizes. This allows us to determine optimal setup for a single and repeated execution of the average function.

We empirically validate our findings throughout this work and provide recommendations for securely evaluating the average function in real world applications. In particular, in all of our experiments the cost of participating in the average computation, i.e., the difference in the entropy before and after the computation is a fraction of a bit (for both Shannon entropy used with discrete distributions and differential entropy used with continuous distributions). This translates to small relative entropy loss in practice. When modeling salary data using log-normal distribution with the parameters specifically chosen for salaries~\cite{cao2022priori}, the entropy loss is below 5\% with at least 5 non-adversarial participants or spectators and achieving 1\% entropy loss requires 24 spectators. These numbers are also surprisingly similar across different distributions. Furthermore, when the computation is repeated (we use a normal distribution to adequately approximate the log-normal setup), engaging in the computation the second time with an overlapping set of 50\% participants whose inputs do not change results in only 30\% entropy loss of the first participation. These and other findings lead to a number of recommendations for running this computation in practice, which we provide at the end of this work.

\ifCONF \else \medskip \noindent \fi
\textbf{On the choice of metric.} Our analysis uses Shannon entropy. One might argue that this is not the best metric because it does not distinguish between, e.g, leaking the least significant vs. most significant bit of one's salary, while learning the latter is much more valuable to an adversary than learning the former. 
However, as we show throughout this work, information leakage for this application is always small regardless of the setup. In particular, the most favorable for the adversary setup across all distributions discloses only about 0.7 bits of entropy, i.e., the adversary cannot learn even a single bit of target's salary.
Furthermore, we derive effective mechanisms for reducing information loss to a controlled low level such that the worst case scenario will not realize.
We conduct similar analysis using min-entropy 
\ifCONF
in the full version of the paper~\cite{baccarini2022understanding}
\else 
(see Section~\ref{sub:min_entropy_analysis}) 
\fi 
and show that Shannon entropy trends are consistent with those for min-entropy. A primary advantage of using Shannon entropy is that we are able to go much further in our analysis and ultimately derive close-form expressions, which cannot be accomplished for other metrics.

\section{Related Work}
\label{related_works}

In what follows, we review prior literature on information disclosure from function output in the context of computing on private data and related techniques that limit information disclosure.

\subsection{Quantitative Information Flow}
\label{sub:quantitative_information_flow}

The field of \emph{quantitative information flow} is closely related to our work.
Denning~\cite{denning1982cryptography} is credited as the first to quantify information flow as a measure of the interference between variables at two stages during a program's execution (typically denoted by ``high-'' and ``low-security'' variables, which equates to the target's inputs and output in our setting, respectively).
Smith~\cite{smith2009foundations} formally established the foundations of quantifying the information leakage under the threat model that an attacker can recover a secret in one attempt (denoted by the notion of \emph{vulnerability}). It has been shown by Massey~\cite{massey1994guessing} that the Shannon entropy cannot capture this information under the guessing assumption, and Smith recommends min-entropy in its place. 
Alvim~et~al.~\cite{m2012measuring} generalized the min-entropy into the  $g$-leakage to incorporate gain functions to model the \emph{benefit} an adversary gains from making guesses about the secret. 
Subsequent works encompassed variations on the $g$-leakage~\cite{alvim2014additive}.
Other works in differential privacy 
feature derivations of leakage bounds~\cite{clark2002quantitative}, leakage analysis in the case of an adaptive adversary~\cite{kopf2011automatically}, and knowledge-based approaches for measuring risk~\cite{mardziel2012knowledge,rastogi2013knowledge}.

The fundamental advantage of our Shannon-based approach is the ability to derive closed-form expressions for the information leakage of the average salary computation, while other metrics do not share this characteristic. For example, the chain rule of entropy (a simple, yet critical component of our analysis) is not satisfied if min-entropy is used \cite{iwamoto2013information,skorski2019strong} in place of Shannon entropy. Our reductions would no longer hold, and we would be forced to resort to complete enumeration or approximation methods to compute the entropy. 
However, in 
\ifCONF 
the full version 
\else
Section~\ref{sub:min_entropy_analysis}
\fi
we provide supplementary analysis that demonstrates similarities between Shannon entropy and min-entropy based analyses. We also remain open to evaluating other metrics in the future.

An additional distinction between our work and existing literature on (quantitative) information flow is that we do not consider possible leakage from intermediate aspects of a computation's execution. Whereas other works may examine a program's loops~\cite{mardziel2012knowledge}, side-channel vectors~\cite{kopf2011automatically}, or inter-dependent structures~\cite{alvim2014not}, we strictly investigate the relationship between the output and target's input, since function itself is assumed to be evaluated using secure multi-party protocols.

\subsection{Function Information Disclosure} 
\label{sub:Information_Loss}

Existing literature on information leakage from the output of a secure function evaluation is limited, relative to the rest of the field of secure computation. Secure multi-party protocols are designed to guarantee no information is disclosed throughout a computation, but do not ensure input protection after the output is revealed. The work of Deshpande~et~al.~\cite{deshpande2005secure,deshpande2005secure_tr,deshpande2011} was pioneering in that respect and designed secure multi-party protocols for business applications that ensured that the function being evaluated is \emph{non-invertible}, i.e., no participant can infer other participants' inputs from the output. A trivially invertible example is the average salary calculation between two individuals, since either party can recover the other's input exactly. Deshpande~et~al.~\cite{deshpande2005secure,deshpande2005secure_tr} first addressed non-invertibility in the context of secure supply chain processes. The proposed protocols offered protection from inference of future inputs to a repeated calculation after a result is disclosed.  A later work by Deshpande et~al.~\cite{deshpande2011} achieved non-invertibility for a framework designed for secure price masking for outsourcing manufacturing.
The authors argued information leakage was minimal by analyzing mutual information between correlated normal random variables, but did not consider other distributions or entropy metrics.

Ah-Fat~and~Huth~\cite{ah2017secure} provided the first in-depth analysis of information leakage from the outputs of secure multi-party computations.
The authors formalized two metrics to measure expected information flow from the attacker's and target's perspectives, namely, the \emph{attacker's weighted average entropy} (awae) and  \emph{target's weighted average entropy} (twae), respectively. Participants' inputs are modeled using probability distributions and
were specified to be uniform, but this constraint can be relaxed. 
The inherent difficulty of this entropy-based approach is the requirement to enumerate every possible input combinations from all parties, which scales poorly as the input space and number of participants grow. We utilize their definitions for our analysis and demonstrate their utility to computation designers to determine potential disclosure about participants' inputs

This model was expanded in~\cite{ah2019optimal} to encompass the R\'enyi, min-, and $g$-entropy. The extension is presented in combination with a technique for distorting secure computation outputs to limit information disclosure from the output and achieve balance between accuracy and privacy.
This was further developed in~\cite{ah2020protecting} with a fuzzing method based on randomized approximations. A closed-form expression for the min-entropy of a two- and three-party auction was derived in~\cite{ah2020two}, alongside a conjecture for the case with an arbitrary number of parties.

Conceptually, the notion of \emph{output privacy} is related to our work. 
The terminology was introduced in the field of data mining~\cite{bu2006preservation,wang2011output,kotecha2017preserving,mendes2017privacy,monreale2016privacy}, with the goal of 
designing techniques to protect inputs from inference attacks on the output model.
Information about the inputs that can be obtained from the output includes, but is not limited to, properties which can be uniquely attributed to a small number of input participants. Conventional approaches for minimizing disclosure involve applying transformations on the result via monotonic functions~\cite{bu2006preservation} or even proactive learning~\cite{wang2011output}. These techniques have little to no impact on the result of the computation. This direction differs from our work since the type of disclosure they aim to rectify is not quantified.

There is also literature that uses specific formulations to demonstrate that computation does not disclose sensitive information about participants. This includes resilience to \emph{membership inference attacks}~\cite{shokri2017membership,song2021systematic} in the context of machine learning training and \emph{differential privacy}~\cite{dwork2008differential,dwork2014algorithmic} for statistical databases. In particular, differential privacy ensures the output of a query is negligibly dependent on a single individual's record in the database and resilience to membership inference attacks prevents one from determining whether a specific individual's data was used for model training. These concepts have no direct relationship to our work, aside from designing mechanisms for lower information disclosure as a result of computation on private data. 
In this work, we do so by varying the number of participants in the computation, while other methods augment the function directly to produce a differentially private output.

\section{Preliminaries}
\label{sec:definitions}

Following \cite{ah2017secure}'s notation, let $P$ denote the set of all participants in a computation with $\norm{P} = m$. All participants $P$ are partitioned into three groups: parties controlled by an attacker $A \subset P$, a group of parties being targeted $T \of P \setminus A$, and the remaining participants called spectators $S = P \setminus (A \cup T)$. 
Let the random variable $X_{P_i}$ correspond to the input of a single participant $P_i$ and $x_{P_i}$ denotes a value that $X_{P_i}$ takes. In addition, the notation $\vec{X}_P = \lp{X_{P_1}, \dots, X_{P_m}}$ denotes a multidimensional random variable and $\vec{x}_P$ is a vector of the individual values of the same size. 
We also let $X_P = \sum_{i}X_{P_i}$ define a new random variable representing the sum of the participants' random variables. The same notation applies to the sets $A$, $T$, and $S$. 
Our present analysis is based upon the assumption that all participants' inputs are independent and identically distributed, which we consequently relax.

For discrete distributions, we use Shannon entropy $H(X)$ to measure the information of a discrete random variable $X$ with mass function $\Pr(X = x)$ defined over domain $D_X$. Specifically, 
\begin{align*}
  H(X) = - \sum\nolimits_{x\in D_X} \Pr(X = x) \cdot \log \Pr(X = x),
\end{align*}
where all logarithms are to the base 2. If we are dealing with continuous distributions, we shift to the differential entropy $h(X)$ 
with density function $f(x)$ over the support set $\X$, defined as 
\begin{align*}
    h(X) = - \int_\X f(x) \log f(x) dx.
\end{align*}
We study information leakage of the computation of the average:
\begin{align*}
  o = f(\vec{x}_A,\vec{x}_T,\vec{x}_S) = \frac{1}{m} \lp{\sum\nolimits_{i} x_{T_i} + \sum\nolimits_{j}x_{A_j} + \sum\nolimits_{k} x_{S_k}},
\end{align*}
where $o$ denotes the output of the function. We model the output $o$ by the random variable $O$ defined over the domain $D_O$, namely
\begin{align*}
  O = \frac{1}{m} \lp{\sum\nolimits_{i} X_{T_i} + \sum\nolimits_{j}X_{A_j} + \sum\nolimits_{k} X_{S_k}}.
\end{align*}
The $1/m$ factor can be ignored in the final expression since the number of participants is typically known by all parties and can trivially be removed from the output. We omit it throughout the remainder of the paper. 

In this work, we consider distributions where the sum of independent individual random variables is well studied and their mass or density functions have closed-forms expressions or can be reasonably approximated. This includes the following distributions:
\begin{itemize}
  \item \emph{Discrete uniform} $\unif \lp{a,b}$, where $a$ and $b$ are integers corresponding to the minimum and maximum of the range of the support set $\lbr{a, a+1, \dots, b-1, b}$. 
  \item \emph{Poisson}  $\pois \lp{\lambda}$, where  $\lambda \in \R_{>0} $ is the shape parameter that indicates the expected (average) rate of an event occurring over a given interval.
  \item \emph{Normal (Gaussian)}  $\N \lp{\mu, \sigma^2}$, where $\mu \in \R$ and $\sigma^2 \in \R_{>0}$ correspond to the mean and squared standard deviation, respectively.
  \item \emph{Log-normal} $\log\N \lp{\mu, \sigma^2}$ with parameters $\mu \in \R$ and $\sigma^2 \in \R_{>0}$, which correspond to the mean and squared standard deviation of the random variable's natural logarithm.
\end{itemize}
$X \sim {\sf Dist}$ indicates that random variable $X$ has distribution $\sf Dist$.

As stated earlier, Ah-Fat and Huth~\cite{ah2017secure} provided multiple information-theoretic measures to quantify information disclosure after a function evaluation, which we use here:
\begin{definition}[\cite{ah2017secure}]
  The joint weighted average entropy (\textup{jwae}) of a target $T$ attacked by parties $A$ is defined over all $\vec{x}_A \in D_A$ and $\vec{x}_T \in D_T$ as
  
\begin{align*}
 \textup{jwae}(\vec{x}_A,\vec{x}_T)
      &=\sum\nolimits_{o\in D_O}\Pr(O=o\mid \vec{X}_A=\vec{x}_A,\vec{X}_T=\vec{x}_T)\ifCONF \breakcmd \myquad[4]\else\fi \cdot H(\vec{X}_T\mid \vec{X}_A=\vec{x}_A,O=o).
\end{align*}

\end{definition}
\noindent This metric measures the information an attacker would learn (on average) about the target when the input vectors are \(\vec{x}_A\) and \(\vec{x}_T\). One can subsequently define the average of the jwae over all possible \(\vec{x}_T\) or \(\vec{x}_A\) vectors weighted by their respective prior probabilities.
\begin{definition}[\cite{ah2017secure}]
  \label{def:twae}
  The target's weighted average entropy (\textup{twae}) of a target $T$ attacked by parties $A$ is  defined for all  $\vec{x}_T \in D_T$ as
  \begin{equation*}
    \textup{twae}(\vec{x}_T)=\sum\nolimits_{\vec{x}_A\in D_A}\Pr(\vec{X}_A=\vec{x}_A) \cdot \textup{jwae}(\vec{x}_A,\vec{x}_T).
  \end{equation*}
\end{definition}
\noindent The twae informs a target how much information an attacker can learn about its input when the input is $\vec{x}_A$.
\begin{definition}[\cite{ah2017secure}]
  \label{def:awae}
  The attacker's weighted average entropy (\textup{awae}) of a target $T$ attacked by parties $A$ is  defined for all  $\vec{x}_A \in D_A$ as
  \begin{equation*}
    \textup{awae}(\vec{x}_A)=\sum\nolimits_{\vec{x}_T\in D_T}\Pr(\vec{X}_T=\vec{x}_T) \cdot \textup{jwae}(\vec{x}_A,\vec{x}_T).
  \end{equation*}
\end{definition}
\noindent The awae informs an attacker about how much information it can learn about the target's input when the attacker's input vector is \(\vec{x}_A\). The attacker can consequently compute the awae on all values in $D_A$ to determine which input maximizes the information learned about the target's input (and thus what should be entered into the computation). Using the definition of jwae, it follows that: 
\begin{align*}
	\awae(\vec{x}_A)&=\sum_{\vec{x}_T\in D_T}\Pr(\vec{X}_T=\vec{x}_T)\sum_{o\in D_O}\Pr(O=o{\mid} \vec{X}_A=\vec{x}_A,\vec{X}_T=\vec{x}_T) \ifCONF \breakcmd\myquad[12] \else\fi \cdot H(\vec{X}_T\mid \vec{X}_A=\vec{x}_A,O=o)\nn\\
	&=\sum\nolimits_{\vec{x}_T\in D_T}\sum\nolimits_{o\in D_O}\Pr(O=o,\vec{X}_T=\vec{x}_T\mid \vec{X}_A=\vec{x}_A)\ifCONF \breakcmd\myquad[6]  \else\fi  \cdot H(\vec{X}_T\mid \vec{X}_A=\vec{x}_A,O=o).
\end{align*}
Since $\vec{X}_T$ is independent of $\vec{X}_A$, we derive that awae equals to conditional entropy:
\begin{align*}
	\awae(\vec{x}_A) &=\sum_{o\in D_O}\Pr(O=o\mid \vec{X}_A=\vec{x}_A)\cdot H(\vec{X}_T\mid \vec{X}_A=\vec{x}_A,O=o)
  \ifCONF \breakcmd \else\fi
 = H(\vec{X}_T\mid \vec{X}_A=\vec{x}_A,O)
\end{align*}
where the last equality is due to the definition of conditional \mbox{entropy.}

\section{Single Execution} 
\label{sec:single_calculation}

In this section we analyze a single execution of the average function on private inputs and the associated information disclosure of the target's inputs. 
Recall that the computation 
is modeled by 
\begin{equation}
    O=f(\vec{X}_A,\vec{X}_T,\vec{X}_S) = \sum\nolimits_{i} X_{T_i} + \sum\nolimits_{j}X_{A_j} + \sum\nolimits_{k} X_{S_k},\end{equation}
and we let $n = \norm{S}$ denote the number of spectators.

\begin{figure}[t]
      \centering
    \ifCONF
      \includegraphics[width=0.33\textwidth]{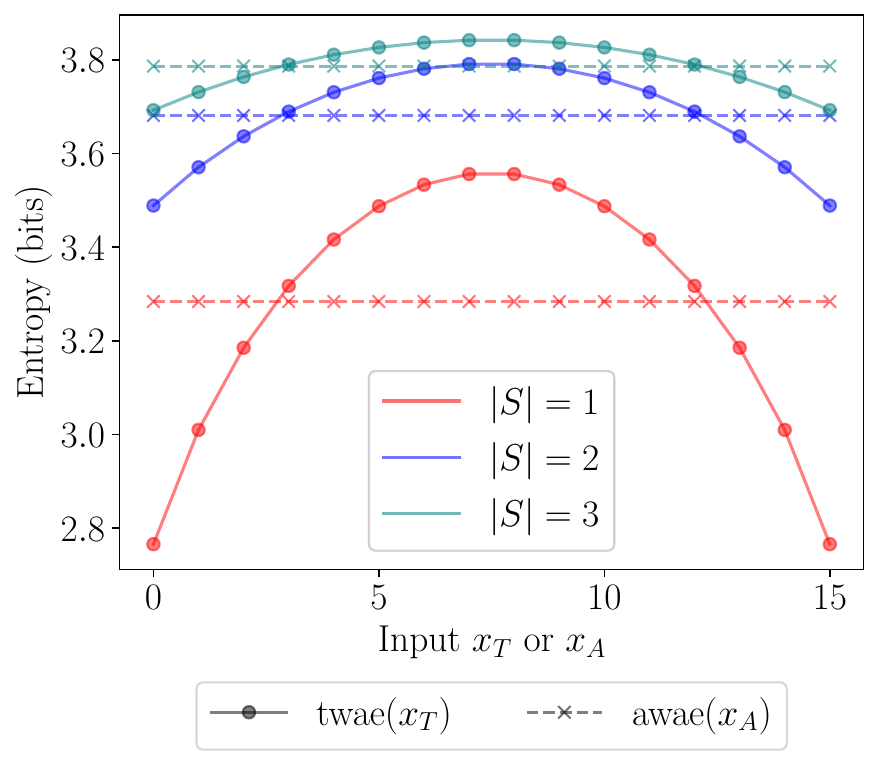}
      \else
      \includegraphics[width=0.4\textwidth]{twae_vs_awae.pdf}
     \fi 
    \caption{The $\twae(\vec{x}_T)$ and $\awae(\vec{x}_A)$ using inputs over $\unif \lp{0,15}$ with a different number of spectators $\norm{S}$.}
    \label{fig:twae_vs_awae}
\end{figure}

As a first step, we plot the values of awae and twae for our function of interest. Figure~\ref{fig:twae_vs_awae} illustrates these values with a single adversarial participant, a single target and a varying number of spectators (1--3). All inputs follow the uniform distribution $\unif \lp{0,15}$. Calculating the twae and awae values using Definitions~\ref{def:twae}~and~\ref{def:awae} requires enumerating all input and output combinations. This quickly becomes computationally inefficient as the input space grows.
  
Each participant, acting as a target, can utilize the twae prior to the computation to determine how much information an attacker can learn (on average) from the output for a specific input that the participant enters into the computation. 
As the figure illustrates, the target's remaining average entropy is maximized when the input is in the middle of the range, indicating that those values have better protection than inputs near the extrema. As the number of spectators increases, the curves shift upwards, i.e., the uncertainty about the target's input increases and the gap in the uncertainty between different input values reduces.

The awae, on the other hand, gives an adversary the ability to determine which input to enter into the computation that leads to maximum information disclosure about a target's input (without knowing what input the target used). As displayed in the figure, the adversarial knowledge does not change by varying its inputs into the computation. This is consistent with our intuition that, given the output, the adversary can remove their contribution to the computation and possess information about the sum of the inputs of the remaining parties. We formalize this as the following 
result:
\begin{claim} \label{claim1}
    $\textup{awae}(\vec{x_A})$  is independent of attacker's input vector $\vec{x_A}$.
\end{claim}
\newcommand\proofone{
    According to the chain rule of entropy which states that $H(X,Y) = H(X\mid Y) + H(Y)$~\cite[Chapter~2.5]{thomas2006elements}, we have that:
    \begin{align}
        H(\vec{X}_T \mid \vec{X}_A  \ifCONF  \else\fi{=}\vec{x}_A,O)  \ifCONF\nn \eqanchor \nn \else \eqanchor\fi
         = H(\vec{X}_T,O\mid \vec{X}_A=\vec{x}_A)-H(O\mid \vec{X}_A=\vec{x}_A)\nn\\
                         & =H(\vec{X}_T\mid \vec{X}_A=\vec{x}_A) + H(O\mid \vec{X}_T,\vec{X}_A=\vec{x}_A)\ifCONF\nn \breakcmd \nn\myquad[2] \else\fi- H(O\mid \vec{X}_A=\vec{x}_A)\nn           \\
                         & =H(\vec{X}_T) {+} H\left(\sum\nolimits_{i} X_{S_i}\right) {-} H\left(\sum\nolimits_{i} X_{T_i} {+} \sum\nolimits_{j} X_{S_j}\right),
                          \ifCONF \nn \else\fi
                         \label{eq:claim1_identical}
    \end{align}
    which is independent of $\vec{x_A}$.
    }
\begin{proof}
    \proofone
    \end{proof}
    \noindent
    Using our notation from Section~\ref{sec:definitions}, the above expression for $\textup{awae}(\vec{x_A})$  simplifies to
    \ifCONF \begin{equation*}
    \else \begin{equation} \label{eq:claim1_followup}
    \fi
        H(\vec{X}_T)+H\left(X_S\right) -H\left(X_T+ X_S\right)
        =H(\vec{X}_T \mid X_T + X_S).
    \ifCONF \end{equation*}
    \else \end{equation}
    \fi

The next step is to determine which measure (awae or twae) we should use in our analysis of the average salary computation. Ah-Fat and Huth~\cite{ah2017secure} argued that the awae served as a more precise metric for measuring information leakage of a secure function evaluation than twae for their choice of function and used awae in their subsequent work~\cite{ah2019optimal}. Our perspective also aligns with that conclusion. In particular, while the twae informs the target of the amount of information leakage for the input they possess, the target may not be technically savvy enough to be able to apply the metric and make an informed decision regarding computation participation (plus, the choice to participate or not participate can leak information about their input). Perhaps more importantly, a function needs to be analyzed by the computation designers in advance and without access to the inputs of future computation participants to determine a safe setup for the participants. Thus, the available mechanism for this purpose is the attacker's perspective or awae, and we focus on this metric in the rest of this work.

Based on the above, in what follows 
we use $H(\vec{X}_T \mid X_T + X_S)$ to measure the leakage, and the simplified function is
\begin{align*}
    f(\vec{X}_T,\vec{X}_S) =  \sum\nolimits_{i} X_{T_i} + \sum\nolimits_{j}X_{S_j}= X_T + X_S.
\end{align*}
This refines the parameters we can vary in our analysis to
(1) the number of participants in the target and spectators groups and 
(2) the types of distributions and statistical parameters of the inputs.
Furthermore, the computational difficulty associated with directly computing the awae is absent when using $H(\vec{X}_T \mid X_T + X_S)$. Instead, the computation simplifies to calculating the entropy of sums of random variables. 
We examine the behavior of the conditional entropy for several characteristic probability distributions next.

\subsection{Discrete Distributions} 
\label{sub:Discrete_Distributions}
\newcommand\discfig{
\begin{figure*}[t] \centering
    \begin{subfigure}[t]{0.32\textwidth} \centering
        \includegraphics[width=0.99\textwidth]{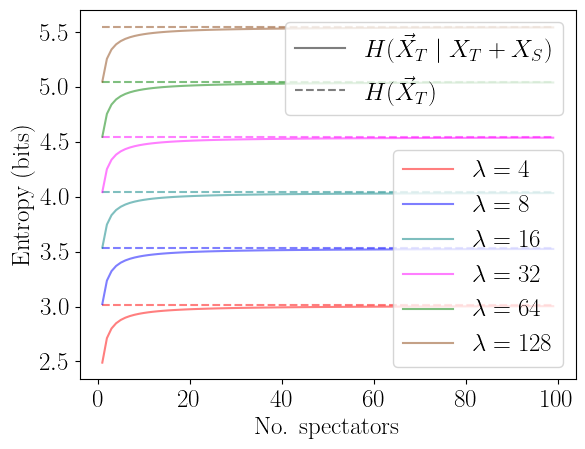}
        \caption{Target's entropy before $H(\vec{X}_T)$ and after $H(\vec{X}_T \mid X_T + X_S)$ the execution.}
        \label{subfig:sum_poisson_multiple_exps}
    \end{subfigure}
    \begin{subfigure}[t]{0.32\textwidth} \centering
        \includegraphics[width=0.99\textwidth]{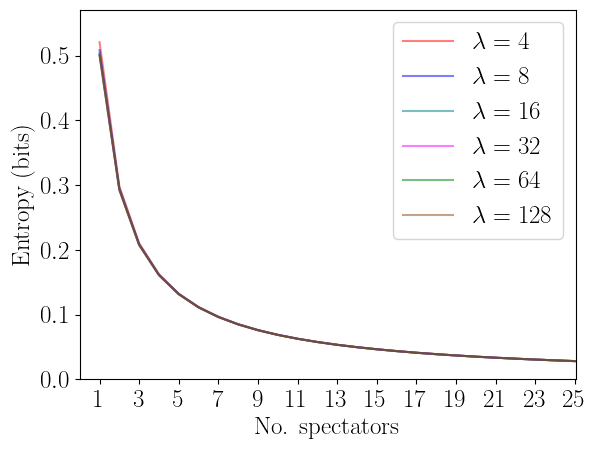}
        \caption{Target's absolute entropy loss $H(\vec{X}_T) - H(\vec{X}_T \mid X_T + X_S)$.}
        \label{subfig:sum_poisson_absolute_loss}
    \end{subfigure}
    \begin{subfigure}[t]{0.32\textwidth} \centering
        \includegraphics[width=0.99\textwidth]{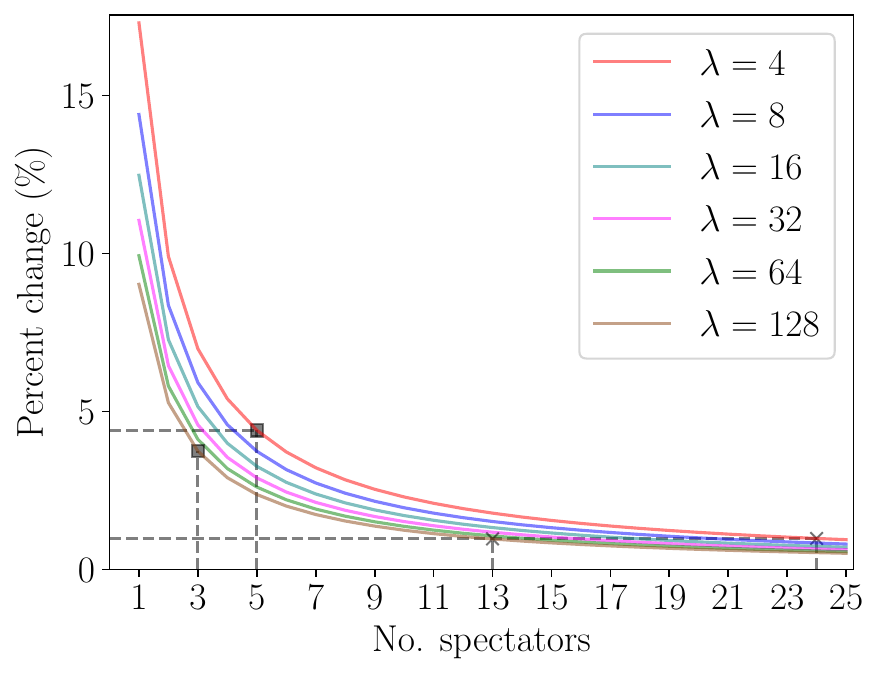}
        \caption{Target's relative entropy loss $\frac{H(\vec{X}_T) - H(\vec{X}_T \mid X_T + X_S)}{H(\vec{X}_T)}$.}
        \label{subfig:sum_poisson_relative_loss}
    \end{subfigure}
    \caption{Analysis of target's entropy loss using the Poisson distribution with $\pois(\lambda)$, and varying $\lambda$ with $|T| = 1$.}
    \label{fig:poisson_entropy}
\end{figure*}

\begin{figure*}[t]
    \centering
    \begin{subfigure}[t]{0.32\textwidth} \centering
        \includegraphics[width=0.99\textwidth]{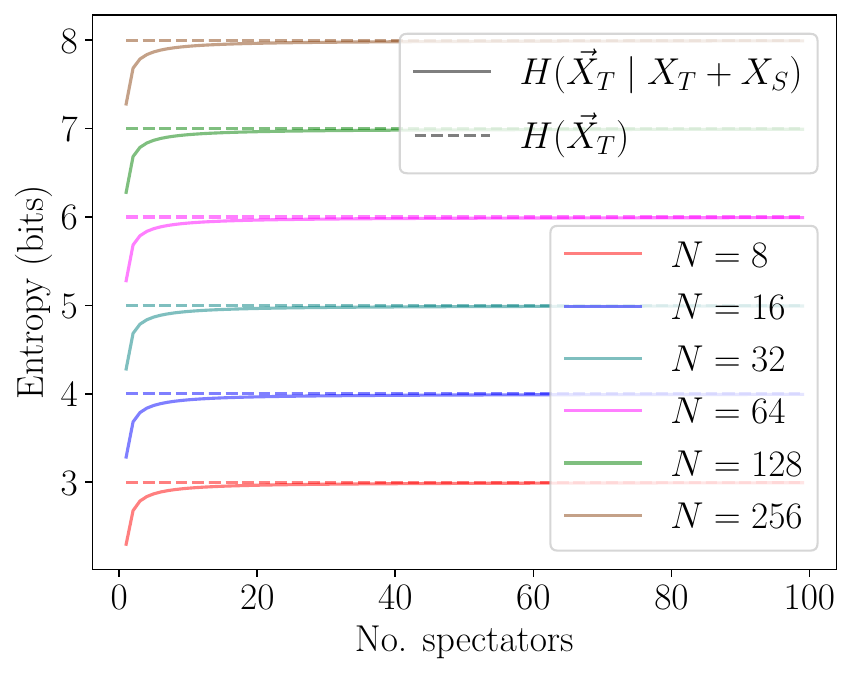}
        \caption{Target's entropy before $H(\vec{X}_T)$ and after $H(\vec{X}_T \mid X_T + X_S)$ the execution.}
        \label{subfig:sum_uniform_multiple_exps}
    \end{subfigure}
    \begin{subfigure}[t]{0.32\textwidth} \centering
        \includegraphics[width=1.02\textwidth]{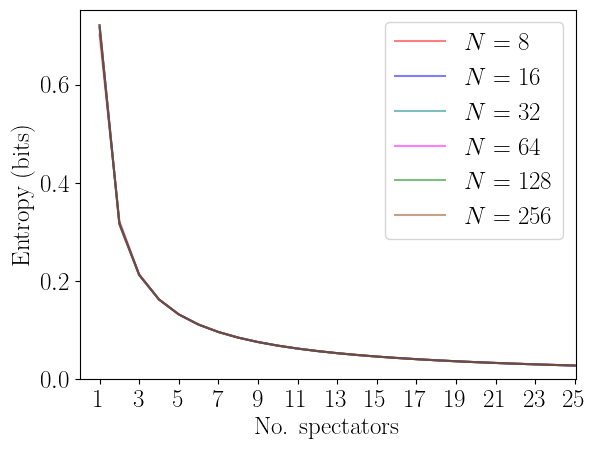}
        \caption{Target's absolute entropy loss $H(\vec{X}_T) - H(\vec{X}_T \mid X_T + X_S)$.}
        \label{subfig:sum_uniform_absolute_loss}
    \end{subfigure}
    \begin{subfigure}[t]{0.32\textwidth} \centering
        \includegraphics[width=0.99\textwidth]{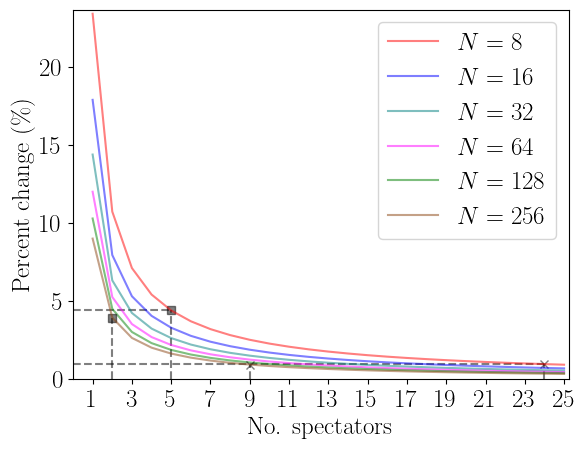}
        \caption{Target's relative entropy loss $\frac{H(\vec{X}_T) - H(\vec{X}_T \mid X_T + X_S)}{H(\vec{X}_T)}$.}
        \label{subfig:sum_uniform_relative_loss}
    \end{subfigure}
    \caption{Analysis of target's entropy loss using the uniform distribution with $\unif\lp{0,N-1}$, and varying $N$ with $|T| = 1$.}
    \label{fig:uniform_entropy}
\end{figure*}
}
\ifCONF 
\discfig 
\fi

We start with discrete input modeled using the uniform and Poisson distributions. The sum of $n$ identical independent Poisson random variables $X_i \sim \text{Pois}(\lambda)$ is equivalent to a single Poisson random variables $X  = \sum_i X_i \sim  \text{Pois}(n\lambda)$ with the mass function
\ifCONF
$ \Pr(X = x) = {\lp{n \lambda}^{k} e^{-n\lambda}}/\lp{k!}. $
\else
\begin{align*}
    \Pr(X = x) = \frac{\lp{n \lambda}^{k} e^{-n\lambda}}{k!}.
\end{align*}
\fi
Note that the Poisson distribution is defined over all non-negative integers, hence the distribution has infinite support. We choose to halt the calculation of $H(X)$ when $\Pr(X = x) < 10^{-7}$ as the contribution of events beyond this point to the entropy is infinitesimal.

Conversely, the sum of $n$ identical independent uniform random variables $X_i \sim \mathcal{U} \lp{0,N-1}$ is not immediately obvious. Caiado~and Rathie~\cite{caiado2007polynomial} derived several equivalent expressions for the mass function of the sum of $n$ uniform random variables, one of which we use in our analysis and is defined as:
\begin{align*}
    \Pr(X = x) = \frac{n}{N^{n}} \lp{ \sum\nolimits_{p = 0}^{ \floor{x / N}} \frac{\Gamma \lp{n + x + p N }\lp{-1}^{p} }{ \Gamma (p + 1 ) \Gamma (n - p + 1) \Gamma(x - pN + 1)}  },
\end{align*}
where $\Gamma(n) = (n-1)!$ is the Gamma function. The domain of $X$ is $\lbr{0,\dots, n(N-1)}$.

Our analysis of awae for these two distribution is given in Figures \ref{fig:poisson_entropy} and~\ref{fig:uniform_entropy}, respectively. We compute and display 
\begin{itemize}
\item the original entropy of target's inputs prior to the computation $H(\vec{X}_T)$ (subfigure a)
\item the awae or target's remaining entropy after the computation $H(\vec{X}_T \mid X_T + X_S)$ (subfigure a)
\item their difference of the two that represents the absolute entropy loss $H(\vec{X}_T) - H(\vec{X}_T \mid X_T + X_S)$ (subfigure b) and
\item the entropy loss relative to the original entropy prior to the execution $\lpp{H(\vec{X}_T) -H(\vec{X}_T \mid X_T + X_S)}/{H(\vec{X}_T)}$ (subfigure c)
\end{itemize} 
with a single target ($|T| = 1$), a varying number of spectators, and varying distribution parameters. Relative entropy loss is included to demonstrate to potential input contributors, who are likely non-experts, that information disclosure is small. That is, disclosure of, e.g., 5\% of input's information is easier to explain to non-experts than 0.1 bits of entropy. 
The absolute loss is equivalent to the mutual information between the target input and the output:
\ifCONF
$ I(\vec{X}_T; O) = H(\vec{X}_T) - H(\vec{X}_T \mid X_T + X_S).$
\else
\begin{align*}
    I(\vec{X}_T; O) 
    &= H(\vec{X}_T) - H(\vec{X}_T \mid X_T + X_S).
\end{align*}
\fi 
\ifCONF 
\else
\discfig 
\fi

Figure~\ref{fig:poisson_entropy} presents this information for the Poisson distribution with $\lambda \in \lbr{4,8,\dots, 128}$. In Figure~\ref{subfig:sum_poisson_multiple_exps}, entropy after the execution converges toward the corresponding entropy prior to the execution for all values of $\lambda$ as the number of spectators increases. Increasing $\lambda$ by a factor of two repeatedly yields an upward shift of these two curves by a constant amount while preserving their respective shapes. The increase is expected as a result of the inputs having more entropy as $\lambda$ increases, but the shape of the remaining entropy is notable, as $\lambda$ does not appear to impact the entropy loss.
This is further confirmed when displaying the absolute entropy loss 
in Figure~\ref{subfig:sum_poisson_absolute_loss}: The resultant curves overlap each other, regardless of $\lambda$.

The relative entropy loss in Figure~\ref{subfig:sum_poisson_relative_loss}, calculated as a percentage of the target's initial entropy, demonstrates how many spectators the computation needs to include to lower the entropy loss to the desired level. The larger the original entropy is (larger $\lambda$), the fewer spectators will be needed to stay within the desired percentage.
For example, 5 spectators are needed with $\lambda = 4$ to limit relative loss to 5\% (marked by $\blacksquare$) and 24 spectators are needed to cap the loss at 1\% (marked by $\times$). When $\lambda=128$, the number of spectators reduces to 3 and 13 to maintain loss tolerances of 5\% and 1\%, respectively.

The same trends hold for the uniform distribution in Figure~\ref{fig:uniform_entropy}, where we use $N \in \lbr{8, 16, \ldots, 256}$, but the values themselves slightly differ. For example, the absolute entropy loss in Figure~\ref{subfig:sum_uniform_absolute_loss} is slightly larger than the loss in Figure~\ref{subfig:sum_poisson_absolute_loss} when the number of spectators is small. When $N=8$ with 3 bits of original entropy, 5 and 24 spectators are needed to achieve at most 5\% and 1\% relative loss, respectively. This is the same as what was observed for Poisson distribution with 3-bit inputs ($\lambda=4$). 

\begin{figure*}[ht]
    \centering
    \begin{subfigure}[t]{0.32\textwidth} \centering
        \includegraphics[width=0.99\textwidth]{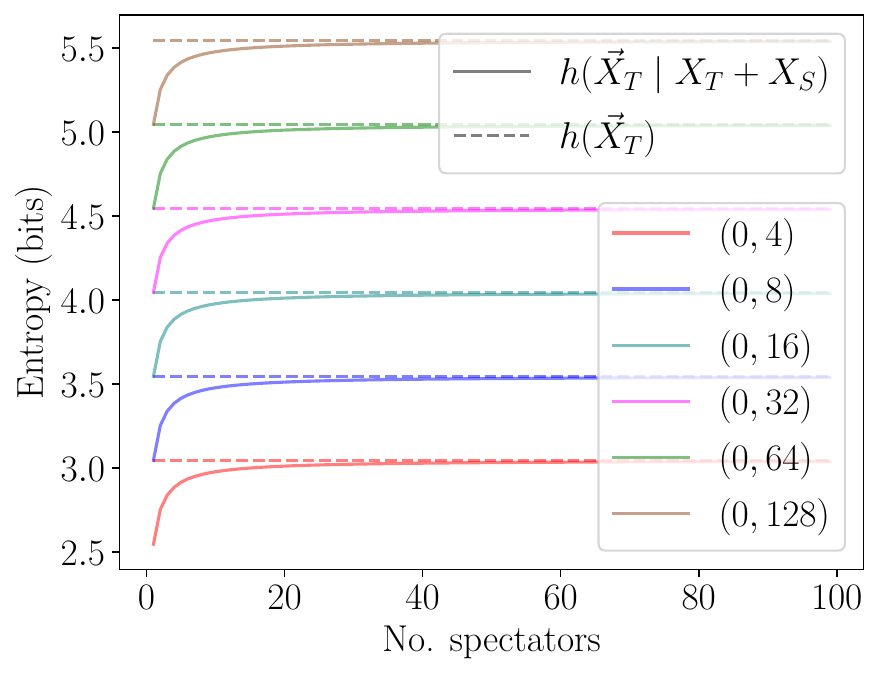}
        \caption{Target's entropy before $h(\vec{X}_T)$ and after $h(\vec{X}_T \mid X_T + X_S)$ the execution.} 
        \label{subfig:normal_v2_multiple_exps}
    \end{subfigure}
    \begin{subfigure}[t]{0.32\textwidth} \centering
        \includegraphics[width=0.99\textwidth]{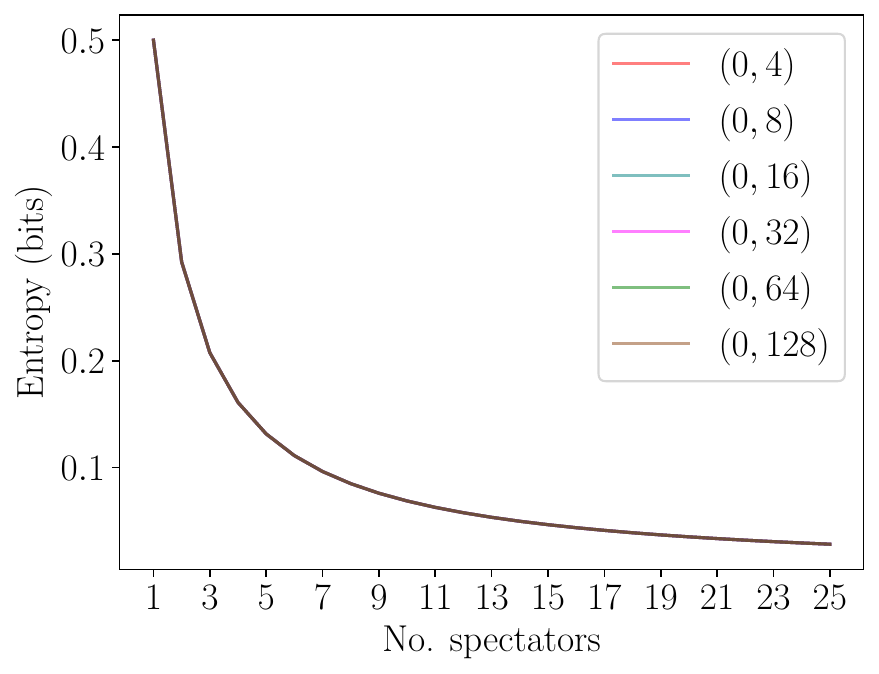}
        \caption{Target's absolute entropy loss $h(\vec{X}_T) - h(\vec{X}_T \mid X_T + X_S)$.} 
        \label{subfig:normal_v2_absolute_loss}
    \end{subfigure}
    \begin{subfigure}[t]{0.32\textwidth} \centering
        \includegraphics[width=0.99\textwidth]{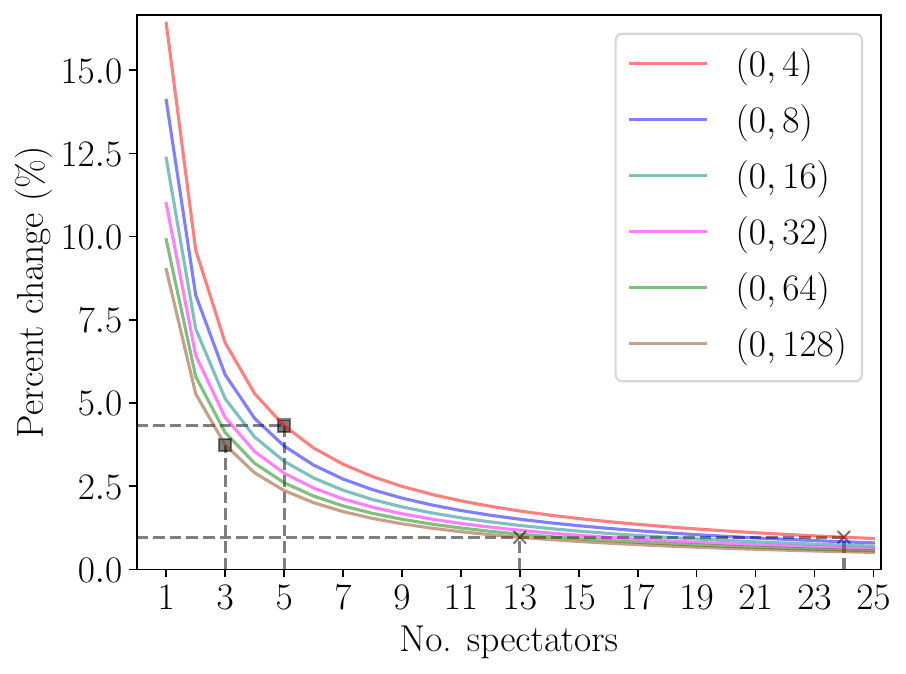}
        \caption{Target's relative entropy loss $\frac{h(\vec{X}_T) - h(\vec{X}_T \mid X_T + X_S)}{h(\vec{X}_T)}$.} 
        \label{subfig:normal_v2_relative_loss}
    \end{subfigure}
    \caption{Analysis of target's entropy loss using the normal distribution with  $\N(0,\sigma^2)$, and varying $\sigma^2$ with $|T| = 1$.}
        \label{fig:normal_v2}
    \end{figure*}

    \begin{figure*}[ht]
        \centering
        \begin{subfigure}[t]{0.32\textwidth} \centering
            \includegraphics[width=0.99\textwidth]{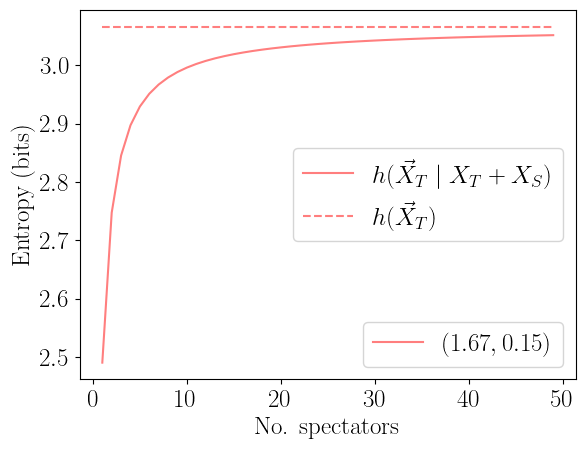}
            \caption{Target's entropy before $h(\vec{X}_T)$ and after $h(\vec{X}_T \mid X_T + X_S)$ the execution.}  
            \label{fig:lognormal_v2_multiple_exps}
        \end{subfigure}
        \begin{subfigure}[t]{0.32\textwidth} \centering
            \includegraphics[width=0.99\textwidth]{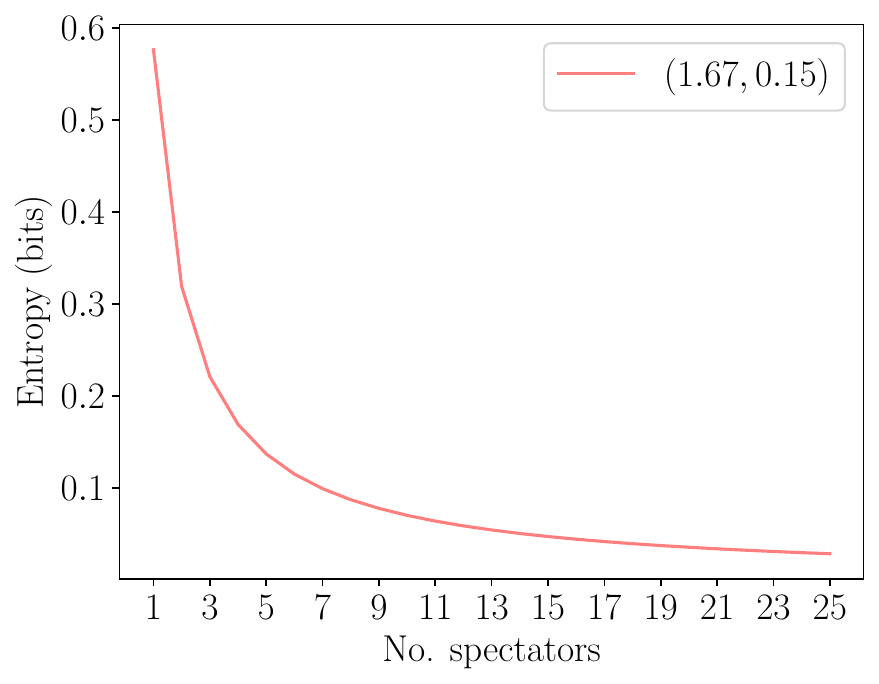}
            \caption{Target's absolute entropy loss $h(\vec{X}_T) - h(\vec{X}_T \mid X_T + X_S)$.} 
            \label{fig:lognormal_v2_absolute_loss}
        \end{subfigure}
        \begin{subfigure}[t]{0.32\textwidth} \centering
            \includegraphics[width=0.99\textwidth]{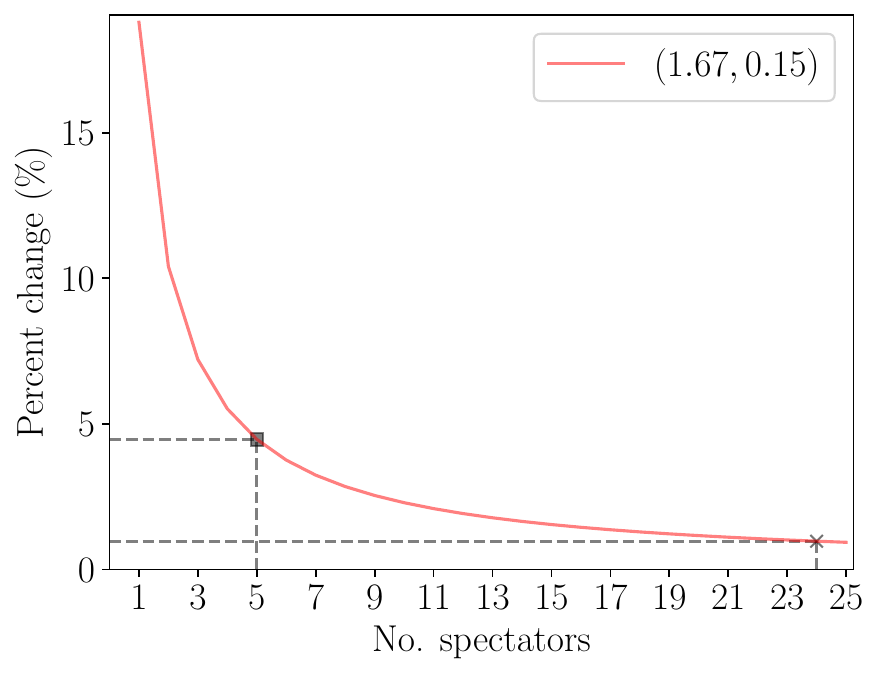}
            \caption{Target's relative entropy loss $\frac{h(\vec{X}_T) - h(\vec{X}_T \mid X_T + X_S)}{h(\vec{X}_T)}$.} 
            \label{fig:lognormal_v2_relative_loss}
        \end{subfigure}
    \caption{Analysis of target's entropy loss using the log-normal distribution with  $\log\N(1.6702,0.145542)$ and $|T| = 1$.}
        \label{fig:lognormal_v2}
    \end{figure*}

\subsection{Continuous Distributions} 
\label{sub:continuous_dist}

For continuous distributions, we shift to differential entropy and analyze normal and log-normal distributions, the latter of which is typically used to model salaries. 
While there is no direct relationship between differential and Shannon entropy (see~\cite[Chapter~8.3]{thomas2006elements}), we demonstrate that they exhibit very similar behavior for the average computation.

The differential entropy of a normal random variable $X_i\sim \N(\mu, \sigma^2)$ is $h(X_i) = \frac{1}{2}\log \lp{2 \pi e \sigma^{2}}$~\cite[Chapter~8.1]{thomas2006elements}. The sum of $n$ identical normal random variables is also normal, namely $X \sim \N(n\mu, n\sigma^2)$. This enables us to directly apply the differential entropy definition to the sum.

The log-normal distribution is a well-established means of modeling salary data for 99\% of the population~\cite{clementi2005pareto}, with the top 1\% modeled by the Pareto distribution \cite{souma2002physics}.  The differential entropy of a log-normal random variable $X_i\sim \log\N(\mu, \sigma^2)$ is $h(X_i) = \log \lp{ e^{\mu + \frac{1}{2}}\sqrt{2 \pi  \sigma^{2}}}$. However, the sum of $n$ log-normal random variables has no closed form and is an active area of research \cite{barakat1976sums,cobb2012approximating,fenton1960sum,beaulieu1995estimating,beaulieu2004optimal,schwartz1982distribution,senaratne2009numerical,wu2005flexible}. We adopt the Fenton-Wilkinson (FW) approximation 
\footnote{
    Other approximations for the sum of log-normal random variables are difficult to translate into an expression for the differential entropy and hence we choose the FW approximation. Its disadvantage is that that the FW approximation deteriorates for  $\sigma^2 > 4$ and small values of $x$ in the density function~\cite{beaulieu1995estimating,wu2005flexible}. Fortunately, our $\sigma^2$ is sufficiently small, allowing us to use the FW approximation free of consequence.
} 
\cite{fenton1960sum,cobb2012approximating} that specifies a sum of $n$ identical independent log-normal random variables $X_i \sim \log\N(\mu, \sigma^2)$ as another log-normal random variable $X \sim \log\N(\hat{\mu}, \hat{\sigma}^2)$ with parameters
\begin{align*}
& \hat{\sigma}^2 =  \ln \lp{\frac{\exp(\sigma^2) - 1}{n} + 1 },\ \hat{\mu} = \ln ( n \cdot \exp(\mu)) + \frac{1}{2} \lp{\sigma^2 - \hat{\sigma}^2}.
\end{align*}
This enables us to compute differential entropy using a closed-form expression. Unlike prior distributions, we use a single set of $\mu$ and $\sigma^2$ parameters calculated from real salary data in~\cite{cao2022priori}; namely, $\mu = 1.6702$ and $\sigma^2 = 0.145542$.

Figures \ref{fig:normal_v2} and~\ref{fig:lognormal_v2} present experimental evaluation of entropy loss with a single target and a varying number of spectators for normal and log-normal distributions, respectively. As before, we report target's entropy before and after the execution, the difference of the two as the absolute entropy loss, and the entropy loss relative to the entropy before the execution.

In Figure~\ref{fig:normal_v2} (normal),
we set the mean $\mu = 0$ for all experiments (since differential entropy does not depend on $\mu$) and vary $\sigma^2$ from 4 to 128. The results are consistent with the discrete counterparts in terms of the trends, curve shapes, and specific values.
The absolute loss in Figure~\ref{subfig:normal_v2_absolute_loss} is once again constant for any $\sigma^2$ and the relative loss is dictated by the amount of input's entropy in Figure~\ref{subfig:normal_v2_relative_loss}. When  $\sigma^2 = 4$ and inputs have 3 bits of entropy, the number of spectators required  to maintain at most 5\% and 1\% entropy loss (5 and 24 spectators, respectively) is the same as for Poisson and uniform distributions with 3-bit inputs ($\lambda = 4$ and $N=8$, respectively). With 5.5-bit inputs ($\sigma^2 = 128$), 3 and 13 spectators are needed to achieve at most 5\% and 1\% loss, respectively, which the same as for Poisson distribution with 5.5-bit inputs ($\lambda = 128$).

The results in Figure~\ref{fig:lognormal_v2} (log-normal with real salary parameters) are consistent with both the discrete and continuous distributions. Surprisingly, we observe the same 5 and 24 spectators achieve at most 5\% and 1\% relative loss, as observed with all other distributions (with input original entropy being slightly over 3 bits).

Before concluding our discussion of continuous distributions, we are able to show one more result. 
We experimentally demonstrated that the amount of absolute entropy loss is parameter-independent for several distributions, but we can formally prove this for normally distributed 
inputs:
\begin{claim} \label{claim2}
    If the inputs are modeled by independent identically distributed normal random variables, the absolute entropy loss $h(\vec{X}_T) - h(\vec{X}_T \mid X_T + X_S) $ depends only on the number of target $\norm{T} = t$ and spectator $\norm{S} = n$ inputs and is $\frac{1}{2} \log\lp{ \frac{t}{n} + 1}$.
\end{claim}
\newcommand\prooftwo{
Let $\norm{T} = t$ and $\norm{S} = n$, such that $X_T \sim \N (0, t\sigma^2)$ and $X_S \sim \N (0, n\sigma^2)$. The absolute entropy loss is therefore 
\begin{align*}
h(\vec{X}_T) \ifCONF\else\fi- h(\vec{X}_T \mid X_T + X_S)  \ifCONF \eqanchor \else \eqanchor \fi
= h(\vec{X}_T) - \lp{h(\vec{X}_T) +h\lp{X_S}   - h\lp{X_T +X_S}}
\ifCONF \breakcmd \else \fi
=  h\lp{X_T +X_S} - h\lp{X_S}  \\
&= \frac{1}{2} \log 2 \pi e (t + n)\sigma^2  - \frac{1}{2} \log 2 \pi e n\sigma^2 \ifCONF \breakcmd  \else \fi=
  \frac12 \log\lp{ \frac{t}{n} + 1} = \Theta\lp{\log\lp{ \frac{t}{n} + 1}} ,
\end{align*}
which depends only on $n$ and $t$.}
\ifCONF 
\begin{proof}
    \prooftwo
    \end{proof}
\else
\begin{proof}
\prooftwo
\end{proof}
\fi

\subsection{Discrete vs. Continuous Distributions} 
\label{ssub:Discrete_v_Continuous}

We next compare the information loss across all four (discrete and continuous) distributions. We choose parameters to maintain the initial entropy of an input, $H({X}_i)$ or $h({X}_i)$, to be approximately 3 bits, as to reasonably correspond to the log-normal distribution. This leads to $\text{Pois}(4)$, $\unif \lp{0, 7}$, and $\N(0, 4)$. We plot this information for a single target and a varying number of spectators in Figure~\ref{fig:disc_cont_abs_rel}. 

In the figure, all distributions converge with $\ge 4$  spectators and are very close even with 3 spectators. This convergence on large values is expected as a consequence of the central limit theorem. From the four distributions, the closest are the Poisson results with $\lambda=4$ (discrete) and the normal distribution $\N(0, 4)$ (continuous). 
Unlike normal, log-normal, and the single-variate uniform distributions, an exact expression of the entropy of a Poisson distribution has not been derived. Instead, when computing the necessary values in Section~\ref{sub:Discrete_Distributions}, we directly applied the definition of Shannon entropy. To draw a parallel between discrete and continuous distributions, and specifically show a similarity between Poisson and normal distributions, we turn to an approximation of Poisson distribution's entropy computation.

It was conjectured that for sufficiently large $\lambda$ (e.g., $\lambda > 10$), the Poisson distribution's Shannon entropy can be approximated by $H(X_i) = \frac{1}{2} \log (2 \pi e \lambda)$, which resembles $h(X_i) = \frac{1}{2} \log (2 \pi e \sigma^2)$ used for normal distributions. Evans~and~Boersma~\cite{evans1988entropy} proposed a tighter bound (further formalized by Cheraghchi~in~\cite{cheraghchi2019expressions}), to be 
\begin{align*}
    H(X_i) = \frac{1}{2} \log (2 \pi e \lambda) - \frac{1}{12\lambda} - \frac{1}{24 \lambda^{2}} - \frac{19}{360 \lambda^3} + O(\lambda^4)
\end{align*}
and remains close to that of normal distribution with $\sigma^2 = \lambda$.

 One implication of this result for us is that Claim~\ref{claim2}, which we demonstrated for normal distributions, would apply to the approximation of Poisson distributions as well. As a result, we obtain independence of the (absolute) entropy loss of distribution parameters for both discrete and continuous distributions and almost identical behavior across the distributions as a function of the number of spectators.

\begin{figure}[t]
    \centering
        \ifCONF
        \includegraphics[width=0.33\textwidth]{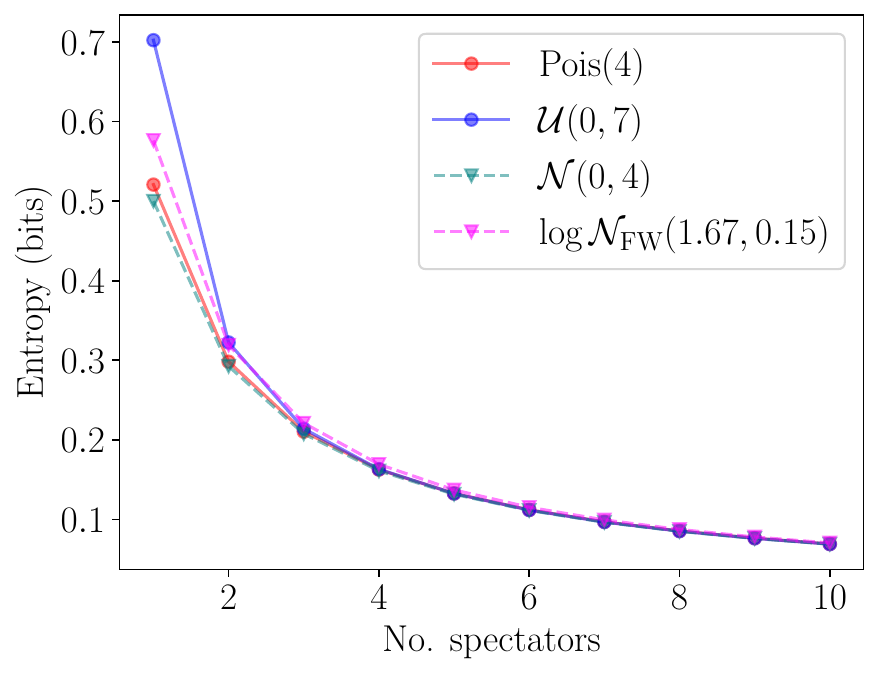}
       \else 
        \includegraphics[width=0.4\textwidth]{dscrt_cnt_abs_ln.pdf}
    \fi
    
        \label{fig:discrete_continuous_absolute_loss}
    \ignore{    
        \begin{subfigure}[t]{0.49\textwidth} \centering
            \includegraphics[width=0.67\textwidth]{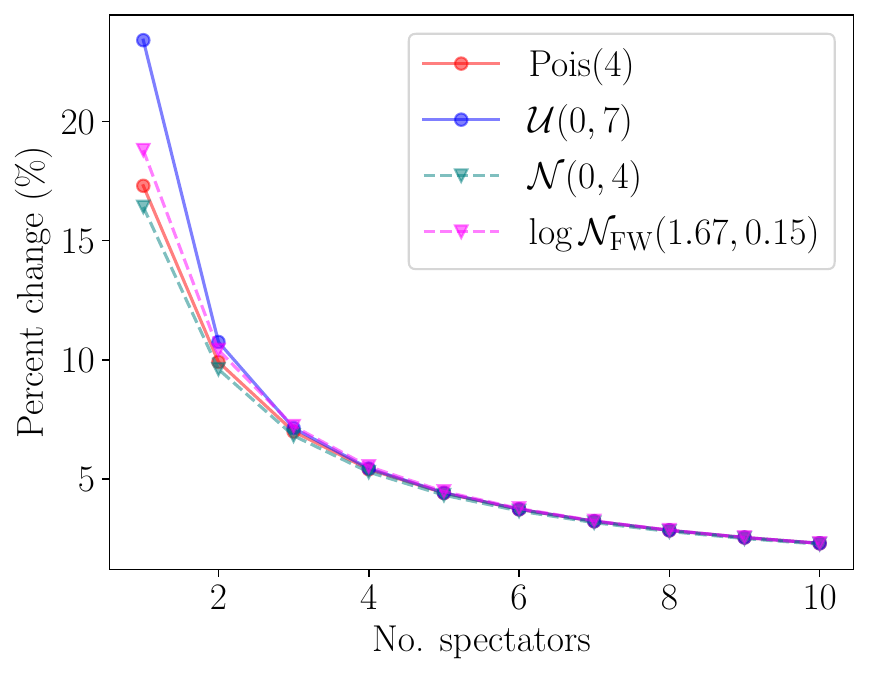}
            \caption{Target's relative entropy loss $\frac{H(\vec{X}_T) - H(\vec{X}_T \mid X_T + X_S)}{H(\vec{X}_T)}$ (discrete) and $\frac{h(\vec{X}_T) - h(\vec{X}_T \mid X_T + X_S)}{h(\vec{X}_T)}$ (continuous).}
    
            \label{fig:discrete_continuous_relative_loss}
        \end{subfigure}
    }
        \caption{Comparing target's absolute entropy loss for discrete $H(\vec{X}_T) - H(\vec{X}_T \mid X_T + X_S)$ and continuous $h(\vec{X}_T) - h(\vec{X}_T \mid X_T + X_S)$ distributions.}
        \label{fig:disc_cont_abs_rel}
    \end{figure}

\ifCONF
Furthermore, our analysis using Shannon/differential entropy is partially echoed for the min-entropy (demonstrated in the full version of the text~\cite{baccarini2022understanding}).
 In particular, we conjecture independence of the attacker's input in the min-entropy-based awae and show that plots for the absolute min-entropy loss closely resemble those for Shannon entropy for the Poisson distribution. 
\else

\subsection{Min-Entropy Analysis}
\label{sub:min_entropy_analysis}

We treat min-entropy as an alternative to Shannon entropy, which was studied in the context of information flow by Smith~\cite{smith2009foundations}. While we are unable to go as far in our analysis as in the case of (Shannon) entropy, we certainly observe similar trends. We begin by defining the concept of \emph{vulnerability}:
\begin{definition}[Vulnerability, \cite{smith2009foundations}]
Given a discrete random variable $X$ with support $\X$, the \emph{vulnerability} of $X$, denoted by $V_{\infty}(X)$ over the unit interval $\lb{0,1}$ is given by
\begin{align*}
V_{\infty}(X) = \max_{x \in \X} \Pr(X = x).
\end{align*} 
\end{definition}
The vulnerability $V_{\infty}(X)$ is interpreted as the worst-case probability that an adversary could guess the value of $X$ in one attempt. If $m$ guesses are allowed, the adversary's success probability is at most $m V_{\infty}(X)$. The implication is that if the vulnerability with a practical number of $m$ guesses is significant, then $V_{\infty}(X)$ must also be significant.  Since the vulnerability is a probability, we can convert it to an entropy measure (in bits) by taking the logarithm of $V_{\infty}(X)$. Conveniently, this is exactly the definition of \emph{min-entropy} $H_{\infty}(X)$:
\begin{align*}
H_{\infty}(X) = \log \frac{1}{V_{\infty}(X)}.
\end{align*}
Smith's~\cite{smith2009foundations} motivation for departing from Shannon entropy stems from its ineffectiveness of properly assessing the threat the output $Y$ has on its input $X$.

Since our analysis studies the relationship between input and output random variables (i.e. $X_T$ and $O$), a necessary extension is the \emph{conditional vulnerability}, which specifies the expected probability of guessing $X$ in one try, given that $Y$ is observed:
\begin{definition}
Given two random variables $X$ and $Y$ with supports $\X$ and $\Y$, resp., the \emph{conditional vulnerability} $V_{\infty}( X \mid Y)$ is
\begin{align*}
V_{\infty}(X\mid Y) = \summ{y \in \Y}{} \Pr(Y = y)\cdot  V_{\infty}(X  \mid Y = y),
\end{align*} 
where 
\begin{align*}
V_{\infty }(X  \mid Y = y) = \max_{x \in \X} \lp{\Pr(X = x \mid Y = y)}.
\end{align*}
\end{definition}

Having established the necessary foundations of min-entropy, we are equipped to extend our single-execution analysis of Section~\ref{sec:single_calculation} from the perspective of min-entropy:
\begin{definition}[] \label{def:awae_min}
    The attacker's weighted average min-entropy ($\textup{awae}_{\infty}$) of a target $\vec{X}_T$ attacked by parties $A$ is defined for all $\vec{x}_A \in D_A$ as  
    \begin{align*}
        \textup{awae}_{\infty} (\vec{x}_A) &=  H_{\infty}(X_T \mid O, \vec{X}_A = \vec{x}_A) =-\log \sum_{o \in D_O} \Pr(O = o \mid \vec{X}_A = \vec{x}_A ) \cdot V_{\infty}( \vec{X}_T \mid \vec{X}_A = \vec{x}_A, O = o), 
      \end{align*}
    where $ V_{\infty}( \vec{X}_T \mid \vec{X}_A = \vec{x}_A, O = o)$ is the conditional vulnerability defined above. 
\end{definition} 
The above definition is a concrete min-entropy specification of Ah-Fat~and~Huth's~\cite{ah2019optimal} generalized awae, which is parameterized by $\alpha$ and a gain function $g$. We can manipulate Definition~\ref{def:awae_min} into terms consistent with Section~\ref{sec:single_calculation} by plugging in the expression for conditional vulnerability:
\begin{align*}
    \textup{awae}_{\infty} (\vec{x}_A) &= -\log \sum_{o \in D_O} \Pr(O = o \mid \vec{X}_A = \vec{x}_A )  \cdot \lp{\max_{\vec{x}_T \in D_T}  \Pr( \vec{X}_T  = \vec{x}_T \mid \vec{X}_A = \vec{x}_A, O = o) } \\ 
    &= -\log \sum_{o \in D_O} \Pr(O = o \mid \vec{X}_A = \vec{x}_A )  \cdot \lp{\max_{\vec{x}_T \in D_T} \frac{\Pr( O = o \mid \vec{X}_T  = \vec{x}_T, \vec{X}_A = \vec{x}_A) \cdot \Pr( \vec{X}_T  = \vec{x}_T) }{\Pr(O = o \mid \vec{X}_A = \vec{x}_A ) }} \\
    &= -\log \sum_{o \in D_O}  \left(\max_{\vec{x}_T \in D_T} \Pr( O = o \mid \vec{X}_T  = \vec{x}_T, \vec{X}_A = \vec{x}_A) \right. \left.\cdot \Pr( \vec{X}_T  = \vec{x}_T) \right).
 \end{align*}
In the second line we invoked Bayes' theorem, and in the third line we observed that the denominator is a constant factor in the max expression and could thus be factored out and subsequently cancelled with the leading $\Pr(O = o \mid \vec{X}_A = \vec{x}_A )$. 

In Claim~\ref{claim1}, we proved  $\textup{awae}(\vec{x_A})$  was independent of the attacker's input $\vec{x}_A$. Conversely, $\textup{awae}_{\infty}(\vec{x_A})$ cannot be simplified further to prove the claim holds. Hence, we conjecture the following:
\begin{conjecture}
  $\textup{awae}_{\infty}(\vec{x_A})$  is independent of attacker's input vector $\vec{x_A}$. 
\end{conjecture}
We can, however, repeat the calculation of Figure~\ref{fig:twae_vs_awae} using min-entropy. In Figure~\ref{fig:min_awae}, we once again observe the same behavior that the adversarial knowledge does not change by varying its inputs into the computation.\footnote{Interestingly,  $\textup{awae}_{\infty}(\vec{x_A})$ is the same for $\norm{S} = 1$ and $\norm{S} = 2$  and the curves overlap on the plot.} This suggests the conjecture holds for the average salary computation. Hence, we assume such in our subsequent analysis.

The next logical step is to examine how the  effect  transitioning from Shannon entropy to min-entropy has on the absolute loss. 
We compute and display both absolute losses in Figure~\ref{fig:min_v_shannon_abs}, where participants' inputs are modeled by the Poisson distribution (as in Section~\ref{sub:Discrete_Distributions}). 
As observed previously in Figure~\ref{subfig:sum_poisson_absolute_loss}, the Shannon absolute loss curves all overlap each other. Interestingly, we observe the min-entropy absolute loss curves converge towards their Shannon counterparts as $\lambda$ grows. This suggests that for a sufficiently large statistical parameter, the choice of metric used to represent information disclosure is less impactful.

\begin{figure*}[t]
    \centering
   \begin{subfigure}[t]{0.45\textwidth} \centering
     \includegraphics[width=0.9\textwidth]{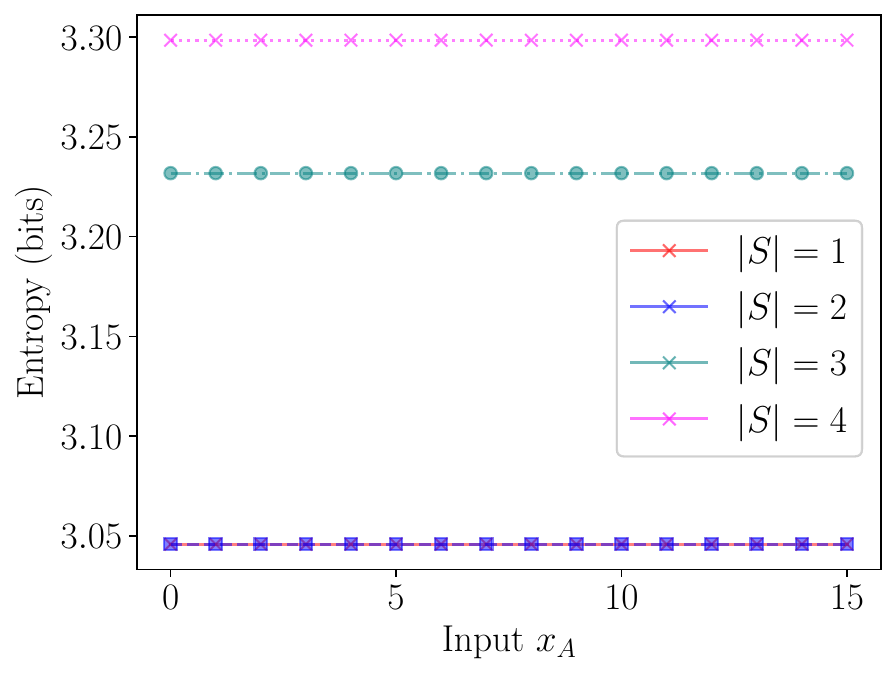}
     \caption{The  $\awae_{\infty}(\vec{x}_A)$ using uniformly distributed inputs over $\unif \lp{0,15}$ with a different number spectators $\norm{S}$.}
     \label{fig:min_awae}
    \end{subfigure}
   \begin{subfigure}[t]{0.45\textwidth} \centering
   \includegraphics[width=0.9\textwidth]{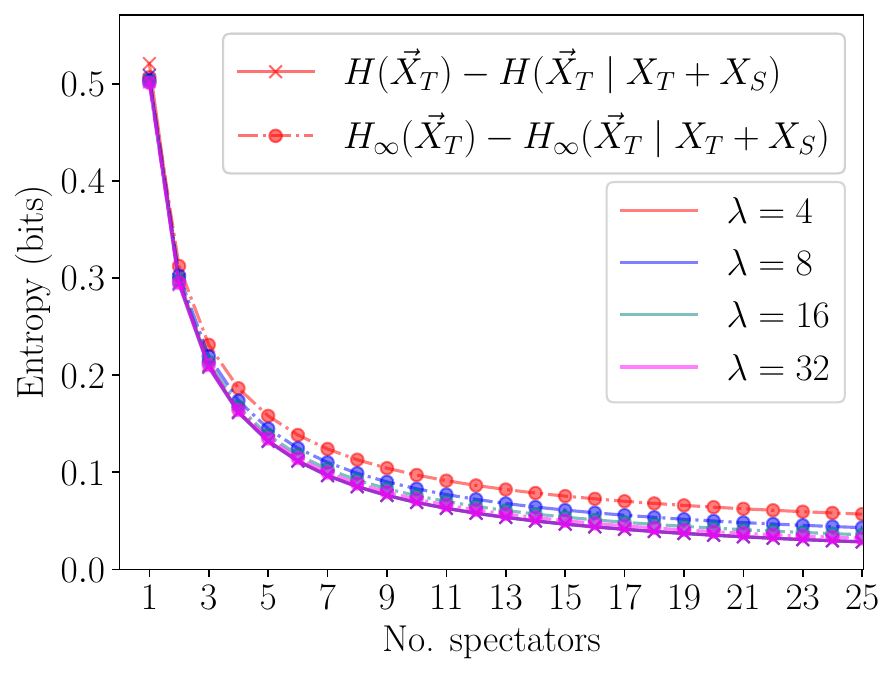}
   \caption{Comparing the Shannon and min-entropy absolute losses using Poisson distribution with $\pois(\lambda)$, varying $\lambda$ with $|T| = 1$.}
   \label{fig:min_v_shannon_abs}
   \end{subfigure}
   \caption{Min-entropy analysis.}
   \label{fig:min_ent_figs}
 \end{figure*}

\fi

\ifCONF \else
\subsection{Mixed Distribution Parameters}
\label{sub:mixed_distribution_parameters}
\fi

Up to this point, we have assumed that all participants' inputs are sampled from identically distributed random variables. However, we can relax this assumption and investigate if/how the information disclosure changes if parties' inputs are non-identically distributed. For example, employee salaries may differ slightly from company to company, while still following the same distribution. We can model this by adjusting the statistical parameters of individual participants.
\ifCONF
This poses an interesting problem where there are multiple groups with different distribution parameters.
As such, we are interested in determining which prior claims are still valid or need to be modified. Claim~\ref{claim1} (attacker input independence) will hold regardless of how participants' inputs are distributed. Conversely, Claim~\ref{claim2} (dependence on the number of targets and participants in the absolute entropy loss) must be reworked since the claim is formulated under the assumption that inputs are identically distributed. We investigate this relaxation, as well as conduct additional experiments, in the full version of the paper~\cite{baccarini2022understanding}.
\else 

We begin by formalizing the notion of participant ``groups''. Define $\grp$ as a finite set of statistical distributions, of which participants' inputs can be sourced from. For example, if we have two groups $\sf B$ and $\sf C$ of normally distributed inputs parameterized by $\N(0, \sigma^2_\textsf{B})$ and $\N(0,\sigma^2_\textsf{C})$ (where $\sigma_\textsf{B}^2 \neq \sigma_\textsf{C}^2$), respectively, then $\grp = \lbr{\textsf{B}, \textsf{C}}$.  This formulation poses two interesting directions for introducing participant group identities, i.e., correspondence of a participant to one of the distribution groups, into our analysis:
\begin{itemize}
    \item \textbf{Group identities of individual participants are known.} The first setting we consider is that the identity, i.e., the group, of each individual participant is known. In practice, this is realized by multiple entities with inputs modelled by different statistical distributions contribute to a computation, where the number of inputs submitted by each is publicly available.

    \item \textbf{Group identities of individual participants are unknown.} Conversely, we have the scenario where we have knowledge of the possible distribution groups participants can belong to with anticipated likelihoods, but the group identity of an individual party is not known. This is objectively more general than the first category, but requires knowledge in the form of the probabilities of an arbitrary participant belonging to each group.
    
\end{itemize}
It is therefore of interest to revisit our prior conclusions under the known \emph{and} unknown group identity generalizations (denoted by Cases 1 and 2, respectively), since both formulations bear operational significance. 

\medskip\noindent
\textbf{Entropy loss as a result of computation participation.}
The first conclusion we revisit is Claim~\ref{claim1}, since it is integral to our analysis as a whole. The claim states that the information disclosure from the average function output is independent of the attackers' inputs. Based on this result, our subsequent analysis enabled us to derive expressions for $\textup{awae}$.

In the current generalized setting, Claim~\ref{claim1} remains true for both Cases 1 and 2, since the derivation in the proof of Claim~\ref{claim1} itself remains unchanged.
However, we must adjust Equation~\ref{eq:claim1_followup} such that participant group membership is captured by our entropy measure.
We recall our definitions of entropy remaining after participation and the absolute entropy loss:
\begin{align}
    H \lp{\vec{X}_T \mid X_T + X_S} &=  H\lp{\vec{X}_T}+H\lp{X_S} -H\lp{X_T+ X_S} \label{eq:disclosure}\\ 
    H(\vec{X}_T) - H(\vec{X}_T \mid X_T + X_S) &=  H\lp{X_T +X_S} - H\lp{X_S} \label{eq:abs_loss}
\end{align}
Accounting for group identities, we introduce the \emph{participant identity random variable} $\ID_{P_i}$ supported by $\grp$. This corresponds to the group identity of participant $P_i$, and we denote $\id_{P_i} \in \grp$ as the value $\ID_{P_i}$ takes. We similarly denote $\vec{\ID}_P = \lp{\ID_{P_1}, \dots, \ID_{P_m}}$ as a multidimensional random variable, with $\vec{\id}_P$ as the vector of individual values of the same size. 

At this point, our analysis splits into two directions based on the knowledge of individual group identities.
\begin{caseof}
    \case{}{ 
If participant group identities are available (i.e., $\vec{\ID}_T = \vec{\id}_T$ and $\vec{\ID}_S = \vec{\id}_S $), Equation~\ref{eq:disclosure} becomes:
\begin{align}
    \begin{split}
    H \lp{\vec{X}_T \mid X_T + X_S, \vec{\ID}_T = \vec{\id}_T,  \vec{\ID}_S = \vec{\id}_S} = H\lp{\vec{X}_T\mid \vec{\ID}_T = \vec{\id}_T }+& H\lp{X_S\mid  \vec{\ID}_S = \vec{\id}_S}  \\
    -&  H\lp{X_T+ X_S \mid \vec{\ID}_T = \vec{\id}_T,  \vec{\ID}_S = \vec{\id}_S }.
    \end{split}
    \label{eq:case_1_entropy}
\end{align}
Since we have exact knowledge of each participant's identity, we can explicitly partition the input random variables accordingly. Therefore, all the above quantities are computable with minimal deviation from our original analysis. For instance, if we recall our earlier example with two participant groups $\textsf{B}$ and $\textsf{C}$, the term $ H\lp{X_S\mid  \vec{\ID}_S = \vec{\id}_S}$ can be computed as:
\begin{align*}
    H\lp{X_S\mid  \vec{\ID}_S = \vec{\id}_S} =  H\lp{\summ{j \in S}{}X_{S_j}\mid  \vec{\ID}_S = \vec{\id}_S} = H\lp{ \lp{\sum_{\substack{j \in S \colon \id_{S_j} = \textsf{B} }}    X_{S_j} + \sum_{\substack{j \in S \colon \id_{S_j} = \textsf{C} }}  X_{S_j}}\mid  \vec{\ID}_S = \vec{\id}_S}.
\end{align*}
The absolute entropy loss directly follows from Equation~\ref{eq:case_1_entropy}, and is computed as: 
\begin{align}
        \begin{split}
    H \lp{\vec{X}_T \mid \vec{\ID}_T = \vec{\id}_T} &-  H\lp{\vec{X}_T \mid X_T + X_S, \vec{\ID}_T = \vec{\id}_T,  \vec{\ID}_S = \vec{\id}_S} \\ 
    &\myquad[10]= 
    H(X_T + X_S \mid\vec{\ID}_T = \vec{\id}_T,  \vec{\ID}_S = \vec{\id}_S ) - H(X_S \mid  \vec{\ID}_S = \vec{\id}_S ) .
    \end{split}
    \label{eq:abs_loss_assump_1}
\end{align} 
Under this generalization, the group to which a target belongs impacts how much information is disclosed from the computation, i.e., the disclosure can fall within a range  based on the values $\vec{\id}_T$ can take. We determine the worst-case information disclosure by iterating over all possible target identities and taking the maximum:
\begin{align*}
\max_{\vec{\id}_{T}} \lp{H(X_T + X_S \mid\vec{\ID}_T = \vec{\id}_T,  \vec{\ID}_S = \vec{\id}_S ) - H(X_S \mid  \vec{\ID}_S = \vec{\id}_S ) }.
\end{align*}
We can further refine our earlier notation established in Section~\ref{sec:definitions} to encompass participant group identities (applicable to both targets and spectators). Let $P_{\GG} \subset P$ be the set of participants belonging to group $\GG \in \grp$. The sum of random variables modelling participant inputs is given as 
\begin{align*}
X_S = \sum_{\GG \in \grp} \sum_{i \in P_{\GG}} X_{P_{i}} =  \sum_{\GG \in \grp} X_{P_{\GG}},
\end{align*}
where  $ X_{P_{\GG}} =   \sum_{i \in P_{\GG}} X_{P_{i}}$. 

    }

\case{}{
When the group identities of individual inputs are not known and only the probability of belonging to a given group is given, the procedure for evaluating the information disclosure changes. The probability mass and density functions, respectively, for the  participant inputs random variables are now:
\begin{align*}
    \Pr \lp{\vec{X}_P=\vec{x}_P} 
    &= \sum_{\vec{\id}_{P}} \Pr \lp{\vec{\ID}_P = \vec{\id}_P } \Pr \lp{\vec{X}_P=\vec{x}_P \mid\vec{\ID}_P = \vec{\id}_P }\\
       f\lp{\vec{x}_P} &= \sum_{\vec{\id}_{P}} \Pr \lp{\vec{\ID}_P = \vec{\id}_P } f \lp{\vec{x}_P \mid\vec{\ID}_P = \vec{\id}_P }.
\end{align*}
For a participant set $P$, there are $\norm{\grp}^{\norm{P}}$ possible identity configurations, such that the number of terms in the summation is exponential in the number of participants  and/or size of the group identity set.
The Shannon and differential entropies are now computed as: 
\begin{align}
    H(\vec{X}_P) 
    &= {-} \sum_{\vec{x}_P\in D_{X_P}} \lp{ \sum_{\vec{\id}_{P}} \Pr \lp{\vec{\ID}_P = \vec{\id}_P }\Pr \lp{\vec{X}_P=\vec{x}_P {\mid}\vec{\ID}_P = \vec{\id}_P }}  \log \lp{ \sum_{\vec{\id}_{P}}\Pr \lp{\vec{\ID}_P = \vec{\id}_P } \Pr \lp{\vec{X}_P=\vec{x}_P {\mid}\vec{\ID}_P = \vec{\id}_P }}, \nonumber\\
    h(\vec{X}_P) 
    &= {-} \int_{\X_P} \lp{ \sum_{\vec{\id}_{P}}\Pr \lp{\vec{\ID}_P = \vec{\id}_P }f \lp{\vec{x}_P \mid\vec{\ID}_P = \vec{\id}_P }} \log \lp{ \sum_{\vec{\id}_{P}}\Pr \lp{\vec{\ID}_P = \vec{\id}_P } f \lp{\vec{x}_P \mid\vec{\ID}_P = \vec{\id}_P }} d\vec{x}_P .
    \label{eq:case2_entropy}
  \end{align}
The fundamental difference between the entropy calculation under this generalization and the  analysis conducted in Sections~\ref{sub:Discrete_Distributions}~and~\ref{sub:continuous_dist} is that \emph{the entropy of these random variables is no longer exactly modeled by the input distribution itself} (e.g., Poisson, uniform, Gaussian, log-normal). Furthermore, for continuous input distributions,  our previous approach of leveraging closed-form expressions is no longer applicable when group identities are unknown  -- the information disclosure must be computed numerically, rather than exactly. 
    }

\end{caseof}

\medskip\noindent 
\textbf{Parameter independence of the absolute loss for normally distributed inputs.}
The next conclusion we revisit is Claim~\ref{claim2}, which previously stated that for normally distributed inputs, the absolute entropy loss depends only on the number of targets and spectators present in the computation. This conclusion changes in the generalized setting when the participants' inputs are no longer identically distributed, as we demonstrate below.

\begin{caseof}
\case{}{
   Let $\sigma^2_{\GG}$ be the standard deviation of the participants inputs that belong to group $\GG$.
Using the definitions from Section~\ref{sub:continuous_dist} for the entropy sums of identically distributed normal random variables, the entropy of $ X_{P_{\GG}}$ is
   $h( X_{P_{\GG}}) = \frac{1}{2} \log \lp{2 \pi e \sigma^2_{\GG} \norm{P_{\GG}}}$. 
   We similarly derive the following expressions needed to compute the absolute entropy loss (given in Equation~\ref{eq:abs_loss_assump_1}):
\begin{align*}
    h(X_T + X_S \mid\vec{\ID}_T = \vec{\id}_T,  \vec{\ID}_S = \vec{\id}_S ) &= \frac{1}{2} \log 2 \pi e \lp{\sum_{\GG \in \grp} \lp{\sigma^2_{\GG}\mdot \norm{T_{\GG}} +\sigma^2_{\GG}\mdot \norm{S_{\GG}}}} \\ 
    h(X_S  \mid \vec{\ID}_S = \vec{\id}_S ) &= \frac{1}{2} \log 2 \pi e \lp{\sum_{\GG \in \grp} \lp{\sigma^2_{\GG}\mdot \norm{S_{\GG}}} }.
\end{align*}
Plugging in these equations to our expression for the absolute entropy loss and simplifying yields:
\begin{align*}
    h(X_T + X_S \mid\vec{\ID}_T = \vec{\id}_T,  \vec{\ID}_S = \vec{\id}_S ) - h(X_S \mid  \vec{\ID}_S = \vec{\id}_S ) = \frac{1}{2} \log \lp{\frac{\sum_{\GG \in \grp} \lp{\sigma^2_{\GG}\mdot \norm{T_{\GG}}}}{\sum_{\GG \in \grp} \lp{\sigma^2_{\GG}\mdot \norm{S_{\GG}}}} + 1}.
\end{align*}
Unlike the analysis in the proof of Claim~\ref{claim2}, the standard deviations do not cancel. However, we can reformulate our interpretation of the sums of standard deviations when accounting for group identities. Let us define $\sigma^2_{\text{B}} \in \R_{>0}$  as the ``base standard deviation'' for all input random variables. Then, for all $\GG \in \grp$, there exists some $\delta_{\GG} > 0$ such that $\sigma^2_{\GG} = \delta_{\GG}\mdot \sigma^2_{\text{B}}$. Substituting into the above expression yields: 
\begin{align*}
    h(X_T + X_S \mid\vec{\ID}_T = \vec{\id}_T,  \vec{\ID}_S = \vec{\id}_S ) - h(X_S \mid  \vec{\ID}_S = \vec{\id}_S )  &= \frac{1}{2} \log \lp{\frac{\sum_{\GG \in \grp} \lp{\delta_{\GG}\mdot\sigma^2_{\text{B}}\mdot \norm{T_{\GG}}}}{\sum_{\GG \in \grp} \lp{\delta_{\GG}\mdot\sigma^2_{\text{B}}\mdot \norm{S_{\GG}}}} + 1} \\
    &=  \frac{1}{2} \log \lp{\frac{\sum_{\GG \in \grp} \lp{\delta_{\GG}\mdot \norm{T_{\GG}}}}{\sum_{\GG \in \grp} \lp{\delta_{\GG}\mdot \norm{S_{\GG}}}} + 1}.
\end{align*}
The key conclusion from the above equation is that the absolute entropy loss is not directly affected by the statistical parameter $\sigma_{\GG}^2$, but rather the \emph{relationship} each $\sigma_{\GG}^2$ has (via the scaling factor $\delta_{\ID}$) to the base standard deviation  $\sigma_{\text{B}}^2$. 

To demonstrate this phenomenon, we compute the absolute entropy loss when the target belongs to one of three possible groups differing by $\pm 10\%$ in their average salary. Concretely, we have $\textsf{B}$, $\textsf{C}$, or $\textsf{D}$ with deviations $\sigma^2_{B}$, $1.1^2\sigma^2_{B}$, and $0.9^2\sigma^2_{B}$, respectively. The target is interpreted to ``move'' from group to group as to maintain a consistent group size, e.g., $\lvert S_\textsf{B} \cup T\rvert = \lvert S_\textsf{C}\rvert = \lvert S_\textsf{D}\rvert$ when the targe is in group $\sf B$. This constitutes the curves displayed in Figure~\ref{fig:mixed_figures}. 
\begin{figure*}[t]
    \centering
   \begin{subfigure}[t]{0.45\textwidth} \centering
     \includegraphics[width=0.9\textwidth]{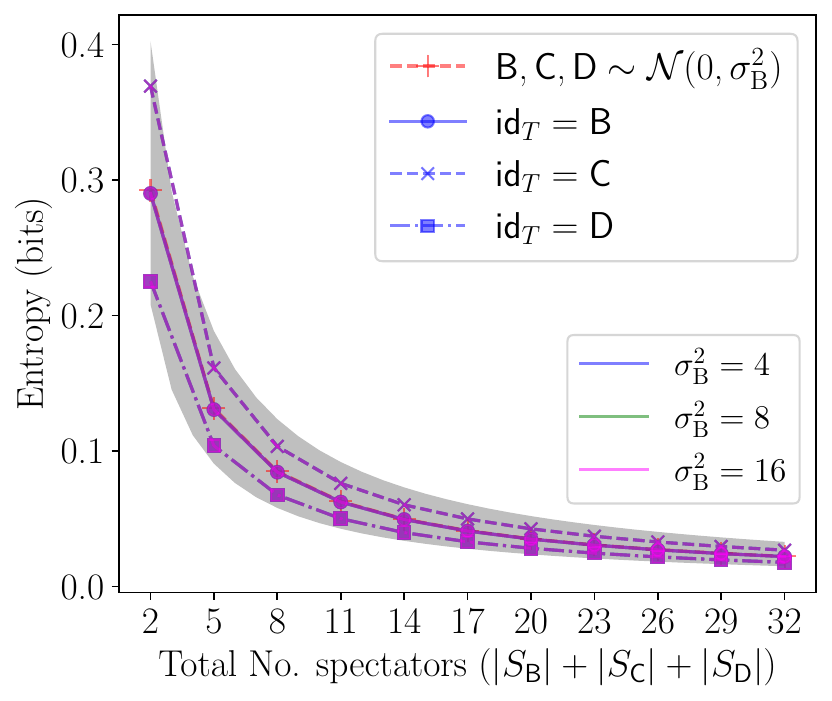}
     \caption{Absolute entropy loss   ${h(\vec{X}_T) - h(\vec{X}_T \mid X_T + X_S)}$}
     \label{fig:abs_mixed}
    \end{subfigure}
   \begin{subfigure}[t]{0.45\textwidth} \centering
   \includegraphics[width=0.94\textwidth]{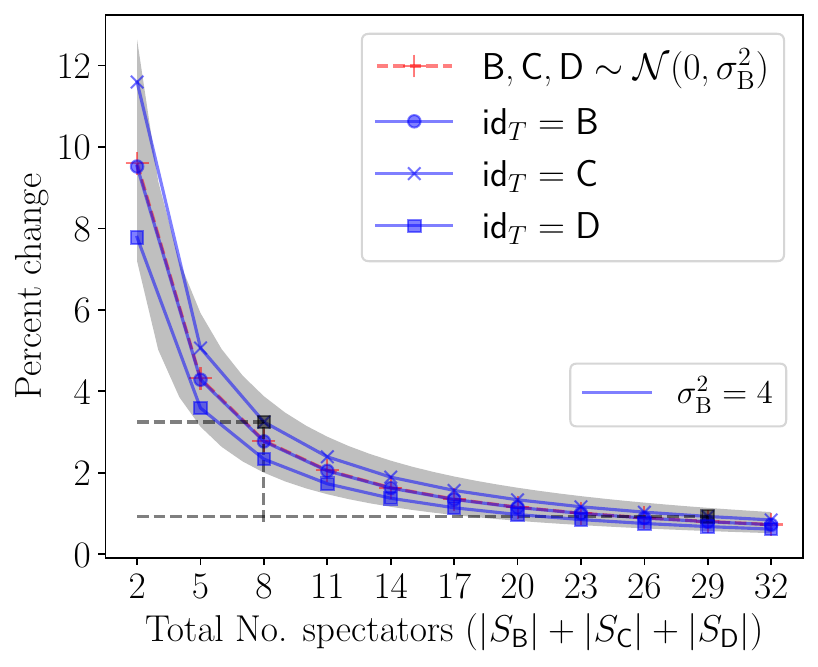}
   \caption{Relative entropy loss  $\frac{h(\vec{X}_T) - h(\vec{X}_T \mid X_T + X_S)}{h(\vec{X}_T) }$}
   \label{fig:rel_mixed}
   \end{subfigure}
   \caption{Mixed distribution analysis under Case 1. The red dashed curves correspond to our baseline where all groups are identically distributed ($\textsf{B},\textsf{C},\textsf{D} \sim  \mathcal{N}(0, \sigma^2_{\text{B}})$), while the remaining curves indicate the target belonging to distinct groups distributed by $\textsf{B} \sim \mathcal{N}(0,\sigma^2_{\text{B}} )$,  $\textsf{C} \sim \mathcal{N}(0, 1.1^2\sigma^2_{\text{B}} )$, and $\textsf{D} \sim \mathcal{N}(0, 0.9^2\sigma^2_{\text{B}} )$. The shaded regions illustrate the full space for the  absolute entropy loss, generated from every possible spectator and group configurations.}
   \label{fig:mixed_figures}
 \end{figure*}
To further illustrate the best- and worst-case information disclosure, we compute the absolute entropy loss for all possible spectator-group configurations and compose the shaded region from the maximums and minimums. Figure~\ref{fig:abs_mixed} reflects our observation that regardless of the base standard deviation $\sigma^2_{B}$ (4, 8, or 16), the curves fall on top of each other. Moreover, the absolute loss when  $\id_T =\mathsf{B}$ is equivalent to our original computation when all groups are identically distributed. We also reproduce our relative loss experiment using the same $\pm 10\%$ salary configuration for  $\sigma^2_{B} = 4$ in Figure~\ref{fig:rel_mixed}. Achieving maximum relative losses of 5\% and 1\% now requires at least 6 and 27 spectators, respectively. 
}

\case{}{When group identities of individual inputs are unknown, we refer to the definition of absolute entropy loss (Equation~\ref{eq:abs_loss}) alongside the expressions for the differential entropy we derived in Equation~\ref{eq:case2_entropy}  for the required quantities and obtain:
\begin{align*}
    h(X_T + X_S) &= 
    - \int_{\X_T \cup \X_S} \lp{ \sum_{\vec{\id}_{T}, \vec{\id}_{S}} 
    \Pr \lp{\vec{\ID}_T = \vec{\id}_T ,\vec{\ID}_S = \vec{\id}_S } 
    f \lp{x_T +x_S \mid\vec{\ID}_T = \vec{\id}_T, \vec{\ID}_S = \vec{\id}_S }} \\ 
    &\myquad[5] \cdot \log \lp{ \sum_{\vec{\id}_{T}, \vec{\id}_{S}} 
    \Pr \lp{\vec{\ID}_T = \vec{\id}_T ,\vec{\ID}_S = \vec{\id}_S } 
    f \lp{x_T+x_S \mid\vec{\ID}_T = \vec{\id}_T, \vec{\ID}_S = \vec{\id}_S }} d(x_T + x_S)
     \\
    h(X_S ) &= 
    - \int_{\X_S} \lp{ \sum_{\vec{\id}_{S}}
    \Pr \lp{\vec{\ID}_S = \vec{\id}_S } 
    f \lp{x_S \mid\vec{\ID}_S = \vec{\id}_S }} \cdot \log \lp{ \sum_{\vec{\id}_{S}}
        \Pr \lp{\vec{\ID}_S = \vec{\id}_S } 
    f \lp{x_S \mid\vec{\ID}_S = \vec{\id}_S }}dx_S
\end{align*}
These values need to be computed numerically since, as previously stated, the random variables that represent the target and spectators' inputs no longer exactly translate  to the input distribution itself.
Utilizing the same group configuration as specified above under Case 1 and assuming the probability for each identity is equally likely (for convenience), we compute the absolute loss in Figure~\ref{fig:mixed_normal_case2} alongside our baseline where every participant belongs to a single group. The most interesting observation is that all the curves overlap each other, a trend originally observed in Sections~\ref{sub:Discrete_Distributions}~and~\ref{sub:continuous_dist}. We note that this is likely a consequence of the experimental configuration itself (groups' salaries differ by $\pm10\%$, identity probabilities are equally likely).

\begin{figure*}[t]
    \centering
     \includegraphics[width=0.4\textwidth]{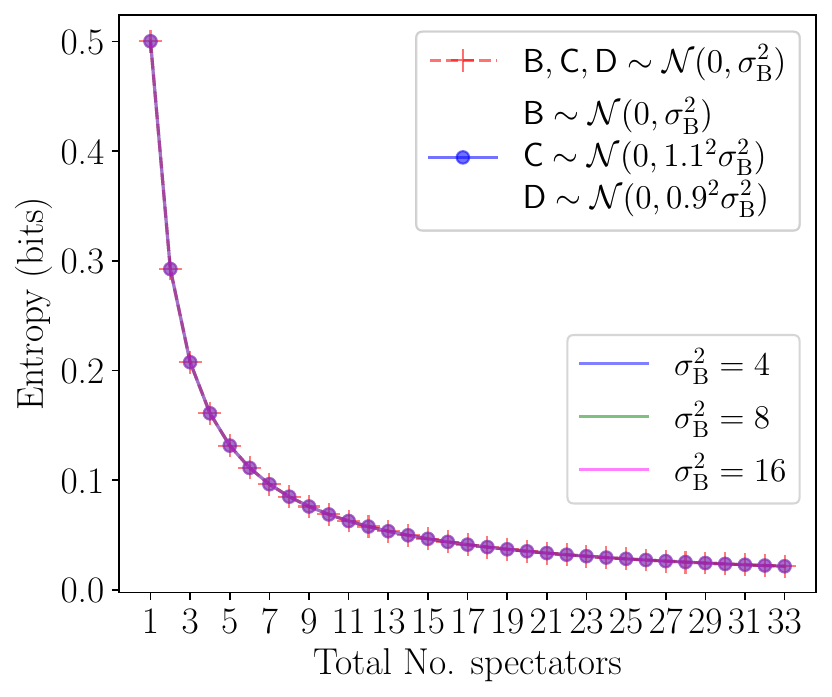}
   \caption{Mixed distribution analysis under Case 2, where the probability of an arbitrary participant belonging to any specific group is equally likely, i.e., $\Pr(\ID_P = \textsf{B})=\Pr(\ID_P = \textsf{C})=\Pr(\ID_P = \textsf{D}) = 1/3$.}
   \label{fig:mixed_normal_case2}
 \end{figure*}
 }
\end{caseof}

\fi

\section{Two Executions} 
\label{sec:jcon_entropy}

A natural generalization of the results of the prior section is to consider executing the average salary computation more than once. For example, after running the Boston gender pay gap study once, the same computation was executed the following year with an extended set of participants. In this case, if the time interval between the executions is small enough such that the inputs do not change between the executions or change minimally, one would expect that repeated participations would lead to additional information disclosure compared to a single execution. Thus, in this section we analyze the case of two executions and demonstrate their impact on the participants. We consider both the cases when a target contributes its input to both executions and when the target participates only in one of the executions and other takes place without the target, but on related inputs. Both cases result in additional information disclosure compared to a single execution, which we quantify in this section.

We partition the set of spectators $S$ into the following subsets: 
\begin{itemize}
\item spectators present only in the first execution $S_1 \subset S$, 
\item spectators present only in the second execution $S_2 \of S \setminus S_1$, 
\item and spectators present in both executions $S_{12} = S \setminus (S_1 \cup S_2)$. 
\end{itemize}
A person participating more than once (target or spectator) enters the same input into both execution.

When the target participates in both executions, we have:
\begin{gather*}
	O_1 =\sum\nolimits_{ i} X_{T_i} + \sum\nolimits_{i \in S_{12}} { X_{i}}+ \sum\nolimits_{i \in S_1} { X_{i}} =
	X_T + X_{S_{12}} + X_{S_1}\\
	O_2 =\sum\nolimits_{ i} X_{T_i} + \sum\nolimits_{i \in S_{12}} { X_{i}}+ \sum\nolimits_{i \in S_2} { X_{i}} = X_T + X_{S_{12}} + X_{S_2}.
\end{gather*}
The random variables $O_1$ and $O_2$ are \emph{not} independent, as they both are comprised of $X_T$ and $X_{S_{12}}$. 
We therefore want to compute the conditional entropy (using differential entropy notation):
\begin{equation} \label{eq:cond-ent}
	h(\vec{X}_T\mid  O_1, O_2) = h(\vec{X}_T, O_1, O_2) - h(O_1, O_2).
\end{equation}

\begin{claim} \label{claim3}
  The above conditional entropy can be expressed as
\begin{align} \label{eq:cond-ent-final}
	h(\vec{X}_T|O_1, O_2) & = h(\vec{X}_T) + h(X_{S_{12}} {+} X_{S_1} ,  X_{S_{12}}{+} X_{S_2} ) - h( O_1, O_2). 
\end{align}  
\end{claim}

	\newcommand\thisproof{
	Simplifying the first term of Equation~\ref{eq:cond-ent} using  the chain rule of entropy $h(X,Y) = h(X\mid Y) + h(Y)$ \cite{thomas2006elements}, we obtain:
	\begin{align*}
		h(\vec{X}_T, O_1, O_2) & = 	h(\vec{X}_T, X_T+ X_{S_{12}}+ X_{S_1}, X_T+ X_{S_{12}}+ X_{S_2})                                               \\
		& = h(\vec{X}_T) + h(X_T + X_{S_{12}}+ X_{S_1} \mid  \vec{X}_T)\ifCONF \breakcmd \else \fi
		+ h(X_T + X_{S_{12}}+ X_{S_2} \mid  X_T + X_{S_{12}}+ X_{S_1} ,\vec{X}_T).     
	\end{align*}
	{Using the fact that all participants' inputs are independent, we have:}
	\begin{align*}
		h(\vec{X}_T, O_1, O_2)  & = h(\vec{X}_T) + h(X_{S_{12}}+ X_{S_1}) 
		\ifCONF \breakcmd\qquad  \else \fi
		+ h(X_T + X_{S_{12}} + X_{S_2} , X_T + X_{S_{12}}+ X_{S_1} \mid \vec{X}_T) \ifCONF \breakcmd \qquad \else \fi
		 - h(X_T + X_{S_{12}} + X_{S_1} \mid \vec{X}_T)\\
						 & = h(\vec{X}_T) + h(X_{S_{12}}+ X_{S_1}) + h(X_{S_{12}} + X_{S_2} ,  X_{S_{12}}+ X_{S_1} ) \ifCONF \breakcmd\qquad \else\fi  - h(X_{S_{12}} + X_{S_1} )\\
						 & = h(\vec{X}_T) + h(X_{S_{12}} + X_{S_1} ,  X_{S_{12}}+ X_{S_2} ).
	\end{align*}
	The second term of Equation~\ref{eq:cond-ent} can be rewritten as:
	\begin{align*}
		h( O_1, O_2) & = h( X_T+ X_{S_{12}}+X_{S_1}, X_T+ X_{S_{12}}+X_{S_2}) \ifCONF \breakcmd \else \fi= h( X_T{+} X_{S_{12}}{+}X_{S_1}) +   h ({X_T {+} X_{S_{12}}{+}X_{S_2}}\mid  {X_T{+}X_{S_{12}}{+} X_{S_1}}),
	\end{align*}
	but cannot be simplified further. Therefore, the final expression of the conditional entropy is 
	\begin{align*} 
		h(\vec{X}_T | O_1, O_2) & = h(\vec{X}_T) + h(X_{S_{12}} + X_{S_1} ,  X_{S_{12}}{+} X_{S_2} ) - h( O_1, O_2). \qedhere
	\end{align*}
	}
	\noindent \ifCONF 
	The derivation can be found in the full version of the text.
	\else
	\begin{proof}
		\thisproof
	\end{proof}
\fi

In the special case when no spectators participate in both executions (i.e., $S_{12} = \emptyset$), the middle term simplifies to 
$h(X_{S_1})  + h( X_{S_2})$.

When the target participates only in one of the experiments, we define executions $O'_1$ and $O'_2$, which are the same as $O_1$ and $O_2$, respectively, except that the target's inputs are not included. For instance, $O'_1 = X_{S_{12}} + X_{S_1}$. The relevant entropies in that case are $h(\vec{X}_T|O'_1, O_2)$ and $h(\vec{X}_T|O_1, O'_2)$.

The above requires us to introduce the definition of joint entropy of correlated random variables. Now, the normal distribution stands out among those considered in Section~\ref{sec:single_calculation} as a suitable candidate for our analysis. The generalized multivariate normal distribution is well-studied and has a closed-form differential entropy, which we discuss next.

\subsection{Bivariate Normal Distributions}
\label{sub:BivariateandTrivariatenormalDistribution}

Evaluating Equation~\ref{eq:cond-ent-final} requires defining the differential entropy of a multivariate normal random variable. We then derive the necessary core parameters for our distributions and use them to compute the conditional entropy.

Let $X_i \sim \N(\mu_i, \sigma_i^2)$ be a single normal random variable as defined in Section~\ref{sec:definitions}. We define $\vec{X} = \lp{X_1, \dots, X_k}^\tran$ to be a general multivariate normal distribution of a $k$-dimensional random vector, with $\vec{X} \sim \N( \bm{\mu}, \bm{\Sigma})$. Here,  {$\bm{\mu} \in \R^k$} is the mean vector specified as
$\bm{\mu} = \text{E}[\vec{X}]  = \lp{\Ex{X_1}, \Ex{X_2} ,\dots, \Ex{X_k}}^{\tran}  
=\lp{\mu_1, \mu_2, \dots, \mu_k}^{\tran},$
and $\bm{\Sigma} \in \R^{k \times k}$ is the ${k \times k}$  covariance matrix with each element defined as 
$\Sigma_{i,j} = \Ex{(X_i - \mu_i)(X_j - \mu_j)} = \Cov{X_i, X_j}$.
The differential entropy of the multivariate normal distribution $\vec{X}$ is given by 
\ifCONF
$h(\vec{X}) = \frac12 \log \lp{\lp{2 \pi e}^{k} \det\bm{\Sigma}},$
\else
\begin{align*}
	h(\vec{X}) & = \frac12 \log \lp{\lp{2 \pi e}^{k} \det\bm{\Sigma}},
\end{align*}
\fi 
\cite[Chapter~8.4]{thomas2006elements} where $ \det\bm{\Sigma}$ is the determinant of the covariance matrix. 
The next step is to characterize our multivariate distributions and determine their covariance matrices.  We also derive their mean vectors which are used for intermediate results.

To compute the second and third terms of Equation~\ref{eq:cond-ent-final}, we formalize the bivariate distributions  $\vec{S} = (X_{S_{12}}+X_{S_1}, X_{S_{12}}+X_{S_2})^\tran$ and  $\vec{O} = \lp{O_1,O_2}^\tran$. 
We denote $\mu_P = \summ{ i}{} \mu_{P_i}$ and  $\sigma_P^2 = \summ{ i}{} \sigma_{P_i}^2$ as the sum of the means and standard deviations, respectively, of all participants within a group $P$. 
Note that the mean is absent from the formula for the differential entropy, and therefore we can safely assume all $\mu_i = 0$. 
Starting with $\vec{O}$, we invoke the linearity of the expectation for the mean vector:
\begin{align*}
	\bm{\mu}_{\vec{O}}
	 &= \begin{pmatrix}
		           \Ex{O_1} \\
		           \Ex{O_2}
	           \end{pmatrix}
	= \begin{pmatrix}
		  \Ex{X_T  {+} X_{S_{12}} {+} X_{S_1} } \\
		  \Ex{X_T  {+} X_{S_{12}} {+} X_{S_2}}
	  \end{pmatrix} 
	  =
	\begin{pmatrix}
		\mu_T {+}\mu_{S_{12}} {+} \mu_{S_1} \\
		\mu_T {+}\mu_{S_{12}} {+} \mu_{S_2}
	\end{pmatrix}
	=\begin{pmatrix}
		 \mu_1 \\
		 \mu_2
	 \end{pmatrix}.
\end{align*}
For the covariance matrix, using the properties  $\Cov{X, X} = \text{Var}\lb{X}$ $ = \sigma_X^2$ and $\Cov{X,Y} = \Cov{Y,X}$ yields
\begin{align*}
	\bm{\Sigma}_{\vec{O}} & =
	\begingroup 
	\setlength\arraycolsep{1pt}
	 \begin{pmatrix}
		                \Cov{O_1, O_1} & \Cov{O_1, O_2} \\
		                \Cov{O_2, O_1} & \Cov{O_2, O_2}
	                \end{pmatrix}
	=
	\begin{pmatrix} 
		\Var{O_1} & \Cov{O_1, O_2} \\
		\Cov{O_1, O_2} & \Var{O_2}
	\end{pmatrix} 
	\endgroup
	\\
	& 
	\begingroup 
	\setlength\arraycolsep{1pt}
				=
	\begin{pmatrix}
		\sigma_T^2 {+} \sigma_{S_{12}}^2 {+} \sigma_{S_1}^2 & \Cov{O_1, O_2}                           \\
		\Cov{O_1, O_2}                           & \sigma_T^2 {+} \sigma_{S_{12}}^2 {+} \sigma_{S_2}^2
	\end{pmatrix} =
	\begin{pmatrix}
		\renewcommand{\arraystretch}{0.6}
		\sigma_{1}^2   & {\text{Cov}[{O_1, O_2}] }\\
		\Cov{O_1, O_2} & \sigma_{2}^2
	\end{pmatrix}.
	\endgroup
\end{align*}
The expression for $\Cov{O_1, O_2}$ can be stated as follows:
\begin{claim} \label{claim4}
	\normalfont $\Cov{O_1, O_2} = \sigma_T^2 + \sigma_{S_{12}}^2$ if ${S_{12}}$ is non-empty,
	and $\Cov{O_1, O_2} $ $= \sigma_T^2$ otherwise.
\end{claim}
\newcommand\proofthree{
	\begin{equation*}
		\begin{split}
			\Cov{O_1, O_2}  &= \Ex{(O_1 - \mu_1)(O_2 - \mu_2)} \ifCONF\breakcmd \else \fi
			= \Ex{O_1O_2 - \mu_2 O_1 - \mu_1 O_2  + \mu_1 \mu_2 } \\
			&= \Ex{\lp{X_T {+} X_{S_{12}} + X_{S_1}}\lp{X_T +X_{S_{12}}+  X_{S_2}}} 
			\ifCONF \breakcmd\quad
			\else \breakcmd \quad\fi - \Ex{
			\mu_2 \lp{X_T + X_{S_{12}} +X_{S_1}}}  	\ifCONF\breakcmd\quad \else  \fi 
			\ifCONF  \else \fi-\Ex{\mu_1 \lp{X_T + X_{S_{12}} +X_{S_2}}} + \Ex{\mu_1 \mu_2} \\ 
			&= \Ex{X_T^2}  + \Ex{X_{S_{12}}^2} +2\Ex{X_T X_{S_{12}}} +\Ex{X_T X_{S_1}}  \ifCONF
			\breakcmd\qquad \else\fi+ \Ex{X_T X_{S_2}} + \Ex{X_{S_{12}} X_{S_1}} + \Ex{X_{S_{12}} X_{S_2}} \\
			&\qquad + {\Ex{X_{S_1} X_{S_2}}} -  \mu_2 (\overbrace{\Ex{X_T} + \Ex{X_{S_{12}}} +\Ex{X_{S_1}}}^{\mu_1} ) \ifCONF \breakcmd\qquad- \else - \fi \mu_1 (\underbrace{\Ex{X_T} + \Ex{X_{S_{12}}} +\Ex{X_{S_2}}}_{\mu_2}) + \mu_1 \mu_2
		\end{split}
	\end{equation*}
	Exploiting the definition of variance $ \Ex{X^2} = \sigma_X^2  + \mu_X^2$ and fundamental properties of expectation:
	\begin{align*}
		&\Cov{O_1, O_2} = \sigma_T^2 + \sigma_{S_{12}}^2 - \mu_1 \mu_2 \ifCONF \breakcmd\quad \else \fi+  \underbrace{ \mu_T^2  + \mu_{S_{12}}^2  + 2\mu_T\mu_{S_{12}} +\mu_T\mu_{S_1} + \mu_T\mu_{S_2} + \mu_{S_{12}}\mu_{S_1} + \mu_{S_{12}}\mu_{S_2} +\mu_{S_1}\mu_{S_2}}_{= \mu_1 \mu_2} \\
		               & \ifCONF\else\myquad[5]\fi= \sigma_T^2 + \sigma_{S_{12}}^2 .
	\end{align*}
	Clearly, if $S_{12}= \emptyset$, the above result simplifies to~${\text{Cov}\left[O_1, O_2\right] = \sigma_T^2}$. This result is intuitive since the covariance measures the strength of correlation between two random variables, and $O_1$ and $O_2$ are both comprised of $X_T$ and $X_{S_{12}}$.}
\ifCONF The proof is available in the full version of the text~\cite{baccarini2022understanding}.
\else
\begin{proof}
\proofthree
\end{proof}
\fi
The final parameters of the bivariate distribution $\vec{O}$ are
\begin{align*}
	\bm{\mu}_{\vec{O}}
	=\begin{pmatrix}
		 \mu_1 \\
		 \mu_2
	 \end{pmatrix},
	\bm{\Sigma}_{\vec{O}}
	=
	\begin{pmatrix}
		\sigma_{1}^2           & \sigma_T^2 +\sigma_{S_{12}}^2 \\
		\sigma_T^2 +\sigma_{S_{12}}^2 & \sigma_{2}^2
	\end{pmatrix}.
\end{align*}
Repeating this procedure for the spectator joint distribution $\vec{S}$ yields a similar set of parameters:
\begin{align*}
	\bm{\mu}_{\vec{S}}
	=\begin{pmatrix}
		 \mu_{S_{12}} +  \mu_{S_1}\\
		 \mu_{S_{12}} +  \mu_{S_2}
	 \end{pmatrix},
	\bm{\Sigma}_{\vec{S}}
	=
	\begin{pmatrix}
		\sigma_{S_{12}}^2  + \sigma_{S_1}^2         &\sigma_{S_{12}}^2 \\
		\sigma_{S_{12}}^2 & \sigma_{S_{12}}^2  + \sigma_{S_2}^2
	\end{pmatrix}.
\end{align*}
Equipped with expressions for $\bm{\Sigma}_{\vec{O}}$ and $\bm{\Sigma}_{\vec{S}}$, we are prepared to begin our experimental analysis of $h(X_T \mid O_1, O_2)$.

\subsection{Experimental Evaluation}
\label{sub:experimental_evaluation}

The above allows us to experimentally evaluate the target's entropy loss for when inputs are normally distributed. We use normal distribution $\N(0,4)$ 
to reasonably approximate the log-normal distribution with real data. Once again, $\norm{T} = 1$ for concreteness and we let $\norm{S_1} = \norm{S_2}$ in all experiments, i.e., the number of spectators is the same in both executions.

\begin{figure*}[t]
    \begin{small}
			\centering
			\begin{subfigure}[t]{0.32\textwidth} \centering
				\includegraphics[width=0.99\textwidth]{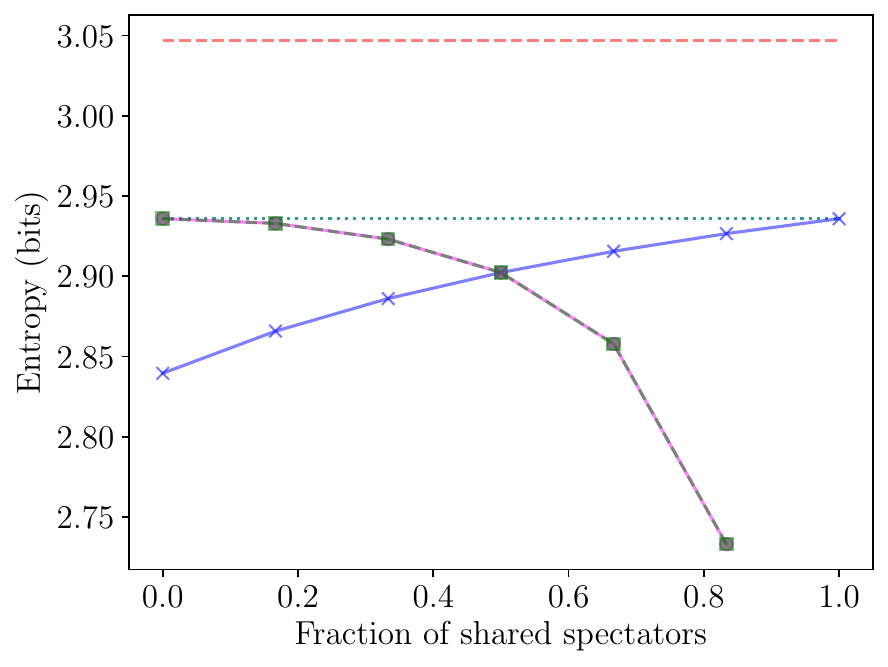}
				\caption{$n = 6$ spectators per execution.}
				\label{fig:one_vs_two_exp_mult_6}
			\end{subfigure}
				\begin{subfigure}[t]{0.32\textwidth} \centering
					\includegraphics[width=0.99\textwidth]{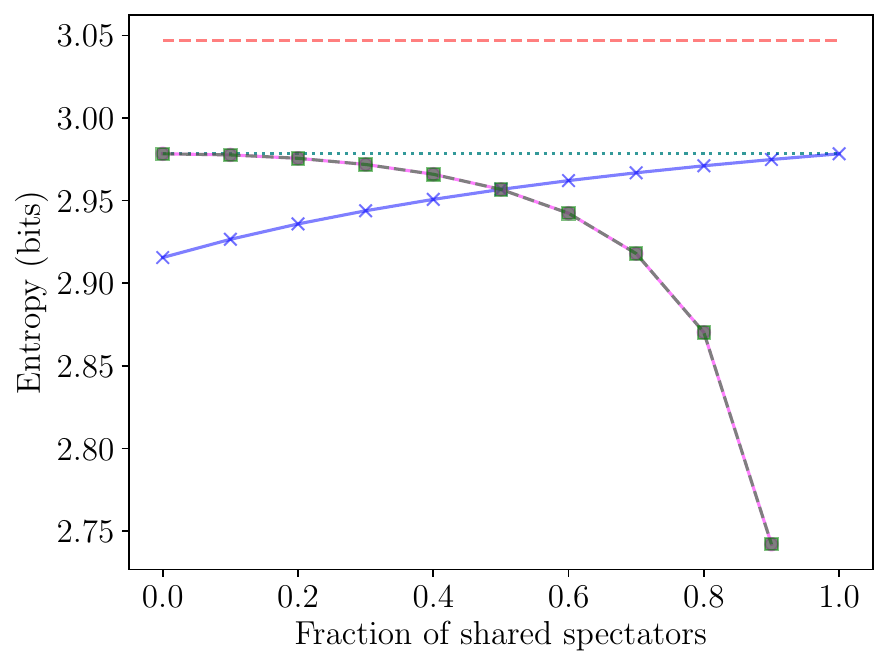}
					\caption{$n = 10$ spectators per execution.}
					\label{fig:one_vs_two_exp_mult_10}
				\end{subfigure}
				\begin{subfigure}[t]{0.32\textwidth} \centering
					\includegraphics[width=0.99\textwidth]{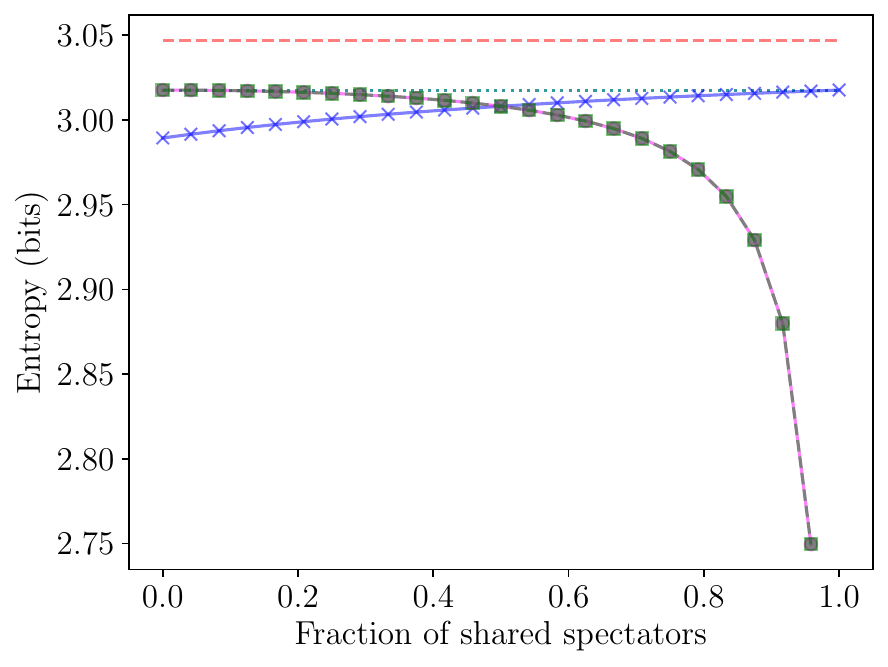}
					\caption{$n = 24$ spectators per execution.}
					\label{fig:one_vs_two_exp_mult_24}
				\end{subfigure}
\ifCONF
\begin{subfigure}[t]{.75\textwidth} \centering
\else 
\begin{subfigure}[t]{.85\textwidth} \centering
\fi
\includegraphics[width=1.0\textwidth]{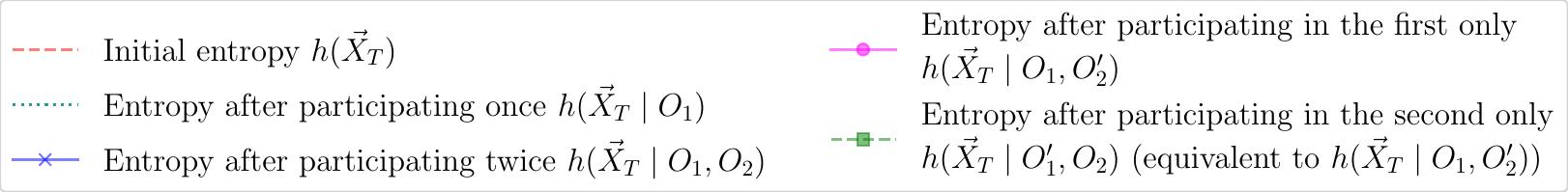}
				\end{subfigure}

        \caption{Target information loss after participating in one or two computations.
		Omitted: if the target participates in one experiment and all the shared spectators are reused, then $h(X_T \mid O_1, O_2') =0$.}
        \label{fig:one_vs_two_exp}
    \end{small}
\end{figure*}

It is informative to analyze information loss as the fraction of shared spectators changes and we do so for three different computation sizes. To be as close to the setup that guarantee 1\%--5\% entropy loss for the log-normal distribution (5--24 spectators), we choose to execute our experiments with 6, 10, and 24 spectators (where having an even number is beneficial for illustration purposes). This corresponds to the number of non-adversarial participants when the target is absent and the number of non-adversarial participants is one higher when the target is participating.

We display the following information in Figure~\ref{fig:one_vs_two_exp}:
\begin{itemize}
\item the target's initial entropy $h(\vec{X}_T)$,
\item the target's entropy after a single execution $h(\vec{X}_T \mid O_1)$, 
\item the target's entropy after participating twice $h(\vec{X}_T | O_1, O_2)$,
\item the target's entropy after participating in one of the two executions, i.e., $h(\vec{X}_T \mid O_1, O_2')$ and $h(\vec{X}_T \mid O'_1, O_2)$
\end{itemize}
and plot the values as a function of the fractional overlap between two executions for a given number of spectators.

Naturally, the value of $h(\vec{X}_T \mid O_1)$ remains constant when the number of participants is fixed. We observe that when participating twice, $h(\vec{X}_T \mid O_1, O_2)$ converges to $h(\vec{X}_T \mid O_1)$ as the fraction of shared spectators increases. This is expected because at 100\% overlap, we are functionally calculating $h(\vec{X}_T \mid O_1, O_1) = h(\vec{X}_T \mid O_1)$. 
\ifCONF
\else
We formalize this into the claim:
\begin{claim} \label{two_exp_two_part}
If the target participates in both evaluations and 100\% of the spectators are reused, $h(\vec{X}_T \mid O_1, O_2) = h(\vec{X}_T \mid O_1)$.
\end{claim}
\fi
\newcommand\prooftwoexptwopart{
We begin by analyzing the absolute loss
between the first and second evaluations when the target participates twice, namely:
\begin{align*}
 h(\vec{X}_T \mid O_1) -  h(\vec{X}_T\mid O_1, O_2).
  \end{align*}
Assume all participants' inputs are normally distributed ($X_i \sim \N(0, \sigma^2)$). Denote $p = \norm{P}$ as the size of an arbitrary group $P$ (e.g., $s_{12} = \norm{S_{12}}$), such that $X_P \sim \N(0, p\sigma^2)$. 
Simplifying the absolute loss between the first and second evaluations, we obtain:
\begin{align*}
  h(\vec{X}_T \mid O_1)  -  h(\vec{X}_T\mid O_1, O_2)  &= h(\vec{X}_T ) +  h(X_{S_{12}} + X_{S_1})-  h(X_{T} + X_{S_{12}} + X_{S_1})  \\ &\myquad[10]- \lp{  h(\vec{X}_T) + h(X_{S_{12}} + X_{S_1} ,  X_{S_{12}}+ X_{S_2} ) - h( O_1, O_2)} \\ 
  &=  h(X_{S_{12}} + X_{S_1})-  h(X_{T} + X_{S_{12}} + X_{S_1})  + h( O_1, O_2) - h(X_{S_{12}} + X_{S_1} ,  X_{S_{12}}+ X_{S_2} ).
\end{align*}
Using the definitions from Section~\ref{sub:BivariateandTrivariatenormalDistribution}, we calculate the remaining terms as 
\begin{align*}
   h(X_{S_{12}} + X_{S_1})  +h(X_{T} + X_{S_{12}} + X_{S_1}) &= \frac12 \log \lp{\frac
  {s_{12} +s_1}
  {t +s_{12} + s_1}
  }\\
  h( O_1, O_2) &= \frac12 \log \lp{(2\pi e)^2 ((t + s_{12})(s_1 + s_2) + s_1 s_2) \sigma^2} \\
  h(X_{S_{12}} + X_{S_1} ,  X_{S_{12}}+ X_{S_2} ) &= \frac12 \log \lp{(2\pi e)^2 (s_{12} (s_1 + s_2) + s_1 s_2) \sigma^2} \\
   h( O_1, O_2) - h(X_{S_{12}} + X_{S_1} ,  X_{S_{12}}+ X_{S_2} )  &= \frac12\log \lp{
    \frac
    {(t + s_{12})(s_1 + s_2) + s_1 s_2}
    {s_{12} (s_1 + s_2) + s_1 s_2}
  }
\end{align*}
Therefore, the absolute entropy loss between the first and second evaluations is
\begin{align*}
h(\vec{X}_T \mid O_1) -  h(\vec{X}_T\mid O_1, O_2) &= \lp{
  \frac12 \log \lp{\frac
  {s_{12} +s_1}
  {t +s_{12} + s_1}
  }
} 
+ \frac12\log \lp{
  \frac {(t + s_{12})(s_1 + s_2) + s_1 s_2}
  {s_{12} (s_1 + s_2) + s_1 s_2}
}
   \\
  &=
  \frac12\log \lp{\lp{\frac
  {s_{12} +s_1}
  {t +s_{12} + s_1}
  } \lp{
  \frac {t(s_1 + s_2)+ s_{12}(s_1 + s_2) + s_1 s_2}
  {s_{12} (s_1 + s_2) + s_1 s_2}
}}.
\end{align*}
Since we assume $s_1 = s_2$, the above expression simplifies to 
\begin{align*}
h(\vec{X}_T \mid O_1) &-  h(\vec{X}_T\mid O_1, O_2) =
  \frac12\log \lp{\lp{\frac
  {s_{12} +s_1}
  {t +s_{12} + s_1}
  } \lp{
  \frac {2t+ 2s_0+ s_1}
  {2s_0  + s_1 }
}}.
\end{align*}
This function is monotonically decreasing when the total number of spectators is fixed to $s_{12} + s_1$, and we vary the ratio  $\frac{s_{12}}{s_1} \in \lb{0,1}$, which is consistent with our observation that the absolute loss will converge to $h(\vec{X}_T \mid O_1)$.
}
\ifCONF
\else
\begin{proof}
	\prooftwoexptwopart
\end{proof}
\fi
Conversely, increasing the fraction of the overlap has the inverse effect for $h(\vec{X}_T \mid O_1, O_2')$, causing it to trend downward. At 100\% overlap, $h(\vec{X}_T \mid O_1, O_2') = 0$ (point omitted from the plots). This is a consequence of effectively computing  $h(\vec{X}_T \mid O_1, X_{S_{12}})$:
\begin{align*}
	h(\vec{X}_T | O_1, X_{S_{12}}) &= h(\vec{X}_T, O_1, X_{S_{12}}) {-} h(O_1, X_{S_{12}})\\
	&= h(\vec{X}_T) {+} h(X_{S_{12}}) {-} \lp{h({X}_T {+} X_{S_{12}} \mid X_{S_{12}}) {+} h(X_{S_{12}})}\\
	&{=} h(\vec{X}_T) {+} h(X_{S_{12}}) {-} \lp{h({X}_T) {+} h(X_{S_{12}})} {=} h(\vec{X}_T) {-} h({X}_T).
\end{align*}
When $\norm{T} = 1$, then $h(\vec{X}_T) = h({X}_T)$, thus reducing the above equation to zero. This informs us that the output of the second computation $O_2'$ without any unique spectators reveals the target's information entirely.
\ifCONF
\else
We state this observation as follows:
\begin{claim} \label{two_exp_one_part}
If the target participates in one evaluation and 100\% of the spectators are reused, $h(\vec{X}_T \mid O_1, O_2') = 0$.
\end{claim}
\fi
\ifCONF
We analytically prove these observations
in the full version of the text by deriving exact expressions for the absolute entropy loss.
\else

\newcommand\prooftwoexponepart{
Next, we examine the absolute entropy loss between the first and second evaluations when the target participates in only the first evaluation:
\begin{align*}
h(\vec{X}_T \mid O_1) -  h(\vec{X}_T\mid O_1, O_2')
\end{align*}
The only difference from the prior calculation arises is replacing $h(O_1, O_2)$ with $h(O_1, O_2')$, which evaluates to 
\begin{align*}
  h( O_1, O_2) &= \frac12 \log \lp{(2\pi e)^2 (t(s_{12} + s_2) + s_{12}(s_1 + s_2) + s_1 s_2) \sigma^2},
\end{align*}
such that our final expression is 
\begin{align*}
h(\vec{X}_T \mid O_1) -  h(\vec{X}_T\mid O_1, O_2') 
    &=  \frac12\log \lp{\lp{\frac{s_{12} +s_1}{t +s_{12} + s_1}} \lp{
  \frac{t(s_{12} + s_2)+ s_{12}(s_1 + s_2) + s_1 s_2}
  {s_{12} (s_1 + s_2) + s_1 s_2}
}} \\ 
&=  \frac12\log \lp{\lp{\frac
{s_{12} +s_1}
{t +s_{12} + s_1}
} \lp{
\frac {t(s_{12} + s_1) + s_1(2s_0+ s_1)}
{s_1(2s_0  + s_1) }
}}.
\end{align*}
This function equals to infinity when $s_1 = s_2 = 0$, which is confirms that the output of the second computation $O_2'$ without the presence of any unique spectators reveals the target's information entirely. 
}
\ifCONF
\else
\begin{proof}
	\prooftwoexponepart
\end{proof}
\fi

A passive result of both proofs is that all forms of absolute loss are parameter-independent, which is consistent with Claim~\ref{claim2}. 

\fi

Our next observation pertains to the
point of intersection where $h(\vec{X}_T \mid O_1, O_2) = h(\vec{X}_T \mid O_1, O_2')$, which  occurs when 50\% of the spectators are shared across the computation. This appears for the special case when the total number of spectators in a single evaluation is even. Concretely, 
we compare 
\begin{align}
	\begin{split}
		h(\vec{X}_T{\mid}  O_1, O_2) & = h(\vec{X}_T, O_1, O_2) - h(O_1, O_2),\\
		h(\vec{X}_T{\mid}  O_1, O_2') & = h(\vec{X}_T, O_1, O_2') - h(O_1, O_2').
	\end{split}	
	\label{eq:overlap_eqs}
\end{align}
It can be shown using the procedure outlined in Section~\ref{sec:jcon_entropy} that  $h(\vec{X}_T, O_1, O_2) = h(\vec{X}_T, O_1, O_2')$. Therefore, we prove 
the following: 
\begin{claim} \label{claim5}
  With normally distributed inputs, the terms $ h(O_1, O_2)$ and $ h(O_1, O_2')$ are equal when $|S_{12}| = |S_1|$.
\end{claim}

	\newcommand\prooffive{
	Following the steps used to derive the covariance matrix of $\vec{O} = (O_1, O_2)$,  the covariance matrix of $\vec{O'} = (O_1, O_2')$  is  
  \begin{align*}
	  \bm{\Sigma_{\vec{O'}}}
	  =
	  \begin{pmatrix}
		  \sigma_{T}^2  + \sigma_{S_{12}}^2  + \sigma_{S_1}^2         &\sigma_{S_{12}}^2 \\
		  \sigma_{S_{12}}^2 & \sigma_{S_{12}}^2  + \sigma_{S_2}^2
	  \end{pmatrix}.
  \end{align*}
  Recall that the differential entropy of the multivariate normal is $h(\vec{X}) = \frac12 \log \lp{\lp{2 \pi e}^{k} \det\bm{\Sigma}}$. The sole object of interest is the $\det\bm{\Sigma}$ term, as the remainder contribute a constant factor. We calculate 
  \begin{align*}
	  \det \bm{\Sigma_{\vec{O}}}
	   \ifCONF \eqanchor \else \fi = 
	  (\sigma_{T}^2  + \sigma_{S_{12}}^2  + \sigma_{S_1}^2 )(\sigma_{T}^2  + \sigma_{S_{12}}^2  + \sigma_{S_2}^2 ) - (\sigma_{T}^2  + \sigma_{S_{12}}^2 )^2 
	  \ifCONF \breakcmd \else \fi
	  = \sigma_{T}^2 (\sigma_{S_1}^2  + \sigma_{S_2}^2 ) +  \sigma_{S_{12}}^2 (\sigma_{S_1}^2  + \sigma_{S_2}^2 ) +  \sigma_{S_1}^2  \sigma_{S_2}^2.
  \end{align*}
  Similarly,
  \begin{align*}
	  \det \bm{\Sigma_{\vec{O'}}} 
	\ifCONF \eqanchor \else \fi = 
	(\sigma_{T}^2  + \sigma_{S_{12}}^2  + \sigma_{S_1}^2 )(\sigma_{S_{12}}^2  + \sigma_{S_2}^2 ) - \sigma_{S_{12}}^4 
	\ifCONF \breakcmd \else \fi
	  = \sigma_{T}^2 (\sigma_{S_{12}}^2  + \sigma_{S_2}^2 ) +  \sigma_{S_{12}}^2 (\sigma_{S_1}^2  + \sigma_{S_2}^2 ) +  \sigma_{S_1}^2  \sigma_{S_2}^2.
  \end{align*}
  Therefore, the equality  $h(\vec{X}_T \mid O_1, O_2) = h(\vec{X}_T \mid O_1, O_2')$  is satisfiable if and only if $\sigma_{S_{12}}^2 = \sigma_{S_1}^2$, which occurs when $\norm{S_{12}} = \norm{S_1}$.
  }
  \ifCONF 
  \begin{proof}
	\prooffive
\end{proof}
  \else
\begin{proof}[Proof of Claim~\ref{claim5}]
	\prooffive
\end{proof}
\fi

\ignore{If a target is aware of the percentage of overlapping spectators, they can determine their best course of action by calculating the information loss after the second execution (with and without their participation), relative to the loss after the first execution. Specifically,
\begin{align*}
\frac{h(\vec{X}_T) -  h(\vec{X}_T\mid O_1) - (h(\vec{X}_T \mid O_1) -  h(\vec{X}_T\mid O_1, O_2))}{h(\vec{X}_T) -  h(\vec{X}_T\mid O_1)},
\end{align*} 
and
\begin{align*}
	\frac{h(\vec{X}_T) -  h(\vec{X}_T\mid O_1) - (h(\vec{X}_T \mid O_1) -  h(\vec{X}_T\mid O_1, O_2'))}{h(\vec{X}_T) -  h(\vec{X}_T\mid O_1)}.
\end{align*}
}
As computation designers, we can minimize information disclosure for all participants by targeting 50\% participants' overlap between the first and second executions. For the configurations in Figure~\ref{fig:one_vs_two_exp},
at 50\% overlap, the percentages of information loss from the second evaluation relative to the first evaluation are comparable for the selected number of spectators $n$ (30.18\% for $n=6$, 31.3\% for $n=10$, and 32.45\% for $n=24$). This corresponds to the intersection points in Figure~\ref{fig:one_vs_two_exp}. 

As we may be unable to guarantee that exactly 50\% of participants overlap between two executions, we can increase our tolerance for entropy loss by inviting more participants and building a buffer to accommodate overlaps in a range, e.g., 40--60\%. Using data in Figure~\ref{fig:one_vs_two_exp}, this information is convenient to gather for $n=10$. That is,
if we increase the fraction of overlapping spectators, single-participation targets are most at risk. The converse is true if the overlap decreases -- the target suffers less exposure from participating one evaluation. Table~\ref{tab:loss_ratio} summarizes the results. 
\begin{table}[t]
\centering
\begin{tabular}[c]{|c|c|c|c|} \hline
Number of evaluations   &\multicolumn{3}{c|}{Spectator overlap}\\ \cline{2-4}
the target participates in & 40\% & 50\% & 60\% \\ \hline
One &18.0\% &31.3\% & 52.3\% \\
Two &40.1\% &31.3\% & 23.5\% \\ \hline
\end{tabular}
\caption{Percentage of information loss after two executions relative to a single execution for $n = 10$.}
\label{tab:loss_ratio}
\end{table}
This means that performing two executions in the worst case costs a participant entropy loss 1.5 times higher than if only a single computation is executed. As a result, with the target entropy loss of 5\% and 1\%, we need to increase the number of spectators from 5 and 24 to 7 and 33, respectively. 

We note that our analysis of repeated executions applies only when the inputs of the participants in the overlapping set of participants do not change. And if the executions are distant enough in time that the participants' inputs significantly change, they would no longer be treated as repeated dependent executions.

\ifCONF 
In the full version of the text~\cite{baccarini2022understanding} we conduct additional two-evaluation experiments, such as adjusting the number of shared spectators while maintaining a fixed  number of unique spectators.
\else
\subsection{Additional Two Executions Experiments}
\label{sub:two-exec-exp}

\begin{figure*}
\centering
\begin{subfigure}[t]{0.45\textwidth} \centering
	\includegraphics[width=0.9\textwidth]{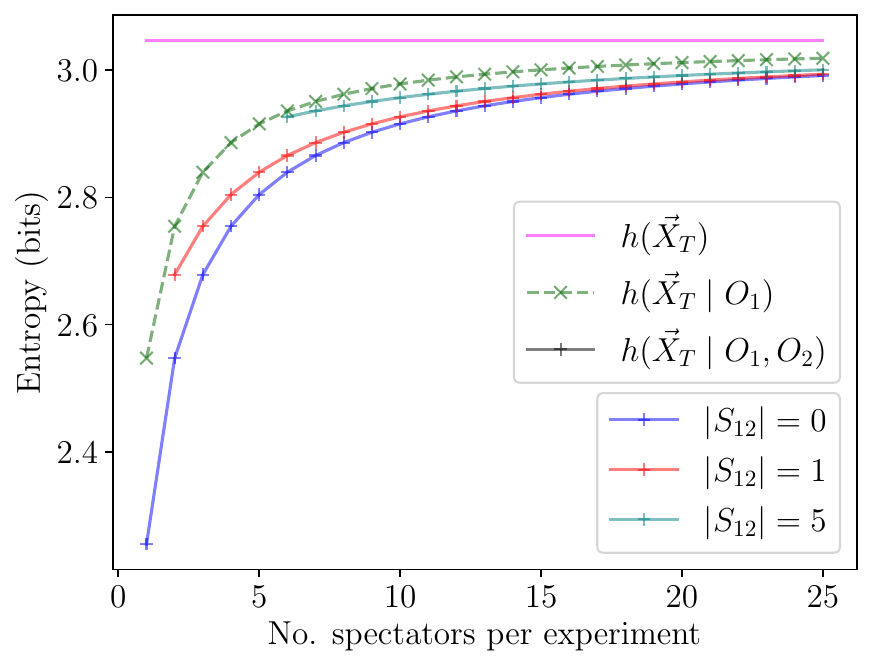}
    \caption{Computing $h(X_T)$, $h(X_T\mid O_1)$, and $h(X_T \mid O_1, O_2)$ for several $\norm{S_{12}}$ sizes.}
    \label{fig:joint_sum_normal_vary_shared_and_unique_a}
\end{subfigure}
    \begin{subfigure}[t]{0.45\textwidth} \centering
	\includegraphics[width=0.9\textwidth]{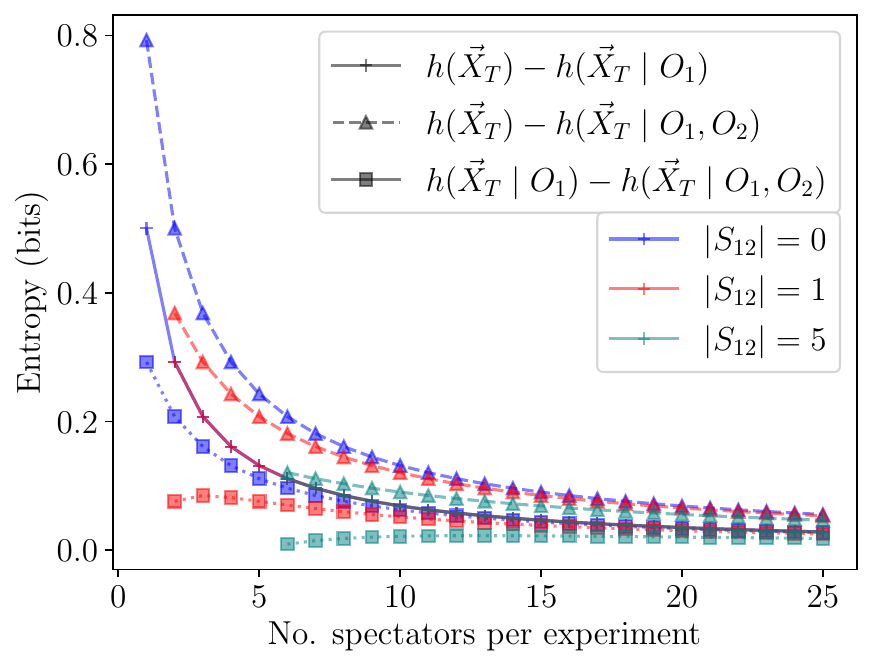}
	\caption{Absolute entropy loss.}
        \label{fig:loss_joint_sum_normal_vary_shared_and_unique_16}
    \end{subfigure}

    \caption{Comparing the relative and absolute entropy losses of participants with normally distributed inputs. The number of spectators per experiment on the $x$-axis is computed as $\norm{S_{12} \cup S_{1}} = | S_{12} \cup S_2|$, starting with $\norm{S_1} = \norm{S_2} = 1$.}
	\label{fig:joint_sum_normal_vary_shared_and_unique}
\end{figure*}

We examine the impact of shared spectators' presence on target's information loss. In Figure~\ref{fig:joint_sum_normal_vary_shared_and_unique_a}, we plot: 
\begin{itemize}
\item the target's initial entropy $h(\vec{X}_T)$,
\item the target's entropy after a single execution $h(\vec{X}_T \mid O_1)$, and
\item the target's entropy after two executions $h(\vec{X}_T \mid O_1, O_2)$ with a different number of spectators participating in both executions.
\end{itemize}
We vary the total number of spectators per experiment $\norm{S_{12} \cup S_{1}} = | S_{12} \cup S_2|$ on the $x$-axis, starting with one unique spectator per experiment $\norm{S_1} = \norm{S_2} = 1$. 
The $h(\vec{X}_T \mid O_1, O_2)$ curves correspond to awae after two executions and start when the number of participants reaches their respective number of shared spectators $\norm{S_{12}}$ in order to make an accurate comparison. A single curve for $h(\vec{X}_T \mid O_1)$ suffices since it does not use the notion of shared spectators. 

We observe in Figure~\ref{fig:joint_sum_normal_vary_shared_and_unique_a} that the larger the number of shared spectators for a given $\norm{S_0}$  is, the less information is revealed about the target. These spectators function as ``noise'' that protects the target. The protection offered by a small number of shared spectators becomes less pronounced as the number of participants grows. 

We also compute and present in Figure~\ref{fig:loss_joint_sum_normal_vary_shared_and_unique_16} the target's absolute entropy loss for the following experiments:
\begin{itemize}
\item after a single execution $h(\vec{X}_T) - h(\vec{X}_T \mid O_1)$,
\item after two executions $h(\vec{X}_T) - h(\vec{X}_T \mid O_1, O_2)$, and 
\item after the second execution $h(\vec{X}_T \mid O_1) - h(\vec{X}_T \mid O_1, O_2)$
\end{itemize}
using a varying number of shared spectators $|S_{12}|$. We see that for each fixed number of shared spectators $\norm{S_{12}}$, the absolute loss as a result of the first participation ($h(\vec{X}_T) - h(\vec{X}_T \mid O_1)$) is greater than the absolute loss of the second participation ($h(\vec{X}_T \mid O_1) - h(\vec{X}_T \mid O_1, O_2)$). With no shared spectators the curves converge at about 15 participants per execution, while increasing the number of shared spectators causes the curves to converge at a slower rate. 
\fi
\ifCONF
Furthermore, we extend our analysis to three or more executions.
\else \fi 

\ifCONF 
\else 
\subsection{Mixed Distribution Parameters for Two Executions}
\label{sub:mixed_distribution_parameters_two_eval}

In Section~\ref{sub:mixed_distribution_parameters}, we examined how our conclusions changed when we generalize our analysis to non-identically distributed participant inputs. We are similarly interested in how this affects our two-execution analysis.
Recalling the formulation of our two-evaluation setting where spectators are present in the first, second, or both executions, we now consider these spectator subsets can be further partitioned into sub-subsets based on their group identities.
Combining our two-evaluation notation and that of Section~\ref{sub:mixed_distribution_parameters}, we define $ P_{k,\GG} \of P$ as the set of participants present in execution(s) $k = \lbr{1, 2, 12}$ belonging to group $\GG \in \grp$. For example, $k = 1$ means participation in the first execution, while $k = 12$ means participation in both.
The sum of these participants' inputs is 
\begin{align*}
    X_{P_{k}} = \sum_{\GG \in \grp} \sum_{i \in P_{k,\GG}} X_{P_{k,i}} =  \sum_{\GG \in \grp} X_{P_{k,\GG}},
\end{align*}
where  $ X_{P_{\GG}} =   \sum_{i \in P_{k,\GG}} X_{P_{k,i}}$. 

\medskip\noindent
\textbf{Optimal setup for minimizing information disclosure.}
In Section~\ref{sub:experimental_evaluation}, we determined that the point of intersection  of the entropies $ h(\vec{X}_T\mid O_1, O_2 )$ and $ h(\vec{X}_T\mid O_1, O_2' )$ at 50\% participant overlap
provides the best level of protection for all types of targets, and moving in either direction (increasing or decreasing the overlap) causes the leakage to increase. We revisit our proof of Claim~\ref{claim5} in our generalized setting to determine whether the equality $h(\vec{X}_T\mid O_1, O_2 )= h(\vec{X}_T \mid O_1, O_2')$ still holds when 50\% of the spectators are shared across the computation, i.e., $|S_{12}| = |S_1|$.

Considering the expansions of  $ h(\vec{X}_T\mid O_1, O_2 ) $ and $h(\vec{X}_T \mid O_1, O_2')$ in Equation~\ref{eq:overlap_eqs}, it can be shown that the terms $h(\vec{X}_T, O_1, O_2)$ and $h(\vec{X}_T, O_1, O_2')$ are equal, leaving  $ h(O_1, O_2)$ and $ h(O_1, O_2')$ for us to compare. We perform this analysis under Case 1 (participant group identities are known), since the proof of Claim~\ref{claim5} relies on the existence of the closed-form expression for the differential entropy.

Computing $ h(O_1, O_2)$ and $ h(O_1, O_2')$ in the generalized setting requires re-formalizing the covariance matrices of the joint random variables $\vec{O} = (O_1, O_2)$ and $\vec{O'} = (O_1, O_2')$.  Performing the steps outlined in Section~\ref{sub:BivariateandTrivariatenormalDistribution} yields the following:
\begin{align*}
	\bm{\Sigma}_{\vec{O}}
	&= 
	\begin{pmatrix}
		\sum_{\GG \in \grp}\sigma_{\GG}^2\lp{\norm{T_{\GG}} {+}  \norm{S_{12,\GG}} {+} \norm{S_{1,\GG}}   } & \sum_{\GG \in \grp}\sigma_{\GG}^2\lp{ \norm{T_{\GG}} {+} \norm{S_{12,\GG}}} \\
		\sum_{\GG \in \grp}\sigma_{\GG}^2\lp{\norm{T_{\GG}} {+}\norm{S_{12,\GG}}} & \sum_{\GG \in \grp}\sigma_{\GG}^2\lp{\norm{T_{\GG}} {+}\norm{S_{12,\GG}} {+} \norm{S_{2,\GG}}   }
	\end{pmatrix} \\
	\bm{\Sigma}_{\vec{O'}}
	&= 
	\begin{pmatrix}
		\sum_{\GG \in \grp}\sigma_{\GG}^2\lp{\norm{T_{\GG}}  {+}  \norm{S_{12,\GG}} {+} \norm{S_{1,\GG}}   } & \sum_{\GG \in \grp}\sigma_{\GG}^2\norm{S_{12,\GG}} \\
		\sum_{\GG \in \grp}\sigma_{\GG}^2\norm{S_{12,\GG}} & \sum_{\GG \in \grp}\sigma_{\GG}^2\lp{ \norm{S_{12,\GG}} {+} \norm{S_{2,\GG}}   }
	\end{pmatrix}.
\end{align*}
Given the definition of the differential entropy of a multivariate normal as $h(\vec{X}) = \frac12 \log \lp{\lp{2 \pi e}^{k} \det\bm{\Sigma}}$, we compute the  determinants of the above matrices as:
\begin{align*}
	\det \bm{\Sigma_{\vec{O}}} &=
	\lp{\sum_{\GG \in \grp}\sigma_{\GG}^2\lp{\norm{T_{\GG}}  {+}  \norm{S_{12,\GG}} {+} \norm{S_{1,\GG}}  }} \lp{\sum_{\GG \in \grp}\sigma_{\GG}^2\lp{\norm{T_{\GG}}  {+}  \norm{S_{12,\GG}} {+} \norm{S_{2,\GG}}   }}
	- \lp{\sum_{\GG \in \grp}\sigma_{\GG}^2\lp{ \norm{T_{\GG}} {+} \norm{S_{12,\GG}}} }^2 \\ 
	&=	
	\sum_{\GG \in \grp}\sigma_{\GG}^2 \norm{T_{\GG}} \lp{ \sum_{\GG \in \grp}\lp{\sigma_{\GG}^2\lp{\norm{S_{1,\GG}} {+} \norm{S_{2,\GG}}   }} } 
+	\sum_{\GG \in \grp}\sigma_{\GG}^2 \norm{S_{(12,\GG)}} \lp{ \sum_{\GG \in \grp}\lp{\sigma_{\GG}^2\lp{\norm{S_{1,\GG}} {+} \norm{S_{2,\GG}}   }} } \\ 
&\myquad[4]
 + \sum_{\GG \in \grp}\sigma_{\GG}^2\lp{  \norm{S_{1,\GG}}}	\sum_{\GG \in \grp}\sigma_{\GG}^2\lp{  \norm{S_{2,\GG}}}\\ 
	\det \bm{\Sigma_{\vec{O'}}} &=
	\lp{\sum_{\GG \in \grp}\sigma_{\GG}^2\lp{\norm{T_{\GG}}  {+}  \norm{S_{12,\GG}} {+} \norm{S_{1,\GG}}  }} \lp{\sum_{\GG \in \grp}\sigma_{\GG}^2\lp{\norm{S_{12,\GG}} {+} \norm{S_{1,\GG}}   }}
	- \lp{\sum_{\GG \in \grp}\sigma_{\GG}^2\lp{ \norm{S_{12,\GG}}} }^2\\ 
	&=	
	\sum_{\GG \in \grp}\sigma_{\GG}^2 \norm{T_{\GG}} \lp{ \sum_{\GG \in \grp}\lp{\sigma_{\GG}^2\lp{\norm{S_{12,\GG}} {+} \norm{S_{2,\GG}}   }} } 
+	\sum_{\GG \in \grp}\sigma_{\GG}^2 \norm{S_{(12,\GG)}} \lp{ \sum_{\GG \in \grp}\lp{\sigma_{\GG}^2\lp{\norm{S_{1,\GG}} {+} \norm{S_{2,\GG}}   }} } \\ 
&\myquad[4]
 + \sum_{\GG \in \grp}\sigma_{\GG}^2\lp{  \norm{S_{1,\GG}}}	\sum_{\GG \in \grp}\sigma_{\GG}^2\lp{  \norm{S_{2,\GG}}}\\ 
\end{align*}
By comparing the equations, we obtain that the equality  $h(\vec{X}_T \mid O_1, O_2) = h(\vec{X}_T \mid O_1, O_2')$  is satisfiable if and only if $\sum_{\GG \in \grp}\lp{ \sigma_{\GG}^2 \norm{S_{12,\GG}}} = \sum_{\GG \in \grp}\lp{ \sigma_{\GG}^2 \norm{S_{1,\GG}}}$. However, this no longer implies that the optimal configuration is at 50\% overlap. For $\norm{\grp}>1$, there can be multiple solutions with respect to the individual group sizes, statistical parameters, and overlap percentages such that the equality can be satisfied.

\fi

\newcommand\threeExec{

\section{Three Executions and Beyond}
\label{sec:three_executions_and_beyond}

\ifCONF
In this section, we further generalize our analysis to three executions.
\else
The next logical step is to further generalize our analysis to three and any number $M$ executions. 
\fi

\subsection{Three Executions}
\label{sub:three_executions}

For three evaluations, there are additional possibilities for spectators to overlap between experiments. Specifically, we have:
\begin{itemize}
	\item spectators who participate in one experiment	($S_1, S_2, S_3$), 
	\item spectators who participate in two experiments, but not a third (${S_{12}}, {S_{13}}, {S_{32}}$), and
	\item spectators who participate in all three experiments $({S_{123}} )$.
\end{itemize}
Let $n$ be the (fixed) total number of spectators per experiment. For each evaluation, let superscript $(\tau_i)$ denote a target's participation flag defined as:
\begin{align*}
	\tau_i = \begin{cases}
		0 & \text{ $T$ does \emph{not} participate in evaluation $i$}\\
		1 & \text{ $T$ participates  in evaluation $i$}
	\end{cases}.
\end{align*}
We require $\summ{i = 1}{3} \tau_i > 0$ to signify that the target participates at least once. Therefore, there are $2^3 - 1 = 7$ possible target configurations. For example, $(\tau_1,\tau_2,\tau_3) = (1,0,1)$ means the target participated in the first and third executions. We use this notation to generate expressions for all configurations of the targets' participation in evaluations.
The random variables for each evaluation are:
\begin{gather*}
	O_1^{(\tau_1)} = \tau_1 \mdot X_T + X_{S_1} + X_{S_{12}}+ X_{S_{13}}  + X_{S_{123}}   = \tau_1 \mdot X_T +  X_{\hat{S}_1}\\
	O_2^{(\tau_2)} = \tau_2 \mdot X_T + X_{S_2} + X_{S_{12}}+ X_{S_{23}}  + X_{S_{123}}   = \tau_2 \mdot X_T + X_{\hat{S}_2} \\
	O_3^{(\tau_3)} = \tau_3 \mdot X_T + X_{S_2} + X_{S_{23}}+ X_{S_{13}}  + X_{S_{123}}  = \tau_3 \mdot X_T + X_{\hat{S}_3} ,
\end{gather*}
where $X_{\hat{S}_i}$ is the sum of all spectator configurations in evaluation $i$.
If we denote $p = \norm{P}$ as the size of an arbitrary group $P$ such that $X_P \sim \N(0, p\sigma^2)$, then covariance matrix for the random vector $\vec{O}_{1,2,3} = \lp{ O_1^{(\tau_1)}, O_2^{(\tau_2)}, O_3^{(\tau_3)} }^\tran$ is 
\begin{align*}
	\ifCONF  \else \fi 
	\bm{\Sigma}_{\vec{O}_{1,2,3}}  =
	\eqanchor
	\begingroup 
\setlength\arraycolsep{1pt}
	\begin{pmatrix}
\Cov{O_1^{(\tau_1)}, O_1^{(\tau_1)}} & \Cov{O_1^{(\tau_1)}, O_2^{(\tau_2)}} & \Cov{O_1^{(\tau_1)}, O_3^{(\tau_3)}} \\
\Cov{O_2^{(\tau_2)}, O_1^{(\tau_1)}} & \Cov{O_2^{(\tau_2)}, O_2^{(\tau_2)}} & \Cov{O_2^{(\tau_2)}, O_3^{(\tau_3)}} \\
\Cov{O_3^{(\tau_3)}, O_1^{(\tau_1)}} & \Cov{O_3^{(\tau_3)}, O_2^{(\tau_2)}} & \Cov{O_3^{(\tau_3)}, O_3^{(\tau_3)}}
\end{pmatrix} 
\endgroup
\ifCONF \breakcmd \else  \fi 
\begingroup 
\setlength\arraycolsep{1pt}
= \begin{pmatrix}
	\begin{pmatrix}	
	\tau_1\mdot{t} {+}  {s_{1}} {+} {s_{12}} \\
	 {+}{s_{13}} {+}{s_{123}} 
\end{pmatrix}
	&\tau_1\tau_2\mdot{t} {+} {s_{12}} {+}{s_{123}}  & \tau_1\tau_3\mdot{t} {+} {s_{13}} {+}{s_{123}} \\
	\tau_1\tau_2\mdot{t} {+} {s_{12}} {+}{s_{123}} & 
		\begin{pmatrix}	
	\tau_2\mdot{t} {+}  {s_{2}} {+} {s_{12}}\\ 
	 +{s_{23}} {+}{s_{123}}  
	\end{pmatrix}
	& \tau_2\tau_3\mdot{t} {+} {s_{23}} {+}{s_{123}}\\ 
	\tau_1\tau_3\mdot{t} {+} {s_{13}} {+}{s_{123}} &  \tau_2\tau_3\mdot{t} {+} {s_{23}} {+}{s_{123}}  & 
	\begin{pmatrix}	
		\tau_3\mdot{t} {+}  {s_{3}} {+} {s_{23}} \\ 
		+{s_{13}} {+}{s_{123}} 
	\end{pmatrix}
\end{pmatrix} \sigma^2
\endgroup
\ifCONF \breakcmd \else \breakcmd \fi 
= \begin{pmatrix}
	\tau_1\mdot{t} + n
	&\tau_1\tau_2\mdot{t} {+} {s_{12}} {+}{s_{123}}  & \tau_1\tau_3\mdot{t} {+} {s_{13}} {+}{s_{123}} \\
	\tau_1\tau_2\mdot{t} {+} {s_{12}} {+}{s_{123}} & 
	\tau_2\mdot{t} + n
	& \tau_2\tau_3\mdot{t} {+} {s_{23}} {+}{s_{123}}\\ 
	\tau_1\tau_3\mdot{t} {+} {s_{13}} {+}{s_{123}} &  \tau_2\tau_3\mdot{t} {+} {s_{23}} {+}{s_{123}}  & 
		\tau_3\mdot{t} + n
\end{pmatrix} \sigma^2.
\end{align*}
The second covariance matrix required is for the random vector $\vec{S}_{1,2,3} = \lp{ X_{\hat{S}_1}  , X_{\hat{S}_2}  , X_{\hat{S}_3} }^\tran$ and is given as
\begin{align*}
	\bm{\Sigma}_{\vec{S}_{1,2,3}}= \begin{pmatrix}
		n  &{s_{12}} {+}{s_{123}}  &  {s_{13}} {+}{s_{123}} \\
		{s_{12}} {+}{s_{123}} &n &   {s_{23}} {+}{s_{123}}\\ 
		{s_{13}} {+}{s_{123}} &   {s_{23}} {+}{s_{123}}  & n	   
	\end{pmatrix}\sigma^2.
\end{align*}
With these matrices, we are capable of computing the conditional entropy $h\lp{\vec{X}_T \mid O_1^{(\tau_1)}, O_2^{(\tau_2)}, O_3^{(\tau_3)}}$.
\ifCONF
Note
\else
It will be important later
\fi
that the above covariance matrices only depend on pairwise spectator overlaps between the executions $(s_{13} + s_{123})$, $(s_{12} + s_{123})$, and $(s_{23} + s_{123})$, rather than individual sets $s_{12}$, $s_{23}$, $s_{123}$, etc.

\ifCONF
Our results can be generalized to any $M$ executions, which we provide in 
 the full version of the paper.
\else

\subsection{$M$ Executions}
\label{sub:m_executions}
We can generalize  the prior section's analysis to obtain the target's conditional entropy for an arbitrary number of evaluations. Let $M$ be to the total number of evaluations where $M \in \mathbb{Z}_{>0}$. Denote $A$ as the set of integers from $1$ to $M$, such that $A = \lbr{1,\dots,M}$. We can generate the set of all subsets of spectators which overlap and do not overlap between evaluation using the power set of $A$ (denoted by $\P(A)$). Specifically
$\S = \P(A) \setminus \lbr{\emptyset}$,
the empty set is excluded as it corresponds to the target not participating in any computation. The number of spectator subsets and target participation configurations is $\norm{\S} = 2^{M} - 1$.  The output random variable of experiment $i\in \lbr{1,\dots,M}$ is therefore
\begin{align*}
	O_{i}^{(\tau_i)} = \tau_i \cdot X_T + \sum_{\stackrel{R \subseteq \S:}{i \in R}} X_{S_R} = \tau_i \cdot X_T + X_{\hat{S}_i} .
\end{align*}
We can generate elements of the $M \times M$ covariance matrix of the random vector $\vec{O}_{1,\dots,M} = \lp{ O_1^{(\tau_1)}, \dots, O_M^{(\tau_M)} }^\tran$ using the following expression for $i,j \in \lbr{1,\dots,M}$:
\begin{align*}
	\Cov{O_i^{(\tau_i)}, O_j^{(\tau_j)}} = \begin{cases}
		\tau_i\tau_j\mdot \sigma_{T}^2 +  \sum_{\stackrel{R \subseteq \S:}{(i,j) \in R}} \sigma_{R}^2 & \text{ if } i \neq j \\
		\tau_i\mdot \sigma_{T}^2 +  \sum_{\stackrel{R \subseteq \S:}{i \in R}} \sigma_{R}^2 & \text{ if } i = j 
	\end{cases}.
\end{align*}
Similarly, elements of the covariance matrix of the random vector $\vec{S}_{1,\dots,M} = \lp{ X_{\hat{S}_1}  ,\dots, X_{\hat{S}_M} }^\tran$ can be generated as follows:
\begin{align*}
	\Cov{X_{\hat{S}_i}, X_{\hat{S}_j}} = \begin{cases}
		  \sum_{\stackrel{R \subseteq \S:}{(i,j) \in R}} \sigma_{R}^2 & \text{ if } i \neq j \\
		\sum_{\stackrel{R \subseteq \S:}{i \in R}} \sigma_{R}^2 & \text{ if } i = j 
	\end{cases}.
\end{align*}
If the total number of spectators per evaluation is fixed to $n$, then $\sum_{\stackrel{R \subseteq \S:}{i \in R}} \sigma_{R}^2 = \sigma^2n$.
\fi

\subsection{Experimental Evaluation}
\label{sub:experimental_evaluation_3}
\begin{figure*}[t]
    \begin{small} \centering
			\centering
			
			\begin{subfigure}[t]{0.45\textwidth} \centering
				\includegraphics[width=0.88\textwidth]{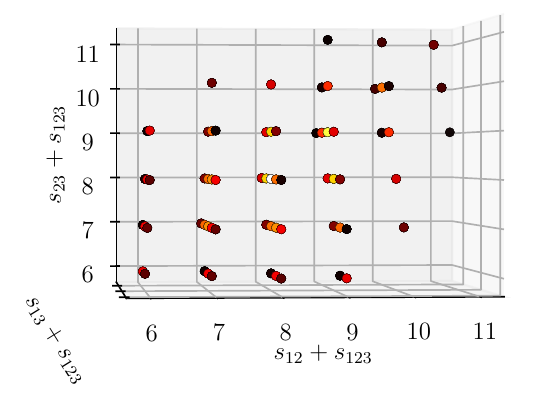}
				\caption{}
				\label{fig:3d_pairwise_1}
			\end{subfigure}
				\begin{subfigure}[t]{0.45\textwidth} \centering
					\includegraphics[width=0.625\textwidth]{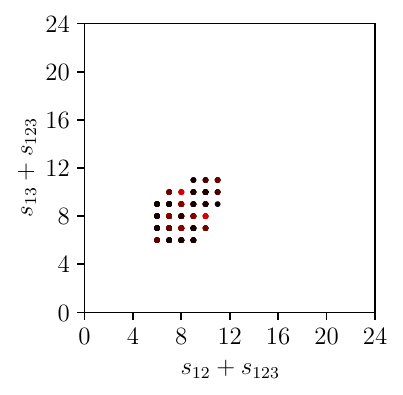}
					\caption{}
					\label{fig:3d_pariwise_projection}
				\end{subfigure}
				\begin{subfigure}[t]{0.45\textwidth} \centering
					\includegraphics[width=0.83\textwidth]{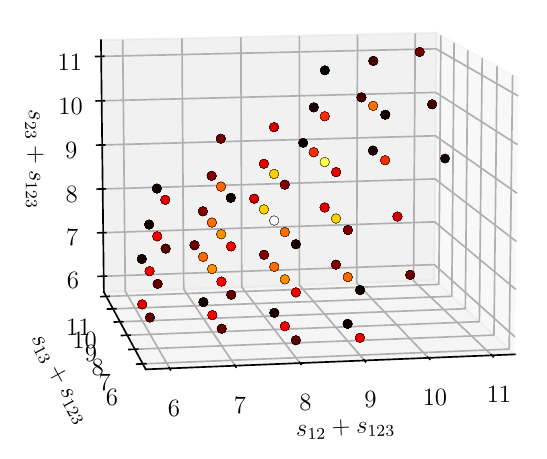}
					\caption{}
					\label{fig:3d_pairwise_3}
				\end{subfigure}
				\begin{subfigure}[t]{0.45\textwidth} \centering
					\includegraphics[width=0.94\textwidth]{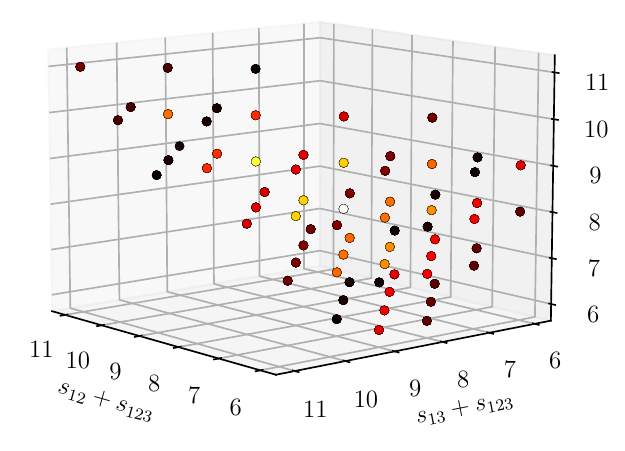}
					\caption{}
					\label{fig:3d_pairwise_4}
				\end{subfigure}
				\begin{subfigure}[t]{0.6\textwidth} \centering
					\includegraphics[width=0.9\textwidth]{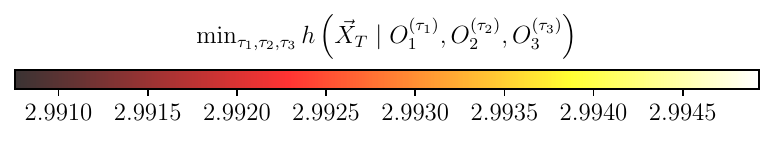}
					\label{fig:3d_pairwise_legend}
				\end{subfigure}
        \caption{Configurations and values of minimal information disclosure as functions of the pairwise spectator overlaps for three evaluations.}
        \label{fig:3d_pairwise}
    \end{small}
\end{figure*}

Unlike two executions, we can no longer graphically represent conditional entropy as a function of overlap sizes, as there are several dimensions to consider.  Instead, we enumerate all possible spectator configurations for $n = 24$ and for each spectator configuration we compute the minimum of the seven conditional entropies corresponding to valid target configurations $\tau_1,\tau_2,\tau_3$: 
\begin{align*}
	\min_{\tau_1,\tau_2,\tau_3}  h\lp{\vec{X}_T \mid O_1^{(\tau_1)}, O_2^{(\tau_2)}, O_3^{(\tau_3)}}.
\end{align*}
We then determine the maximums across all spectator configurations which correspond to the optimal choices that minimize target's information disclosure.

We plot the top 500 spectator configurations, which yield 20 unique entropy values (displayed in the color map), in Figure~\ref{fig:3d_pairwise} from different viewing angles. The axes correspond to pairwise overlap sizes, and a point with a fixed overlap, e.g., $s_{12} + s_{123}$, corresponds to different individual sizes of $s_{12}$ and $s_{123}$ that add to the same values. Recall that only the sum contributes to the entropy computation.

The maximum conditional entropy (singular white point) occurs when the pairwise overlaps are ${1}/{3}$ of $n$, i.e. when $s_{13} + s_{123} = s_{12} + s_{123} = s_{23} + s_{123} = 8$. Other top configurations are located nearby, but do not deviate from the center evenly.
In the projection of two of the three pairwise overlap dimensions ($s_{12} + s_{123} $ vs. $s_{13} + s_{123}$, Figure~\ref{fig:3d_pariwise_projection}), the top-500 configurations are concentrated in the $1/3$ overlap region. The shape is preserved (and thus the figures are identical) in the other two projections. It is important to point out that the difference in entropy between the largest and smallest value plotted is less than $1/100$th of a bit.

\begin{figure}[ht]
\centering
	\ifCONF
		\includegraphics[width=0.33\textwidth]{overlap_comparison.pdf}
		\else
		\includegraphics[width=0.4\textwidth]{overlap_comparison.pdf}
	\fi	
	  \caption{The optimal shared spectators overlap relative to the total number of participants $n$ for $M$ evaluations.}
	  \label{fig:overlap_comparison}
  \end{figure}

  \ifCONF 
  We generalize the optimal configurations to any number of experiments and spectators $n$. 
  \else 
Having examined optimal configurations for two and three executions, we want to generalize the findings to any number of experiments and spectators $n$. 
\fi
In Figure~\ref{fig:overlap_comparison}, we plot the optimal pairwise overlap percentages as a function of $n$ for 2, 3, and 4 executions. Information leakage is always the smallest when all pairwise overlaps are equal (i.e.,  for $M = 3$, $s_{13} + s_{123} = s_{12} + s_{123} = s_{23} + s_{123}$). The optimal overlap percentage for $M = 2$ is upper bounded by 50\% and tends towards 50\% as $n$ grows. Interestingly, the optimal overlap for both $M = 3$ and $M = 4$ trend toward $1/3$ overlap, while ideal overlaps are generally smaller for $M=4$. Analysis of large $M$, while potentially interesting, is of limited practical value.
  
}

\ifCONF 
\else 
\threeExec
\fi

\section{Conclusions and Recommendations}

In this work we study information disclosure associated with revealing the output of average salary computation on private inputs. Using the framework of~\cite{ah2017secure}, we analyze the function and derive several information-theoretic properties associated with the computation. Inputs are modeled using several discrete and continuous distributions, leading to multiple interesting conclusions about their entropy loss. We expand the scope to multiple executions on related inputs and determine the best configurations that minimize information disclosure. This leads to the following recommendations for computation designers: 
\begin{itemize}

\item The amount of information disclosure about a target is independent of adversarial inputs. It was also experimentally shown to be independent of distribution parameters for three different distributions and analytically shown for normal distribution. All examined distributions produce nearly identical entropy loss curves.

\item One can reduce the amount of entropy loss to a desired level by increasing the number of participants. For example, 
the computation designer can advertise at most 5\% or 1\% maximum entropy loss for the average salary application, which will require recruiting 6 or 25, respectively, non-adversarial participants when running only a single evaluation.

\item In the presence of repeated computations, information disclosure about inputs continues for both participants who stay and participants who leave. With two executions, protection is the largest with 50\% overlap in the participants, while both a small overlap and an overwhelming overlap result in undesirable information disclosure about different types of participants (i.e., those who stay vs. those who leave).

\item With more executions, pairwise overlaps sizes determine information disclosure. For 3 and 4 executions, optimal configurations have overlap sizes near 1/3 of the number of participants.

\item Information disclosure about participants' inputs can still be kept at a desirable level by enrolling enough participants and restricting percentage of reused inputs to be in a desired range. For example, with two executions and following the guidelines of the keeping the overlap near 50\%, the number of non-adversarial input contributors needs to be at least 8 to meet the target of 5\% information loss.
\end{itemize}

\begin{acks}
This work was supported in part by a Google Faculty Research Award and NSF grants 2112693 and 2213057. Any opinions, findings, and conclusions or recommendations expressed in this publication are those of the authors and do not necessarily reflect the views of the funding sources. 
\end{acks}
\bibliographystyle{ACM-Reference-Format}
\bibliography{refs}


\begin{thebibliography}{56}


\ifx \showCODEN    \undefined \def \showCODEN     #1{\unskip}     \fi
\ifx \showDOI      \undefined \def \showDOI       #1{#1}\fi
\ifx \showISBNx    \undefined \def \showISBNx     #1{\unskip}     \fi
\ifx \showISBNxiii \undefined \def \showISBNxiii  #1{\unskip}     \fi
\ifx \showISSN     \undefined \def \showISSN      #1{\unskip}     \fi
\ifx \showLCCN     \undefined \def \showLCCN      #1{\unskip}     \fi
\ifx \shownote     \undefined \def \shownote      #1{#1}          \fi
\ifx \showarticletitle \undefined \def \showarticletitle #1{#1}   \fi
\ifx \showURL      \undefined \def \showURL       {\relax}        \fi
\providecommand\bibfield[2]{#2}
\providecommand\bibinfo[2]{#2}
\providecommand\natexlab[1]{#1}
\providecommand\showeprint[2][]{arXiv:#2}

\bibitem[Ah-Fat and Huth(2017)]%
        {ah2017secure}
\bibfield{author}{\bibinfo{person}{P. Ah-Fat} {and} \bibinfo{person}{M. Huth}.} \bibinfo{year}{2017}\natexlab{}.
\newblock \showarticletitle{Secure Multi-party Computation: Information Flow of Outputs and Game Theory}. In \bibinfo{booktitle}{\emph{POST}}. \bibinfo{pages}{71--92}.
\newblock


\bibitem[Ah-Fat and Huth(2019)]%
        {ah2019optimal}
\bibfield{author}{\bibinfo{person}{P. Ah-Fat} {and} \bibinfo{person}{M. Huth}.} \bibinfo{year}{2019}\natexlab{}.
\newblock \showarticletitle{Optimal Accuracy-privacy Trade-off for Secure Computations}.
\newblock \bibinfo{journal}{\emph{IEEE Transactions on Information Theory}} \bibinfo{volume}{65}, \bibinfo{number}{5} (\bibinfo{year}{2019}), \bibinfo{pages}{3165--3182}.
\newblock


\bibitem[Ah-Fat and Huth(2020a)]%
        {ah2020protecting}
\bibfield{author}{\bibinfo{person}{P. Ah-Fat} {and} \bibinfo{person}{M. Huth}.} \bibinfo{year}{2020}\natexlab{a}.
\newblock \showarticletitle{Protecting Private Inputs: Bounded Distortion Guarantees With Randomised Approximations.}
\newblock \bibinfo{journal}{\emph{PoPETS}} \bibinfo{volume}{2020}, \bibinfo{number}{3} (\bibinfo{year}{2020}), \bibinfo{pages}{284--303}.
\newblock


\bibitem[Ah-Fat and Huth(2020b)]%
        {ah2020two}
\bibfield{author}{\bibinfo{person}{P. Ah-Fat} {and} \bibinfo{person}{M. Huth}.} \bibinfo{year}{2020}\natexlab{b}.
\newblock \bibinfo{title}{Two and Three-Party Digital Goods Auctions: Scalable Privacy Analysis}.
\newblock \bibinfo{howpublished}{arXiv preprint arXiv:2009.09524}.
\newblock


\bibitem[Alvim et~al\mbox{.}(2014a)]%
        {alvim2014additive}
\bibfield{author}{\bibinfo{person}{M. Alvim}, \bibinfo{person}{K. Chatzikokolakis}, \bibinfo{person}{A. McIver}, \bibinfo{person}{C. Morgan}, \bibinfo{person}{C. Palamidessi}, {and} \bibinfo{person}{G. Smith}.} \bibinfo{year}{2014}\natexlab{a}.
\newblock \showarticletitle{Additive and multiplicative notions of leakage, and their capacities}. In \bibinfo{booktitle}{\emph{IEEE Computer Security Foundations Symposium}}. \bibinfo{pages}{308--322}.
\newblock


\bibitem[Alvim et~al\mbox{.}(2012)]%
        {m2012measuring}
\bibfield{author}{\bibinfo{person}{M. Alvim}, \bibinfo{person}{K. Chatzikokolakis}, \bibinfo{person}{C. Palamidessi}, {and} \bibinfo{person}{G. Smith}.} \bibinfo{year}{2012}\natexlab{}.
\newblock \showarticletitle{Measuring information leakage using generalized gain functions}. In \bibinfo{booktitle}{\emph{IEEE Computer Security Foundations Symposium}}. \bibinfo{pages}{265--279}.
\newblock


\bibitem[Alvim et~al\mbox{.}(2014b)]%
        {alvim2014not}
\bibfield{author}{\bibinfo{person}{M. Alvim}, \bibinfo{person}{A. Scedrov}, {and} \bibinfo{person}{F. Schneider}.} \bibinfo{year}{2014}\natexlab{b}.
\newblock \showarticletitle{When Not All Bits Are Equal: Worth-Based Information Flow.}. In \bibinfo{booktitle}{\emph{POST}}. \bibinfo{pages}{120--139}.
\newblock


\bibitem[Barakat(1976)]%
        {barakat1976sums}
\bibfield{author}{\bibinfo{person}{R. Barakat}.} \bibinfo{year}{1976}\natexlab{}.
\newblock \showarticletitle{Sums of Independent Lognormally Distributed Random Variables}.
\newblock \bibinfo{journal}{\emph{Journal of the Optical Society of America}} \bibinfo{volume}{66}, \bibinfo{number}{3} (\bibinfo{year}{1976}), \bibinfo{pages}{211--216}.
\newblock


\bibitem[Beaulieu et~al\mbox{.}(1995)]%
        {beaulieu1995estimating}
\bibfield{author}{\bibinfo{person}{N. Beaulieu}, \bibinfo{person}{A. Abu-Dayya}, {and} \bibinfo{person}{P. McLane}.} \bibinfo{year}{1995}\natexlab{}.
\newblock \showarticletitle{Estimating the Distribution of a Sum of Independent Lognormal Random Variables}.
\newblock \bibinfo{journal}{\emph{IEEE Transactions on Communications}} \bibinfo{volume}{43}, \bibinfo{number}{12} (\bibinfo{year}{1995}), \bibinfo{pages}{2869--2873}.
\newblock


\bibitem[Beaulieu and Xie(2004)]%
        {beaulieu2004optimal}
\bibfield{author}{\bibinfo{person}{N. Beaulieu} {and} \bibinfo{person}{Q. Xie}.} \bibinfo{year}{2004}\natexlab{}.
\newblock \showarticletitle{An Optimal Lognormal Approximation to Lognormal Sum Distributions}.
\newblock \bibinfo{journal}{\emph{IEEE Transactions on Vehicular Technology}} \bibinfo{volume}{53}, \bibinfo{number}{2} (\bibinfo{year}{2004}), \bibinfo{pages}{479--489}.
\newblock


\bibitem[Bhowmick et~al\mbox{.}(2021)]%
        {bhowmick2021apple}
\bibfield{author}{\bibinfo{person}{A. Bhowmick}, \bibinfo{person}{D. Boneh}, \bibinfo{person}{S. Myers}, {and} \bibinfo{person}{K. Tarbe}.} \bibinfo{year}{2021}\natexlab{}.
\newblock \bibinfo{title}{The {Apple} {PSI} system}.
\newblock \bibinfo{howpublished}{\url{https://www.apple.com/child-safety/pdf/Apple_PSI_System_Security_Protocol_and_Analysis.pdf}}.
\newblock


\bibitem[Boston Women's Workforce Council ({BWWC})(2017)]%
        {bwwc2016}
Boston Women's Workforce Council ({BWWC}) \bibinfo{year}{2017}\natexlab{}.
\newblock \bibinfo{title}{2016 Report}.
\newblock \bibinfo{howpublished}{\url{https://htv-prod-media.s3.amazonaws.com/files/bwwc-report-final-january-4-2017-1483635889.pdf}}.
\newblock


\bibitem[Boston Women's Workforce Council ({BWWC})(2018)]%
        {bwwc2017}
Boston Women's Workforce Council ({BWWC}) \bibinfo{year}{2018}\natexlab{}.
\newblock \bibinfo{title}{2017 Report}.
\newblock \bibinfo{howpublished}{\url{https://www.boston.gov/sites/default/files/document-file-01-2018/bwwc_2017_report.pdf}}.
\newblock


\bibitem[Bu et~al\mbox{.}(2006)]%
        {bu2006preservation}
\bibfield{author}{\bibinfo{person}{S. Bu}, \bibinfo{person}{L. Lakshmanan}, \bibinfo{person}{R. Ng}, {and} \bibinfo{person}{G. Ramesh}.} \bibinfo{year}{2006}\natexlab{}.
\newblock \showarticletitle{Preservation of patterns and input-output privacy}. In \bibinfo{booktitle}{\emph{IEEE International Conference on Data Engineering}}. \bibinfo{pages}{696--705}.
\newblock


\bibitem[Caiado and Rathie(2007)]%
        {caiado2007polynomial}
\bibfield{author}{\bibinfo{person}{C. Caiado} {and} \bibinfo{person}{P. Rathie}.} \bibinfo{year}{2007}\natexlab{}.
\newblock \showarticletitle{Polynomial Coefficients and Distribution of the Sum of Discrete Uniform Variables}. In \bibinfo{booktitle}{\emph{SSFA}}.
\newblock


\bibitem[Cao et~al\mbox{.}(2022)]%
        {cao2022priori}
\bibfield{author}{\bibinfo{person}{L. Cao}, \bibinfo{person}{T. Tong}, \bibinfo{person}{D. Trafimow}, \bibinfo{person}{T. Wang}, {and} \bibinfo{person}{X. Chen}.} \bibinfo{year}{2022}\natexlab{}.
\newblock \showarticletitle{The A Priori Procedure for estimating the mean in both log-normal and gamma populations and robustness for assumption violations}.
\newblock \bibinfo{journal}{\emph{Methodology}} \bibinfo{volume}{18}, \bibinfo{number}{1} (\bibinfo{year}{2022}), \bibinfo{pages}{24--43}.
\newblock


\bibitem[Cheraghchi(2019)]%
        {cheraghchi2019expressions}
\bibfield{author}{\bibinfo{person}{M. Cheraghchi}.} \bibinfo{year}{2019}\natexlab{}.
\newblock \showarticletitle{Expressions for the Entropy of Basic Discrete Distributions}.
\newblock \bibinfo{journal}{\emph{IEEE Transactions on Information Theory}} \bibinfo{volume}{65}, \bibinfo{number}{7} (\bibinfo{year}{2019}), \bibinfo{pages}{3999--4009}.
\newblock


\bibitem[Clark et~al\mbox{.}(2002)]%
        {clark2002quantitative}
\bibfield{author}{\bibinfo{person}{D. Clark}, \bibinfo{person}{S. Hunt}, {and} \bibinfo{person}{P. Malacaria}.} \bibinfo{year}{2002}\natexlab{}.
\newblock \showarticletitle{Quantitative analysis of the leakage of confidential data}.
\newblock \bibinfo{journal}{\emph{Electronic Notes in Theoretical Computer Science}} \bibinfo{volume}{59}, \bibinfo{number}{3} (\bibinfo{year}{2002}), \bibinfo{pages}{238--251}.
\newblock


\bibitem[Clementi and Gallegati(2005)]%
        {clementi2005pareto}
\bibfield{author}{\bibinfo{person}{F. Clementi} {and} \bibinfo{person}{M. Gallegati}.} \bibinfo{year}{2005}\natexlab{}.
\newblock \showarticletitle{{Pareto}'s law of income distribution: Evidence for {Germany}, the {United Kingdom}, and the {United States}}.
\newblock In \bibinfo{booktitle}{\emph{Econophysics of Wealth Distributions}}. \bibinfo{publisher}{Springer}, \bibinfo{pages}{3--14}.
\newblock
\urldef\tempurl%
\url{https://doi.org/10.1007/88-470-0389-X_1}
\showDOI{\tempurl}


\bibitem[Cobb et~al\mbox{.}(2012)]%
        {cobb2012approximating}
\bibfield{author}{\bibinfo{person}{B. Cobb}, \bibinfo{person}{R. Rum{\'\i}}, {and} \bibinfo{person}{A. Salmer{\'o}n}.} \bibinfo{year}{2012}\natexlab{}.
\newblock \showarticletitle{Approximating the Distribution of a Sum of Log-normal Random Variables}.
\newblock \bibinfo{journal}{\emph{Statistics and Computing}} \bibinfo{volume}{16}, \bibinfo{number}{3} (\bibinfo{year}{2012}), \bibinfo{pages}{293--308}.
\newblock


\bibitem[Cover and Thomas(2006)]%
        {thomas2006elements}
\bibfield{author}{\bibinfo{person}{T. Cover} {and} \bibinfo{person}{J. Thomas}.} \bibinfo{year}{2006}\natexlab{}.
\newblock \bibinfo{booktitle}{\emph{Elements of Information Theory}}.
\newblock \bibinfo{publisher}{Wiley-Interscience}.
\newblock


\bibitem[Denning(1982)]%
        {denning1982cryptography}
\bibfield{author}{\bibinfo{person}{D. Denning}.} \bibinfo{year}{1982}\natexlab{}.
\newblock \bibinfo{booktitle}{\emph{Cryptography and data security}}.
\newblock \bibinfo{publisher}{Addison-Wesley Reading}.
\newblock


\bibitem[Deshpande et~al\mbox{.}(2011)]%
        {deshpande2011}
\bibfield{author}{\bibinfo{person}{V. Deshpande}, \bibinfo{person}{L. Schwarz}, \bibinfo{person}{M. Atallah}, \bibinfo{person}{M. Blanton}, {and} \bibinfo{person}{K. Frikken}.} \bibinfo{year}{2011}\natexlab{}.
\newblock \showarticletitle{Outsourcing manufacturing: Secure price-masking mechanisms for purchasing component parts}.
\newblock \bibinfo{journal}{\emph{Production and Operations Management}} \bibinfo{volume}{20}, \bibinfo{number}{2} (\bibinfo{year}{2011}), \bibinfo{pages}{165--180}.
\newblock


\bibitem[Deshpande et~al\mbox{.}(2005)]%
        {deshpande2005secure_tr}
\bibfield{author}{\bibinfo{person}{V. Deshpande}, \bibinfo{person}{L. Schwarz}, \bibinfo{person}{M. Atallah}, \bibinfo{person}{M. Blanton}, \bibinfo{person}{K. Frikken}, {and} \bibinfo{person}{J. Li}.} \bibinfo{year}{2005}\natexlab{}.
\newblock \bibinfo{title}{Secure Collaborative Planning, Forecasting and Replenishment ({SCPFR})}.
\newblock \bibinfo{howpublished}{CERIAS Tech Report 2006-65}.
\newblock


\bibitem[Deshpande et~al\mbox{.}(2006)]%
        {deshpande2005secure}
\bibfield{author}{\bibinfo{person}{V. Deshpande}, \bibinfo{person}{L. Schwarz}, \bibinfo{person}{M. Atallah}, \bibinfo{person}{M. Blanton}, \bibinfo{person}{K. Frikken}, {and} \bibinfo{person}{J. Li}.} \bibinfo{year}{2006}\natexlab{}.
\newblock \showarticletitle{Secure Collaborative Planning, Forecasting and Replenishment ({SCPFR})}. In \bibinfo{booktitle}{\emph{Multi-Echelon/Public Applications of Supply Chain Management Conference}}. \bibinfo{pages}{165--180}.
\newblock


\bibitem[Dwork(2008)]%
        {dwork2008differential}
\bibfield{author}{\bibinfo{person}{C. Dwork}.} \bibinfo{year}{2008}\natexlab{}.
\newblock \showarticletitle{Differential privacy: A survey of results}. In \bibinfo{booktitle}{\emph{International Conference on Theory and Applications of Models of Computation}}. \bibinfo{pages}{1--19}.
\newblock


\bibitem[Dwork and Roth(2014)]%
        {dwork2014algorithmic}
\bibfield{author}{\bibinfo{person}{C. Dwork} {and} \bibinfo{person}{A. Roth}.} \bibinfo{year}{2014}\natexlab{}.
\newblock \showarticletitle{The algorithmic foundations of differential privacy}.
\newblock \bibinfo{journal}{\emph{Foundations and Trends in Theoretical Computer Science}} \bibinfo{volume}{9}, \bibinfo{number}{3--4} (\bibinfo{year}{2014}), \bibinfo{pages}{211--407}.
\newblock


\bibitem[Evans and Boersma(1988)]%
        {evans1988entropy}
\bibfield{author}{\bibinfo{person}{R. Evans} {and} \bibinfo{person}{J. Boersma}.} \bibinfo{year}{1988}\natexlab{}.
\newblock \showarticletitle{The Entropy of a {Poisson} Distribution ({C. Robert Appledorn})}.
\newblock \bibinfo{journal}{\emph{SIAM Rev.}} \bibinfo{volume}{30}, \bibinfo{number}{2} (\bibinfo{year}{1988}), \bibinfo{pages}{314--317}.
\newblock
\urldef\tempurl%
\url{https://doi.org/10.1137/1030059}
\showDOI{\tempurl}


\bibitem[Fenton(1960)]%
        {fenton1960sum}
\bibfield{author}{\bibinfo{person}{L. Fenton}.} \bibinfo{year}{1960}\natexlab{}.
\newblock \showarticletitle{The sum of log-normal probability distributions in scatter transmission systems}.
\newblock \bibinfo{journal}{\emph{IRE Transactions on Communications Systems}} \bibinfo{volume}{8}, \bibinfo{number}{1} (\bibinfo{year}{1960}), \bibinfo{pages}{57--67}.
\newblock


\bibitem[Inpher(2024)]%
        {inpher}
Inpher \bibinfo{year}{2024}\natexlab{}.
\newblock \bibinfo{title}{Inpher}.
\newblock
\newblock
\newblock
\shownote{\url{https://inpher.io/}}.


\bibitem[Ion et~al\mbox{.}(2020)]%
        {ion2020}
\bibfield{author}{\bibinfo{person}{M. Ion}, \bibinfo{person}{B. Kreuter}, \bibinfo{person}{A. Nergiz}, \bibinfo{person}{S. Patel}, \bibinfo{person}{S. Saxena}, \bibinfo{person}{K. Seth}, \bibinfo{person}{M. Raykova}, \bibinfo{person}{D. Shanahan}, {and} \bibinfo{person}{M. Yung}.} \bibinfo{year}{2020}\natexlab{}.
\newblock \showarticletitle{On deploying secure computing: Private intersection-sum-with-cardinality}. In \bibinfo{booktitle}{\emph{IEEE EuroS\&P}}. \bibinfo{pages}{370--389}.
\newblock


\bibitem[Iwamoto and Shikata(2013)]%
        {iwamoto2013information}
\bibfield{author}{\bibinfo{person}{M. Iwamoto} {and} \bibinfo{person}{J. Shikata}.} \bibinfo{year}{2013}\natexlab{}.
\newblock \showarticletitle{Information theoretic security for encryption based on conditional {R{\'e}nyi} entropies}. In \bibinfo{booktitle}{\emph{International Conference on Information Theoretic Security}}. \bibinfo{pages}{103--121}.
\newblock


\bibitem[K{\"o}pf and Basin(2011)]%
        {kopf2011automatically}
\bibfield{author}{\bibinfo{person}{B. K{\"o}pf} {and} \bibinfo{person}{D. Basin}.} \bibinfo{year}{2011}\natexlab{}.
\newblock \showarticletitle{Automatically deriving information-theoretic bounds for adaptive side-channel attacks}.
\newblock \bibinfo{journal}{\emph{Journal of Computer Security}} \bibinfo{volume}{19}, \bibinfo{number}{1} (\bibinfo{year}{2011}), \bibinfo{pages}{1--31}.
\newblock


\bibitem[Kotecha and Garg(2017)]%
        {kotecha2017preserving}
\bibfield{author}{\bibinfo{person}{R. Kotecha} {and} \bibinfo{person}{S. Garg}.} \bibinfo{year}{2017}\natexlab{}.
\newblock \showarticletitle{Preserving output-privacy in data stream classification}.
\newblock \bibinfo{journal}{\emph{Progress in Artificial Intelligence}}  \bibinfo{volume}{6} (\bibinfo{year}{2017}), \bibinfo{pages}{87--104}.
\newblock


\bibitem[Kreuter(2017)]%
        {google-ads}
\bibfield{author}{\bibinfo{person}{B. Kreuter}.} \bibinfo{year}{2017}\natexlab{}.
\newblock \bibinfo{title}{Secure Multiparty Computation at {G}oogle}.
\newblock \bibinfo{howpublished}{Real World Crypto}.
\newblock
\newblock
\shownote{Available from https://www.youtube.com/watch?v=ee7oRsDnNNc}.


\bibitem[Lapets et~al\mbox{.}(2018)]%
        {lap18}
\bibfield{author}{\bibinfo{person}{A. Lapets}, \bibinfo{person}{F. Jansen}, \bibinfo{person}{K. Albab}, \bibinfo{person}{R. Issa}, \bibinfo{person}{L. Qin}, \bibinfo{person}{M. Varia}, {and} \bibinfo{person}{A. Bestavros}.} \bibinfo{year}{2018}\natexlab{}.
\newblock \showarticletitle{Accessible Privacy-Preserving Web-Based Data Analysis for Assessing and Addressing Economic Inequalities}. In \bibinfo{booktitle}{\emph{ACM COMPASS}}. \bibinfo{pages}{48:1--48:5}.
\newblock


\bibitem[Lapets et~al\mbox{.}(2016a)]%
        {lap16}
\bibfield{author}{\bibinfo{person}{A. Lapets}, \bibinfo{person}{N. Volgushev}, \bibinfo{person}{A. Bestavros}, \bibinfo{person}{F. Jansen}, {and} \bibinfo{person}{M. Varia}.} \bibinfo{year}{2016}\natexlab{a}.
\newblock \showarticletitle{Secure {MPC} for Analytics as a Web Application}. In \bibinfo{booktitle}{\emph{SecDev}}. \bibinfo{pages}{73--74}.
\newblock


\bibitem[Lapets et~al\mbox{.}(2016b)]%
        {lap16b}
\bibfield{author}{\bibinfo{person}{A. Lapets}, \bibinfo{person}{N. Volgushev}, \bibinfo{person}{A. Bestavros}, \bibinfo{person}{F. Jansen}, {and} \bibinfo{person}{M. Varia}.} \bibinfo{year}{2016}\natexlab{b}.
\newblock \bibinfo{booktitle}{\emph{Secure Multi-Party Computation for Analytics Deployed as a Lightweight Web Application}}.
\newblock \bibinfo{type}{{T}echnical {R}eport} BUCS-TR-2016-008. \bibinfo{institution}{Boston University}.
\newblock


\bibitem[Ligero Inc.(2022)]%
        {ligero}
Ligero Inc. \bibinfo{year}{2022}\natexlab{}.
\newblock \bibinfo{title}{Secure and Private Collaboration for Blockchains and Beyond}.
\newblock \bibinfo{howpublished}{\url{https://ligero-inc.com/}}.
\newblock
\newblock
\shownote{Last accessed: 2022-08-16}.


\bibitem[Mardziel et~al\mbox{.}(2012)]%
        {mardziel2012knowledge}
\bibfield{author}{\bibinfo{person}{P. Mardziel}, \bibinfo{person}{M. Hicks}, \bibinfo{person}{J. Katz}, {and} \bibinfo{person}{M. Srivatsa}.} \bibinfo{year}{2012}\natexlab{}.
\newblock \showarticletitle{Knowledge-oriented secure multiparty computation}. In \bibinfo{booktitle}{\emph{Workshop on Programming Languages and Analysis for Security}}. \bibinfo{pages}{1--12}.
\newblock


\bibitem[Massey(1994)]%
        {massey1994guessing}
\bibfield{author}{\bibinfo{person}{J. Massey}.} \bibinfo{year}{1994}\natexlab{}.
\newblock \showarticletitle{Guessing and entropy}. In \bibinfo{booktitle}{\emph{IEEE International Symposium on Information Theory}}. \bibinfo{pages}{204}.
\newblock


\bibitem[Mendes and Vilela(2017)]%
        {mendes2017privacy}
\bibfield{author}{\bibinfo{person}{R. Mendes} {and} \bibinfo{person}{J. Vilela}.} \bibinfo{year}{2017}\natexlab{}.
\newblock \showarticletitle{Privacy-preserving data mining: methods, metrics, and applications}.
\newblock \bibinfo{journal}{\emph{IEEE Access}}  \bibinfo{volume}{5} (\bibinfo{year}{2017}), \bibinfo{pages}{10562--10582}.
\newblock


\bibitem[Monreale and Wang(2016)]%
        {monreale2016privacy}
\bibfield{author}{\bibinfo{person}{A. Monreale} {and} \bibinfo{person}{W. Wang}.} \bibinfo{year}{2016}\natexlab{}.
\newblock \showarticletitle{Privacy-preserving outsourcing of data mining}. In \bibinfo{booktitle}{\emph{IEEE COMPSAC}}, Vol.~\bibinfo{volume}{2}. \bibinfo{pages}{583--588}.
\newblock


\bibitem[Nth party(2024)]%
        {nthparty}
Nth party \bibinfo{year}{2024}\natexlab{}.
\newblock \bibinfo{title}{Nth party}.
\newblock
\newblock
\newblock
\shownote{\url{https://www.nthparty.com/}}.


\bibitem[Partisia(2024)]%
        {partisia}
Partisia \bibinfo{year}{2024}\natexlab{}.
\newblock \bibinfo{title}{Partisia}.
\newblock
\newblock
\newblock
\shownote{\url{https://partisia.com/}}.


\bibitem[Rastogi et~al\mbox{.}(2013)]%
        {rastogi2013knowledge}
\bibfield{author}{\bibinfo{person}{A. Rastogi}, \bibinfo{person}{P. Mardziel}, \bibinfo{person}{M. Hicks}, {and} \bibinfo{person}{M. Hammer}.} \bibinfo{year}{2013}\natexlab{}.
\newblock \showarticletitle{Knowledge inference for optimizing secure multi-party computation}. In \bibinfo{booktitle}{\emph{SIGPLAN Workshop on Programming Languages and Analysis for Security}}. \bibinfo{pages}{3--14}.
\newblock


\bibitem[Schwartz and Yeh(1982)]%
        {schwartz1982distribution}
\bibfield{author}{\bibinfo{person}{S. Schwartz} {and} \bibinfo{person}{Y. Yeh}.} \bibinfo{year}{1982}\natexlab{}.
\newblock \showarticletitle{On the distribution function and moments of power sums with log-normal components}.
\newblock \bibinfo{journal}{\emph{Bell System Technical Journal}} \bibinfo{volume}{61}, \bibinfo{number}{7} (\bibinfo{year}{1982}), \bibinfo{pages}{1441--1462}.
\newblock


\bibitem[Senaratne and Tellambura(2009)]%
        {senaratne2009numerical}
\bibfield{author}{\bibinfo{person}{D. Senaratne} {and} \bibinfo{person}{C. Tellambura}.} \bibinfo{year}{2009}\natexlab{}.
\newblock \showarticletitle{Numerical Computation of the Lognormal Sum Distribution}. In \bibinfo{booktitle}{\emph{IEEE GLOBECOM}}. \bibinfo{pages}{1--6}.
\newblock


\bibitem[Shokri et~al\mbox{.}(2017)]%
        {shokri2017membership}
\bibfield{author}{\bibinfo{person}{R. Shokri}, \bibinfo{person}{M. Stronati}, \bibinfo{person}{C. Song}, {and} \bibinfo{person}{V. Shmatikov}.} \bibinfo{year}{2017}\natexlab{}.
\newblock \showarticletitle{Membership Inference Attacks Against Machine Learning Models}. In \bibinfo{booktitle}{\emph{IEEE S\&P}}. \bibinfo{pages}{3--18}.
\newblock


\bibitem[Sk{\'o}rski(2019)]%
        {skorski2019strong}
\bibfield{author}{\bibinfo{person}{M. Sk{\'o}rski}.} \bibinfo{year}{2019}\natexlab{}.
\newblock \showarticletitle{Strong chain rules for min-entropy under few bits spoiled}. In \bibinfo{booktitle}{\emph{IEEE International Symposium on Information Theory}}. \bibinfo{pages}{1122--1126}.
\newblock


\bibitem[Smith(2009)]%
        {smith2009foundations}
\bibfield{author}{\bibinfo{person}{G. Smith}.} \bibinfo{year}{2009}\natexlab{}.
\newblock \showarticletitle{On the foundations of quantitative information flow}. In \bibinfo{booktitle}{\emph{FoSSaCS}}. \bibinfo{pages}{288--302}.
\newblock


\bibitem[Song and Mittal(2021)]%
        {song2021systematic}
\bibfield{author}{\bibinfo{person}{L. Song} {and} \bibinfo{person}{P. Mittal}.} \bibinfo{year}{2021}\natexlab{}.
\newblock \showarticletitle{Systematic Evaluation of Privacy Risks of Machine Learning Models}. In \bibinfo{booktitle}{\emph{USENIX Security Symposium}}. \bibinfo{pages}{2615--2632}.
\newblock


\bibitem[Souma(2002)]%
        {souma2002physics}
\bibfield{author}{\bibinfo{person}{W. Souma}.} \bibinfo{year}{2002}\natexlab{}.
\newblock \showarticletitle{Physics of personal income}. In \bibinfo{booktitle}{\emph{Empirical Science of Financial Fluctuations}}. \bibinfo{pages}{343--352}.
\newblock


\bibitem[Walker et~al\mbox{.}(2019)]%
        {walker2019}
\bibfield{author}{\bibinfo{person}{A. Walker}, \bibinfo{person}{S. Patel}, {and} \bibinfo{person}{M. Yung}.} \bibinfo{year}{2019}\natexlab{}.
\newblock \showarticletitle{Helping organizations do more without collecting more data}.
\newblock \bibinfo{journal}{\emph{Google Security Blog}} (\bibinfo{date}{jun} \bibinfo{year}{2019}).
\newblock
\urldef\tempurl%
\url{https://security.googleblog.com/2019/06/helping-organizations-do-more-without-collecting-more-data.html}
\showURL{%
\tempurl}
\newblock
\shownote{Last accessed: 2022-08-16}.


\bibitem[Wang and Liu(2011)]%
        {wang2011output}
\bibfield{author}{\bibinfo{person}{T. Wang} {and} \bibinfo{person}{L. Liu}.} \bibinfo{year}{2011}\natexlab{}.
\newblock \showarticletitle{Output privacy in data mining}.
\newblock \bibinfo{journal}{\emph{ACM Transactions on Database Systems (TODS)}} \bibinfo{volume}{36}, \bibinfo{number}{1} (\bibinfo{year}{2011}), \bibinfo{pages}{1--34}.
\newblock


\bibitem[Wu et~al\mbox{.}(2005)]%
        {wu2005flexible}
\bibfield{author}{\bibinfo{person}{J. Wu}, \bibinfo{person}{N. Mehta}, {and} \bibinfo{person}{J. Zhang}.} \bibinfo{year}{2005}\natexlab{}.
\newblock \showarticletitle{Flexible Lognormal Sum Approximation Method}. In \bibinfo{booktitle}{\emph{IEEE GLOBECOM}}. \bibinfo{pages}{3413--3417}.
\newblock


\end{thebibliography}

\end{document}